\documentclass[12pt,letterpaper,twoside,english,draftclsnofoot, onecolumn]{IEEEtran}
\usepackage[T1]{fontenc}
\usepackage[latin9]{inputenc}
\usepackage{array}
\usepackage{amsmath}
\usepackage{graphicx}
\usepackage{amssymb}

\providecommand{\tabularnewline}{\\}
\usepackage{color, graphicx}
\usepackage{cite, bbm}
\usepackage{epsfig, psfrag, subfigure, amssymb}

\DeclareFontFamily{OT1}{pzc}{}
\DeclareFontShape{OT1}{pzc}{m}{it}{<-> s * [1.35] pzcmi7t}{}
\DeclareMathAlphabet{\mathpzc}{OT1}{pzc}{m}{it}

\newtheorem{definition}{Definition}[section]
\newtheorem{theorem}{Theorem}[section]
\newtheorem{corollary}{Corollary}[section]

\newtheorem{property}{Property}[section]

\usepackage{babel}

\begin{document}

\title{Secure Communication in\\
Stochastic Wireless Networks}

\author{Pedro C.\ Pinto, \emph{Student Member, IEEE}, João Barros,\emph{
Member, IEEE},\\
and Moe Z.\ Win, \emph{Fellow, IEEE}%
\thanks{The present document is a draft submitted for publication on November
24, 2009.%
}%
\thanks{P.~C.~Pinto and M.~Z.~Win are with the Laboratory for Information
and Decision Systems (LIDS), Massachusetts Institute of Technology,
Room~\mbox{32-D674}, 77~Massachusetts Avenue, Cambridge, MA 02139,
USA (\mbox{e-mail}: \texttt{ppinto@mit.edu}, \texttt{moewin@mit.edu}).
J.~Barros is with Departamento de Engenharia Electrotécnica e de
Computadores, Faculdade de Engenharia da Universidade do Porto, Portugal
(\mbox{e-mail}: \texttt{jbarros@fe.up.pt}).%
}%
\thanks{This research was supported, in part, by the Portuguese Science and
Technology Foundation under grant SFRH-BD-17388-2004; the MIT Institute
for Soldier Nanotechnologies; the Office of Naval Research under Presidential
Early Career Award for Scientists and Engineers (PECASE) N00014-09-1-0435;
and the National Science Foundation under grant ECS-0636519.%
}\underbar{}\\
\underbar{}\\
\underbar{}\\
\underbar{}\\
\underbar{Corresponding Address:}\\
Pedro C. Pinto\\
Laboratory for Information and Decision Systems (LIDS)\\
Massachusetts Institute of Technology (MIT)\\
77 Massachusetts Avenue, Room 32-D674\\
Cambridge, MA 02139 USA\\
~\\
Tel.: (857) 928-6444\\
e-mail: \texttt{ppinto@mit.edu}}

\maketitle
%\markboth{IEEE Transactions on Wireless Communications, Vol. XX, No. Y, Month 2006}{Pinto and Win: Communication in a Poisson Field of Interferers}

\thispagestyle{empty}

\newpage

\setcounter{page}{1}

\global\long\def\M{\ell}

\global\long\def\E{\textrm{e}}

\global\long\def\WM{\sigma_{\M}^{2}}

\global\long\def\WE{\sigma_{\E}^{2}}

\global\long\def\W{\sigma^{2}}

\global\long\def\PiM{\Pi_{\M}}

\global\long\def\PiE{\Pi_{\E}}

\global\long\def\LM{\lambda_{\M}}

\global\long\def\LE{\lambda_{\E}}

\global\long\def\PM{P_{\M}}

\global\long\def\PE{P_{\textrm{rx},\E}}

\global\long\def\PET{\widetilde{P}_{\textrm{rx},\E}}

\global\long\def\rM{r_{\M}}

\global\long\def\rate{\mathpzc{R}}

\global\long\def\CM{\rate_{\M}}

\global\long\def\CE{\rate_{\E}}

\global\long\def\Cs{\rate_{\mathrm{s}}}

\global\long\def\cs#1{\rate_{\mathrm{s},#1}}

\global\long\def\cM#1{\rate_{\M,#1}}

\global\long\def\RM#1{R_{\M,#1}}

\global\long\def\RE#1{R_{\E,#1}}

\global\long\def\zM{Z_{\M}}

\global\long\def\LambdaM{\Lambda_{\M}}

\global\long\def\LambdaE{\Lambda_{\E}}

\global\long\def\ZM#1{Z_{\M,#1}}

\global\long\def\ZE#1{Z_{\E,#1}}

\global\long\def\lM#1{L_{\M,#1}}

\global\long\def\lE#1{L_{\E,#1}}

\global\long\def\An{\mathcal{D}}

\global\long\def\th{\varrho}

\global\long\def\ff{\psi}

\global\long\def\is{i\mathcal{S}\textrm{-}}

\newpage{}

\tableofcontents{}

\newpage{}
\begin{abstract}
Information-theoretic security -- widely accepted as the strictest
notion of security -- relies on channel coding techniques that exploit
the inherent randomness of the propagation channels to significantly
strengthen the\emph{ }security of digital communications systems.
Motivated by recent developments in the field, this paper aims at
a characterization of the fundamental secrecy  limits of wireless
networks. Based on a general model in which legitimate nodes and potential
eavesdroppers are randomly scattered in space, the \emph{intrinsically
secure communications graph} ($\is$graph) is defined from the point
of view of information-theoretic security. Conclusive results are
provided for the local connectivity of the Poisson $\is$graph, in
terms of node degrees and isolation probabilities. It is shown how
the secure connectivity of the network varies with the wireless propagation
effects, \textit{\emph{the secrecy rate threshold of each link, and
the noise powers of legitimate nodes and eavesdroppers. }}Sectorized
transmission and eavesdropper neutralization are explored as viable
strategies for improving the secure connectivity. Lastly, the maximum
secrecy rate between a node and each of its neighbours is characterized,
and the case of colluding eavesdroppers is studied. The results help
clarify how the spatial density of eavesdroppers can compromise the
intrinsic security of wireless networks.\end{abstract}

\section{Introduction\label{sec:Introduction}}

Contemporary security systems for wireless networks are based on cryptographic
primitives that generally ignore two key factors: (a)~the physical
properties of the wireless medium, and (b)~the spatial configuration
of both the legitimate and malicious nodes. These two factors are
important since they affect the communication channels between the
nodes, which in turn determine the fundamental secrecy limits of a
wireless network. In fact, the inherent randomness of the physics
of the wireless medium and the spatial location of the nodes can be
leveraged to provide \emph{intrinsic security} of the communications
infrastructure at the physical-layer level.%
\footnote{In the literature, the term {}``security'' typically encompasses
3~different characteristics: \emph{secrecy} (or privacy), \emph{integrity},
and \emph{authenticity}. This paper does not consider the issues of
integrity or authenticity, and the terms {}``secrecy and {}``security''
are used interchangeably.%
}

The basis for information-theoretic security, which builds on the
notion of perfect secrecy~\cite{Sha:49}, was laid in \cite{Wyn:75}
and later in \cite{CsiKor:78}. Moreover, almost at the same time,
the basic principles of public-key cryptography, which lead to the
predominance of computational security, were published in \cite{DifHel:76}.
More recently, there has been a renewed interest in information-theoretic
security over wireless channels. Space-time signal processing techniques
for secure communication over wireless links are introduced in \cite{Her:03}.
The secrecy of cooperative relay broadcast channels is considered
in \cite{EkrUlu:08}. The case of a fixed number of colluding eavesdroppers
placed at the same location is analyzed in \cite{GoeNeg:05}. The
scenario of compound wiretap channels is considered in \cite{LiaKraPooSha:08}.
The capacity of cognitive interference channels with secrecy constraints
is analyzed in \cite{LiaSomPooShaVer:09}. The achievable secret communication
rates using multiple-input multiple-output communications are investigated
in \cite{EkrUlu:09,NegGoe:05,LiuSha:09,WeiLiuShaSteVis:09,ZhaZhaLiaXinCui:09}.
The secrecy capacity of various degraded fading channels is established
in \cite{ParBla:05}. A detailed characterization of the outage secrecy
capacity of slow fading channels is provided in \cite{BloBarRodMcl:08}.
The ergodic secrecy capacity of fading channels was derived independently
in \cite{LiaPooSha:08,LiYatTra:06,GopLaiGam:06}. The notion of strong
secrecy for wireless channels is introduced in \cite{BarBlo:08}.
Some secrecy properties of random geometric graphs were presented
in \cite{Hae:08}.

We are interested in the fundamental secrecy limits of large-scale
wireless networks. The spatial location of the nodes can be modeled
either deterministically or stochastically. Deterministic models include
square, triangular, and hexagonal lattices in the two-dimensional
plane~\cite{SilKle:83,MatMat:95,FerTon:04}, which are applicable
when the position of the nodes in the network is known exactly or
is constrained to a regular structure. In contrast, in many important
scenarios, only a statistical description of the node positions is
available, and thus a stochastic spatial model is the natural choice.
In particular, the Poisson point process~\cite{Kin:93} has been
successfully used in the context of wireless networks, most notably
in what concerns connectivity and coverage~\cite{BetHar:05,MioAlt:05,OrrBar:03},
throughput~\cite{Dar:96,ConDar:04}, interference~\cite{WinPinShe:J09,SalZan:07,SalZan:08,ChiGio:09},
environmental monitoring~\cite{DarConBurVer:07}, and sensor cooperation~\cite{QueDarWin:J07},
among other topics.

In this paper, we aim at a mathematical characterization of the secrecy
properties of stochastic wireless networks. The main contributions
are as follows:
\begin{itemize}
\item \textit{Framework for intrinsic security in stochastic networks:}
We introduce an information-theoretic definition of the intrinsically
secure communications graph ($\is$graph), based on the notion of
strong secrecy. Our framework considers spatially scattered users
and eavesdroppers, subject to generic wireless propagation characteristics.
\item \textit{Local connectivity in the $\is$graph:} We provide a complete
probabilistic characterization of both in-degree and out-degree of
a typical node in the Poisson $\is$graph, using fundamental tools
of stochastic geometry.
\item \emph{Techniques for communication with enhanced secrecy: }We proposed
sectorized transmission and eavesdropper neutralization as two techniques
for enhancing the secrecy of communication, and quantify their effectiveness
in terms of the resulting average node degrees.
\item \emph{Maximum secrecy rate (MSR) in the $\is$graph:} We provide a
complete probabilistic characterization of the MSR between a typical
node of the Poisson $\is$graph and each of its neighbors. In addition,
we derive expressions for the probability of existence of a non-zero
MSR, and the probability of secrecy outage\textit{\emph{.}}
\item \emph{The case of colluding eavesdroppers:} We provide a characterization
of the MSR and average node degrees for scenarios in which the eavesdroppers
are allowed to collude, i.e, exchange and combine information. \textit{\emph{We
quantify exactly how eavesdropper collusion degrades the secrecy properties
of the legitimate nodes, in comparison to  a non-colluding scenario.}}
\end{itemize}
This paper is organized as follows. Section~\ref{sec:System-Model}
describes the system model. Section~\ref{sec:Node-Degrees} characterizes
local connectivity in the Poisson $\is$graph. Section~\ref{sec:Techniques}
analyzes two techniques for enhancing the secrecy of communication.
Section~\ref{sec:Secrecy-Capacity} considers the MSR between a node
and its neighbours. Section~\ref{sec:Secrecy-Colluding} characterizes
the case of colluding eavesdroppers. Section~\ref{sec:Conclusion}
concludes the paper and summarizes important findings.

\section{System Model\label{sec:System-Model}}

We start by describing our system model and defining our measures
of secrecy. The notation and symbols used throughout the paper are
summarized in Table~\ref{tab:notation}.

\subsection{Wireless Propagation Characteristics\label{sub:Propagation}}

In a wireless environment, the received power~$P_{\mathrm{rx}}(x_{i},x_{j})$
associated with the link~$\overrightarrow{x_{i}x_{j}}$ can be written
as\begin{equation}
P_{\mathrm{rx}}(x_{i},x_{j})=\PM\cdot g(x_{i},x_{j},Z_{x_{i},x_{j}}),\label{eq:Pr}\end{equation}
where $\PM$ is the (common) transmit power of the legitimate nodes;
and $g(x_{i},x_{j},Z_{x_{i},x_{j}})$ is the power gain of the link~$\overrightarrow{x_{i}x_{j}}$,
where the random variable (RV)~$Z_{x_{i},x_{j}}$ represents the
random propagation effects (such as multipath fading or shadowing)
associated with link~$\overrightarrow{x_{i}x_{j}}$. We consider
that the $Z_{x_{i},x_{j}},x_{i}\neq x_{j}$ are independent identically
distributed (IID) RVs with common probability density function (PDF)~$f_{Z}(z)$,
and that $Z_{x_{i},x_{j}}=Z_{x_{j},x_{i}}$ due to channel reciprocity.
The channel gain~$g(x_{i},x_{j},Z_{x_{i},x_{j}})$ is considered
constant (quasi-static) throughout the use of the communications channel,
which corresponds to channels with a large coherence time. The gain
function is assumed to satisfy the following conditions:
\begin{enumerate}
\item $g(x_{i},x_{j},Z_{x_{i},x_{j}})$ depends on $x_{i}$ and $x_{j}$
only through the link length~$|x_{i}-x_{j}|$; with abuse of notation,
we can write $g(r,z)\triangleq g(x_{i},x_{j},z)|_{|x_{i}-x_{j}|\rightarrow r}$.
\item $g(r,z)$ is continuous and strictly decreasing\emph{ }in $r$.
\item $\lim_{r\rightarrow\infty}g(r,z)=0$.
\end{enumerate}
The proposed model is general enough to account for common choices
of $g$. One example is the unbounded model where $g(r,z)=\frac{z}{r^{2b}}$.
The term~$\frac{1}{r^{2b}}$ accounts for the far-field path loss
with distance, where the amplitude loss exponent~$b$ is environment-dependent
and can approximately range from $0.8$ (e.g.,~hallways inside buildings)
to $4$ (e.g.,~dense urban environments), with $b=1$ corresponding
to free space propagation. This model is analytically convenient~\cite{WinPinShe:J09},
but since the gain becomes unbounded as the distance approaches zero,
it must be used with care for extremely dense networks. Another example
is the bounded model where $g(r,z)=\frac{z}{1+r^{2b}}$. This model
has the same far-field dependence as the unbounded model, but eliminates
the singularity at the origin. Unfortunately, it often leads to intractable
analytical results. The effect of the singularity at $r=0$ on the
performance evaluation of a wireless system is considered in \cite{InaChiPooWic:09}.

Furthermore, by appropriately choosing of the distribution of $Z_{x_{i},x_{j}}$,
both models can account for various random propagation effects~\cite{WinPinShe:J09},
including:
\begin{enumerate}
\item \emph{Path loss only:} $Z_{x_{i},x_{j}}=1$. 
\item \emph{Path loss and Nakagami-$m$ fading:} $Z_{x_{i},x_{j}}=\alpha_{x_{i},x_{j}}^{2}$
, where $\alpha_{x_{i},x_{j}}^{2}\sim\mathcal{G}(m,\frac{1}{m})$.%
\footnote{We use $\mathcal{G}(x,\theta)$ to denote a gamma distribution with
mean~$x\theta$ and variance~$x\theta^{2}$.%
} 
\item \emph{Path loss and log-normal shadowing:} $Z_{x_{i},x_{j}}=\exp(2\sigma_{\mathrm{s}}G_{x_{i},x_{j}})$,
where $G_{x_{i},x_{j}}\sim\mathcal{N}(0,1)$.%
\footnote{We use $\mathcal{N}(\mu,\sigma^{2})$ to denote a Gaussian distribution
with mean~$\mu$ and variance~$\sigma^{2}$.%
} The term~$\exp(2\sigma_{\mathrm{s}}G_{x_{i},x_{j}})$ has a log-normal
distribution, where $\sigma_{\mathrm{s}}$ is the shadowing coefficient. 
\item \emph{Path loss, Nakagami-$m$ fading, and log-normal shadowing:}
$Z_{x_{i},x_{j}}=\alpha_{x_{i},x_{j}}^{2}\exp(2\sigma_{\mathrm{s}}G_{x_{i},x_{j}})$,
where $\alpha_{x_{i},x_{j}}^{2}\sim\mathcal{G}(m,\frac{1}{m})$, $G_{x_{i},x_{j}}\sim\mathcal{N}(0,1)$,
with $\alpha_{x_{i},x_{j}}$ independent of $G_{x_{i},x_{j}}$.
\end{enumerate}

\subsection{Wireless Information-Theoretic Security}

We now define our measure of secrecy more precisely. While our main
interest is targeted towards the behavior of large-scale networks,
we briefly review the setup for a single legitimate link with a single
eavesdropper. The results thereof will serve as basis for the notion
of $\is$graph to be established later.

Consider the model depicted in Fig.~\ref{fig:wiretap-channel}, where
a legitimate user (Alice) wants to send messages to another user (Bob).
Alice encodes a message~$s$, represented by a discrete RV, into
a codeword, represented by the complex random sequence of length~$n$,
$x^{n}=(x(1),\ldots,x(n))\in\mathbb{C}^{n}$, for transmission over
the channel. Bob observes the output of a discrete-time channel (the
\emph{legitimate} \emph{channel}), which at time~$i$ is given by\[
y_{\ell}(i)=h_{\M}\cdot x(i)+w_{\ell}(i),\quad1\leq i\leq n,\]
where $h_{\M}\in\mathbb{C}$ is the quasi-static amplitude gain of
the legitimate channel,%
\footnote{The amplitude gain~$h_{\L}$ can be related to the power gain in
(\ref{eq:Pr}) as $g(\rM,\zM)=|h_{\M}|^{2}$, where $\rM$ and $\zM$
are, respectively, the length and random propagation effects of the
legitimate link. %
} and $w_{\ell}(i)\sim\mathcal{N}_{\mathrm{c}}(0,\WM)$ is AWGN with
power~$\WM$ per complex sample.%
\footnote{We use $\mathcal{N}_{\textrm{c}}(0,\sigma^{2})$ to denote a CS complex
Gaussian distribution, where the real and imaginary parts are IID
$\mathcal{N}(0,\sigma^{2}/2)$.%
} Bob makes a decision~$\hat{s}_{\M}$ on $s$ based on the output~$y_{\ell}$,
incurring in an error probability equal to $\mathbb{P}\{\hat{s}_{\M}\neq s\}.$
A third party (Eve) is also capable of eavesdropping on Alice's transmissions.
Eve observes the output of a discrete-time channel (the \emph{eavesdropper's
channel}), which at time~$i$ is given by\[
y_{\mathrm{e}}(i)=h_{\mathrm{e}}\cdot x(i)+w_{\mathrm{e}}(i),\quad1\leq i\leq n,\]
where $h_{\mathrm{e}}\in\mathbb{C}$ is the quasi-static amplitude
gain of the eavesdropper channel, and $w_{\mathrm{e}}(i)\sim\mathcal{N}_{\mathrm{c}}(0,\WE)$
is AWGN with power~$\WE$ per complex sample. It is assumed that
the signals~$x$, $h_{\M}$, $h_{\mathrm{e}}$, $w_{\M}$, and $w_{\mathrm{e}}$
are mutually independent. Each codeword transmitted by Alice is subject
to the average power constraint of $\PM$ per complex symbol, i.e.,
\begin{equation}
\frac{1}{n}\sum_{i=1}^{n}\mathbb{E}\{|x(i)|^{2}\}\leq\PM.\label{eq:power-constraint}\end{equation}
We define the transmission rate between Alice and Bob as\[
\rate\triangleq\frac{H(s)}{n},\]
where $H(\cdot)$ denotes the entropy function. 

Throughout the paper, we use \emph{strong secrecy} as the condition
for information-theoretic security, and define it as follows~\cite{MauWol:00}.

\begin{definition}[Strong Secrecy]\label{def:strong-secrecy}The
rate~$\rate^{\!\!*}$ is said to be \emph{achievable with} \emph{strong
secrecy} if $\forall\epsilon>0$, for sufficiently large~$n$, there
exists an encoder-decoder pair with rate~$\rate$ satisfying the
following conditions:\begin{align*}
\rate & \geq\rate^{\!\!*}-\epsilon,\\
H(s|y_{\mathrm{e}}^{n}) & \geq H(s)-\epsilon,\\
\mathbb{P}\{\hat{s}_{\M}\neq s\} & \leq\epsilon.\end{align*}

\end{definition}

We define the \emph{maximum secrecy rate} (MSR)~$\Cs$ of the legitimate
channel to be the maximum rate~$\rate^{\!\!*}$ that is achievable
with strong secrecy.%
\footnote{See \cite{BarBlo:08} for a comparison between the concepts of weak
and strong secrecy. In the case of Gaussian noise, the MSR is \emph{the
same }under the weak and strong secrecy conditions.%
} If the legitimate link operates at a rate below the MSR~$\Cs$,
there exists an encoder-decoder pair such that the eavesdropper is
unable to obtain additional information about $s$ from the observation~$y_{\mathrm{e}}^{n}$,
in the sense that $H(s|y_{\mathrm{e}}^{n})$ approaches $H(s)$ as
the codeword length~$n$ grows. It was shown in \cite{LeuHel:78,BloBarRodMcl:08}
that for a given realization of the channel gains~$h_{\M},h_{\mathrm{e}}$,
the MSR of the Gaussian wiretap channel is\begin{equation}
\Cs(x_{i},x_{j})=\left[\log_{2}\left(1+\frac{\PM\cdot|h_{\M}|^{2}}{\WM}\right)-\log_{2}\left(1+\frac{\PM\cdot|h_{\mathrm{e}}|^{2}}{\WE}\right)\right]^{+},\label{eq:Cs-Gaussian}\end{equation}
in bits per complex dimension, where $[x]^{+}=\max\{x,0\}$.%
\footnote{Operationally, the MSR~$\Cs$ can be achieved if Alice first estimates
$h_{\M}$ and $h_{\mathrm{e}}$ (i.e., has full CSI), and then uses
a code that achieves MSR in the AWGN channel. Estimation of $h_{\mathrm{e}}$
is possible, for instance, when Eve is another active user in the
wireless network, so that Alice can estimate the eavesdropper\textquoteright{}s
channel during Eve\textquoteright{}s transmissions. As we shall see,
the $\is$graph model presented in this paper relies on\emph{ }an
outage formulation, and therefore does \emph{not} make assumptions
concerning availability of full CSI.%
} In the next sections, we use these basic results to analyze secrecy
in large-scale networks.

\subsection{$\is$Graph}

Consider a wireless network where legitimate nodes and potential eavesdroppers
are randomly scattered in space, according to some point process.
The $\is$graph is a convenient geometrical representation of the
information-theoretically secure links that can be established on
such network. In the following, we introduce a precise definition
of the $\is$graph, based on the notion of strong secrecy.

\begin{definition}[$\is$graph]Let $\PiM=\{x_{i}\}\subset\mathbb{R}^{d}$
denote the set of legitimate nodes, and $\PiE=\{e_{i}\}\subset\mathbb{R}^{d}$
denote the set of eavesdroppers. The \emph{$\is$graph} is the directed
graph~$G=\{\PiM,\mathcal{E}\}$ with vertex set~$\PiM$ and edge
set\begin{equation}
\mathcal{E}=\{\overrightarrow{x_{i}x_{j}}:\Cs(x_{i},x_{j})>\th\},\label{eq:edges-general}\end{equation}
where $\th$ is a threshold representing the prescribed infimum secrecy
rate for each communication link; and $\Cs(x_{i},x_{j})$ is the MSR,
for a given realization of the channel gains, of the link between
the transmitter~$x_{i}$ and the receiver~$x_{j}$, given by \begin{equation}
\Cs(x_{i},x_{j})=\left[\log_{2}\left(1+\frac{P_{\mathrm{rx}}(x_{i},x_{j})}{\WM}\right)-\log_{2}\left(1+\frac{P_{\mathrm{rx}}(x_{i},e^{*})}{\WE}\right)\right]^{+},\label{eq:Cs-ij-general}\end{equation}
 with\begin{equation}
e^{*}=\underset{e_{k}\in\PiE}{\mathrm{argmax}}\: P_{\mathrm{rx}}(x_{i},e_{k}).\label{eq:e*-general}\end{equation}
\end{definition}

This definition presupposes that the eavesdroppers are not allowed
to \emph{collude} (i.e.,~they cannot exchange or combine information),
and therefore only the eavesdropper with the strongest received signal
from $x_{i}$ determines the MSR between $x_{i}$ and $x_{j}$. The
case of colluding eavesdroppers is analyzed in Section~\ref{sec:Secrecy-Colluding}.

The $\is$graph admits an outage interpretation, in the sense that
legitimate nodes set a target secrecy rate~$\th$ at which they transmit
without knowing the channel state information (CSI) of the legitimate
nodes and eavesdroppers. In this context, an edge between two nodes
signifies that the corresponding channel is not in secrecy outage.

Consider now the particular scenario where the following conditions
hold: (a)~the infimum desired secrecy rate is zero, i.e.,~$\th=0$;
(b)~the wireless environment introduces only path loss, i.e.,~$Z_{x_{i},x_{j}}=1$
in (\ref{eq:Pr}); and (c)~the noise powers of the legitimate users
and eavesdroppers are equal, i.e.,~$\WM=\WE=\W$. Note that by setting
$\th=0$, we are considering the \emph{existence} of secure links,
in the sense that an edge~$\overrightarrow{x_{i}x_{j}}$ is present
if and only if $\Cs(x_{i},x_{j})>0$. Thus, a positive (but possibly
small) rate exists at which $x_{i}$ can transmit to $x_{j}$ with
information-theoretic security. In this scenario, (\ref{eq:Cs-ij-general})
reduces to%
\footnote{For notational simplicity, when $Z=1$, we omit the second argument
of the function~$g(r,z)$ and simply use $g(r)$.%
}\begin{equation}
\Cs(x_{i},x_{j})=\left[\log_{2}\left(1+\frac{\PM\cdot g(|x_{i}-x_{j}|)}{\W}\right)-\log_{2}\left(1+\frac{\PM\cdot g(|x_{i}-e^{*}|)}{\W}\right)\right]^{+},\label{eq:Cs-ij-pathloss}\end{equation}
where \begin{equation}
e^{*}=\underset{e_{k}\in\PiE}{\mathrm{argmin}}\:|x_{i}-e_{k}|,\label{eq:e*-pathloss}\end{equation}
i.e.,~$e^{*}$ is the eavesdropper closest to the transmitter~$x_{i}$.
Since $g(\cdot)$ is strictly decreasing with its argument, the edge
set~$\mathcal{E}$ in (\ref{eq:edges-general}) simplifies in this
case to\begin{equation}
\mathcal{E}=\Bigl\{\overrightarrow{x_{i}x_{j}}:|x_{i}-x_{j}|<|x_{i}-e^{*}|,\quad e^{*}=\underset{e_{k}\in\PiE}{\mathrm{argmin}}\:|x_{i}-e_{k}|\Bigr\},\label{eq:edges-pathloss}\end{equation}
i.e.,~the transmitter~$x_{i}$ can communicate with information-theoretic
security with $x_{j}$ at some positive rate if and only if $x_{j}$
is closer to $x_{i}$ than any other eavesdropper. Thus, in the special
case where $\th=0$, $Z_{x_{i},x_{j}}=1$, and $\WM=\WE$, the $\is$graph
is characterized by a simple geometrical description. Fig.~\ref{fig:secrecy-graph}
shows an example of such an $\is$graph. Note that the description
in (\ref{eq:edges-pathloss}) -- and therefore all results that will
follow from it -- do not depend on the specific form of the function~$g(r)$,
as long as it satisfies the conditions in Section~\ref{sub:Propagation}.
The special case in (\ref{eq:edges-pathloss}) was also considered
in \cite{Hae:08}, starting from a formulation of security based on
geometrical -- not information-theoretic -- considerations.

\subsection{Poisson $\is$Graph}

The spatial location of the nodes can be modeled either deterministically
or stochastically. However, in many important scenarios, only a statistical
description of the node positions is available, and thus a stochastic
spatial model is more suitable. In particular, when the node positions
are unknown to the network designer a priori, we may as well treat
them as completely random according to a homogeneous Poisson point
process~\cite{Kin:93}.%
\footnote{The spatial Poisson process is a natural choice in such situation
because, given that a node is inside a region~$\mathcal{R}$, the
PDF of its position is conditionally uniform over $\mathcal{R}$.%
} The Poisson process has maximum entropy among all homogeneous processes~\cite{Mcf:65},
and serves as a simple and useful model for the position of nodes
in a network.

\begin{definition}[Poisson $\is$graph]The \emph{Poisson $\is$graph} is
an $\is$graph where $\PiM,\PiE\subset\mathbb{R}^{d}$ are mutually
independent, homogeneous Poisson point processes with densities~$\LM$
and $\LE$, respectively.\end{definition}

In the remainder of the paper (unless otherwise indicated), we focus
on Poisson $\is$graphs in $\mathbb{R}^{2}$. We use $\{\RM i\}_{i=1}^{\infty}$
and $\{\RE i\}_{i=1}^{\infty}$ to denote the ordered random distances
between the origin of the coordinate system and the nodes in $\PiM$
and $\PiE$, respectively, where $\RM 1\leq\RM 2\leq\ldots$ and $\RE 1\leq\RE 2\leq\ldots.$

\section{Local Connectivity in the Poisson $\is$Graph\label{sec:Node-Degrees}}

In graph theory, the node degrees are an important property of a graph,
since they describe the connectivity between a node and its immediate
neighbors. In a  graph, the \emph{in-degree} and \emph{out-degree}
of a vertex are, respectively, the number of edges entering and exiting
the vertex. Since the $\is$graph is a random graph, the in- and out-degrees
of the legitimate nodes are RVs. In this section, we provide a complete
probabilistic characterization of both in-degree~$N_{\mathrm{in}}$
and out-degree~$N_{\mathrm{out}}$ of a typical node in the Poisson
$\is$graph.%
\footnote{In this paper, we analyze the local properties of a \emph{typical
node} in the $\is$graph. This notion is made precise in \cite[Sec. 4.4]{StoKenMec:95}
using Palm theory. Specifically, Slivnyak's theorem states that the
properties observed by a typical legitimate node~$x\in\PiM$ are
the same as those observed by node~$0$ in the process~$\PiM\cup\{0\}$.
Informally, a typical node of $\PiM$ is one that is uniformly picked
from a finite region expanding to $\mathbb{R}^{2}$. In this paper,
we often omit the word {}``typical'' for brevity.%
} We first consider the simplest case of $\th=0$ (the \emph{existence}
of secure links), $Z_{x_{i},x_{j}}=1$ (path loss only), and $\WE=\WM$
(equal noise powers) in Sections~\ref{sub:In-Degree-Characterization},
\ref{sub:Out-Degree-Characterization}, and \ref{sub:General-Relationships}.
This scenario leads to an $\is$graph with a simple geometric description,
thus providing various insights that are useful in understanding more
complex cases. Later, in Sections~\ref{sub:Effect-Propagation} and
\ref{sub:Effect-Theta}, we separately analyze how the node degrees
are affected by\textit{\emph{ wireless propagation effects other than
}}path\textit{\emph{ loss (e.g.,~multipath fading), a non-zero secrecy
rate threshold~$\th$, and unequal noise powers~$\WE,\WM$.}}

We start by showing that under the simple geometric description in
(\ref{eq:edges-pathloss}), the distributions of the in- and out-degree
of a node depend exclusively on the ratio of densities~$\frac{\LM}{\LE}$.

\begin{property}In the case of $\th=0$, $Z_{x_{i},x_{j}}=1$, and
$\WE=\WM$, the probability mass functions (PMFs)~$p_{N_{\mathrm{out}}}(n)$
and $p_{N_{\mathrm{in}}}(n)$ of a node depend on the densities~$\LM$
and $\LE$ only through the ratio~$\frac{\LM}{\LE}$.\end{property}

\begin{proof}Consider a given realization of the processes~$\PiM$
and $\PiE$, with densities~$\LM$ and $\LE$, respectively. This
induces an $\is$graph~$G=(\PiM,\mathcal{E})$ with vertex set~$\PiM$
and edge set~$\mathcal{E}$ given by (\ref{eq:edges-pathloss}).
We now apply the transformation~$x\rightarrow\sqrt{c}x$ in $\mathbb{R}^{2}$,
resulting in scaled processes~$\sqrt{c}\PiM$ and $\sqrt{c}\PiE$,
with densities~$\frac{\LM}{c}$ and $\frac{\LE}{c}$, respectively.
Note that the $\is$graph~$\breve{G}=\{\sqrt{c}\PiM,\mathcal{E}\}$
corresponding to the scaled processes has exactly the same edge set
as $G$, because the scaling transformation does not change the geometrical
configuration of the network. We then conclude that the node degree
distributions before and after scaling are the same, and hence only
depend on the ratio~$\frac{\LM}{\LE}$. This concludes the proof.\end{proof}

\subsection{In-Degree Characterization\label{sub:In-Degree-Characterization}}

The characterization of the in-degree relies on the notion of Voronoi
tessellation, which we now introduce. A \emph{planar tessellation}
is a collection of disjoint polygons whose closures cover $\mathbb{R}^{2}$,
and which is locally finite (i.e., the number of polygons intersecting
any given compact set is finite). Given a generic point process~$\Pi=\{x_{i}\}\subset\mathbb{R}^{2}$,
we define the \emph{Voronoi cell}~$\mathcal{C}_{x_{i}}$ of the point~$x_{i}$
as the set of points of $\mathbb{R}^{2}$ which are closer to $x_{i}$
than any other point of $\Pi$, i.e.,\[
\mathcal{C}_{x_{i}}=\{y\in\mathbb{R}^{2}:|y-x_{i}|<|y-x_{j}|,\forall x_{j}\neq x_{i}\}.\]
The collection~$\{C_{x_{i}}\}$ of all the cells forms a random \emph{Voronoi
tessellation} with respect to the underlying point process~$\Pi$.
Let $\mathcal{C}_{0}$ denote the \emph{typical} \emph{Voronoi cell},
i.e., the Voronoi cell associated with a point placed at the origin,
according to Slivnyak\textquoteright{}s theorem. Using the notions
just introduced, the following theorem provides a probabilistic characterization
of the in-degree of the $\is$graph.

\begin{theorem}The in-degree~$N_{\mathrm{in}}$ of a typical node
in the Poisson $\is$graph has the following moment generating function
(MGF)\begin{equation}
M_{N_{\mathrm{in}}}(s)=\mathbb{E}\left\{ \exp\left(\frac{\LM}{\LE}\widetilde{A}(e^{s}-1)\right)\right\} ,\label{eq:M-Nin}\end{equation}
where $\widetilde{A}$ is the area of a typical Voronoi cell induced
by a unit-density Poisson process. Furthermore, all the moments of
$N_{\mathrm{in}}$ are given by\begin{equation}
\mathbb{E}\{N_{\mathrm{in}}^{n}\}=\sum_{k=1}^{n}\left(\frac{\LM}{\LE}\right)^{k}S(n,k)\,\mathbb{E}\{\widetilde{A}^{k}\},\quad n\geq1,\label{eq:E-Nin-n}\end{equation}
where $S(n,k)$, $1\leq k\leq n$, are the Stirling numbers of the
second kind~\cite[Ch. 24]{AbrSte:70}.\end{theorem}

\begin{proof}Using Slivnyak\textquoteright{}s theorem~\cite[Sec. 4.4]{StoKenMec:95},
we consider the process~$\PiM\cup\{0\}$ obtained by adding a legitimate
node to the origin of the coordinate system, and denote the in-degree
of the node at the origin by $N_{\mathrm{in}}$. The RV~$N_{\mathrm{in}}$
corresponds to the number of nodes from the process~$\PiM$ that
fall inside the typical Voronoi cell~$\mathcal{C}_{0}$ constructed
from the process~$\PiE\cup\{0\}$. This is depicted in Fig.~\ref{fig:in-degree}.
Denoting the random area of such a cell by $A$, the MGF of $N_{\mathrm{in}}$
is given by\begin{align*}
M_{N_{\mathrm{in}}}(s) & =\mathbb{E}\{e^{sN_{\mathrm{in}}}\}\\
 & =\mathbb{E}\{\exp\left(\LM A(e^{s}-1)\right)\},\end{align*}
where we used the fact that conditioned on $A$, the RV~$N_{\mathrm{in}}$
is Poisson distributed with parameter~$\LM A$. If $\widetilde{A}$
denotes the random area of a typical Voronoi cell induced by a \emph{unit-density}
Poisson process, then $\widetilde{A}=A\LE$ and (\ref{eq:M-Nin})
follows. This completes the first half of the proof. 

To obtain the moments of $N_{\mathrm{in}}$, we use Dobinski's formula~\cite{Dob:77}\[
\sum_{k=0}^{\infty}k^{n}\frac{e^{-\mu}\mu^{k}}{k!}=\sum_{k=1}^{n}\mu^{k}S(n,k),\]
which establishes the relationship between the $n$-th moment of a
Poisson RV with mean~$\mu$ and the Stirling numbers of the second
kind, $S(n,k)$. Then,\begin{align*}
\mathbb{E}\{N_{\mathrm{in}}^{n}\} & =\mathbb{E}\left\{ \mathbb{E}\{N_{\mathrm{in}}^{n}|A\}\right\} \\
 & =\mathbb{E}\left\{ \sum_{k=1}^{n}(\LM A)^{k}S(n,k)\right\} \\
 & =\sum_{k=1}^{n}\left(\frac{\LM}{\LE}\right)^{k}S(n,k)\,\mathbb{E}\{\widetilde{A}^{k}\},\end{align*}
for $n\geq1$. This is the result in (\ref{eq:E-Nin-n}) and the second
half of proof is concluded.\end{proof}

Equation~(\ref{eq:E-Nin-n}) expresses the moments of $N_{\mathrm{in}}$
in terms of the moments of $\widetilde{A}$. Note that the Stirling
numbers of the second kind can be obtained recursively as\begin{align*}
S(n,k) & =S(n-1,k-1)+kS(n-1,k),\\
S(n,n) & =S(n,1)=1,\end{align*}
or explicitly as\[
S(n,k)=\frac{1}{k!}\sum_{i=0}^{k}(-1)^{i}{k \choose i}(k-i)^{n}.\]
Table~\ref{tab:Stirling} provides some values for $S(n,k)$. In
general, $\mathbb{E}\{\widetilde{A}^{k}\}$ cannot be obtained in
closed form, except in the case of $k=1$, which is derived below
in (\ref{eq:E-A-tilde-1}). For $k=2$ and $k=3$, $\mathbb{E}\{\widetilde{A}^{k}\}$
can be expressed as multiple integrals and then computed numerically~\cite{Gil:62,Bra:85,HayQui:02}.
Alternatively, the moments of $\widetilde{A}$ can be determined using
Monte Carlo simulation of random Poisson-Voronoi tessellations~\cite{Cra:78,HinMil:80,Bra:86}.
The first four moments of $\widetilde{A}$ are given in Table~\ref{tab:Moments-A}.

The above theorem can be used to obtain the in-connectivity properties
a node, such as the in-isolation probability, as given in the following
corollary.

\begin{corollary}The average in-degree of a typical node in the Poisson
$\is$graph is\begin{equation}
\mathbb{E}\{N_{\mathrm{in}}\}=\frac{\LM}{\LE}\label{eq:E-Nin}\end{equation}
and the probability that a typical node cannot receive from anyone
with positive secrecy rate (in-isolation) is\begin{equation}
p_{\mathrm{in-isol}}=\mathbb{E}\left\{ e^{-\frac{\LM}{\LE}\widetilde{A}}\right\} .\label{eq:p-in-isol}\end{equation}
\end{corollary}

\begin{proof}Setting $n=1$ in (\ref{eq:E-Nin-n}), we obtain $\mathbb{E}\{N_{\mathrm{in}}\}=\frac{\LM}{\LE}\mathbb{E}\{\widetilde{A}\}$.
Noting that\[
\widetilde{A}={\int\intop}_{\mathbb{R}^{2}}\mathbbm{1}\{z\in\mathcal{C}_{0}\}dz,\]
where $\mathcal{C}_{0}$ is the typical Voronoi cell induced by a
unit-density Poisson process~$\widetilde{\Pi}$, we can write%
\footnote{We use $\mathcal{B}_{x}(\rho)\triangleq\{y\in\mathbb{R}^{2}:|y-x|\leq\rho\}$
to denote the closed two-dimensional ball centered at point~$x$,
with radius~$\rho$. %
}\begin{align}
\mathbb{E}\{\widetilde{A}\} & ={\int\intop}_{\mathbb{R}^{2}}\mathbb{P}\{z\in\mathcal{C}_{0}\}dz\label{eq:E-A-tilde-1}\\
 & ={\int\intop}_{\mathbb{R}^{2}}\mathbb{P}\{\widetilde{\Pi}\{\mathcal{B}_{z}(|z|)\}=0\}dz\label{eq:E-A-tilde-2}\\
 & =\int_{0}^{2\pi}\int_{0}^{\infty}e^{-\pi r^{2}}rdrd\theta=\nonumber \\
 & =1.\nonumber \end{align}
Equation (\ref{eq:E-A-tilde-1}) follows from Fubini's Theorem, while
(\ref{eq:E-A-tilde-2}) follows from the fact that, for any $z\in\mathbb{R}^{2}$,
the event~$\{z\in\mathcal{C}_{0}\}$ is equivalent to having no points
of $\widetilde{\Pi}$ in $\mathcal{B}_{z}(|z|)$, as depicted in Fig.~\ref{fig:aux-voronoi}.
This completes the proof of (\ref{eq:E-Nin}). To derive (\ref{eq:p-in-isol}),
note that the RV~$N_{\mathrm{in}}$ conditioned on $A$ is Poisson
distributed with parameter~$\LM A$, and thus $p_{\mathrm{in-isol}}=p_{N_{\mathrm{in}}}(0)=\mathbb{E}\{p_{N_{\mathrm{in}}|A}(0)\}=\mathbb{E}\left\{ e^{-\frac{\LM}{\LE}\widetilde{A}}\right\} $.\end{proof}

We can obtain an alternative expression for (\ref{eq:p-in-isol})
by performing a power series expansion of the exponential function,
resulting in\[
p_{\mathrm{in-isol}}=\sum_{k=0}^{\infty}\frac{(-1)^{k}}{k!}\left(\frac{\LM}{\LE}\right)^{k}\mathbb{E}\{\widetilde{A}^{k}\}.\]
This equation expresses $p_{\mathrm{in-isol}}$ as a power series
with argument~$\frac{\LM}{\LE}$, since $\mathbb{E}\{\widetilde{A}^{k}\}$
are deterministic. The power series can be truncated, since the summands
become smaller as $k\rightarrow\infty$.

\subsection{Out-Degree Characterization\label{sub:Out-Degree-Characterization}}

\begin{theorem}\label{thm:p-Nout}The out-degree~$N_{\mathrm{out}}$
of a typical node in the Poisson $\is$graph has the following geometric
PMF\begin{equation}
p_{N_{\mathrm{out}}}(n)=\left(\frac{\LM}{\LM+\LE}\right)^{n}\left(\frac{\LE}{\LM+\LE}\right),\quad n\geq0.\label{eq:p-Nout}\end{equation}
\end{theorem}

\begin{proof}We consider the process~$\PiM\cup\{0\}$ obtained by
adding a legitimate node to the origin of the coordinate system, and
denote the out-degree of the origin by $N_{\mathrm{out}}$. The RV~$N_{\mathrm{out}}$
corresponds to the number of nodes from the process~$\PiM$  that
fall inside the circle with random radius~$\RE 1$ centered at the
origin, i.e., $N_{\mathrm{out}}=\#\{\RM i:\RM i<\RE 1\}$. This is
depicted in Fig.~\ref{fig:out-degree}. To determine the PMF\ of
$N_{\mathrm{out}}$, consider the one-dimensional arrival processes~$\widetilde{\Pi}_{\ell}=\{\RM i^{2}\}_{i=1}^{\infty}$
and $\widetilde{\Pi}_{\mathrm{e}}=\{\RE i^{2}\}_{i=1}^{\infty}$.
As can be easily shown using the mapping theorem~\cite[Section 2.3]{Kin:93},
$\widetilde{\Pi}_{\ell}$ and $\widetilde{\Pi}_{\mathrm{e}}$ are
independent homogeneous Poisson processes with arrival rates~$\pi\LM$
and $\pi\LE$, respectively. When there is an arrival in the merged
process~$\widetilde{\Pi}_{\ell}\cup\widetilde{\Pi}_{\mathrm{e}}$,
it comes from process $\widetilde{\Pi}_{\ell}$ with probability~$p=\frac{\pi\LM}{\pi\LM+\pi\LE}=\frac{\LM}{\LM+\LE}$,
and from $\widetilde{\Pi}_{\mathrm{e}}$ with probability~$1-p=\frac{\LE}{\LM+\LE}$,
and these events are independent for different arrivals~\cite{BerTsi:02}.
Since the event~$\{N_{\mathrm{out}}=n\}$ is equivalent to the occurrence
of $n$~arrivals from $\widetilde{\Pi}_{\ell}$ followed by one arrival
from $\widetilde{\Pi}_{\mathrm{e}}$, then we have the geometric PMF~$p_{N_{\mathrm{out}}}(n)=p^{n}(1-p),$
$n\geq0$, with parameter~$p=\frac{\LM}{\LM+\LE}$. This is the result
in (\ref{eq:p-Nout}) and the proof is completed.\end{proof}

Note that this particular result was also derived in \cite{Hae:08}.
The above theorem can be used to obtain the out-connectivity properties
a node, such as the out-isolation probability, as given in the following
corollary.

\begin{corollary}The average out-degree of a typical node in the
Poisson $\is$graph is\begin{equation}
\mathbb{E}\{N_{\mathrm{out}}\}=\frac{\LM}{\LE},\label{eq:E-Nout}\end{equation}
and the probability that a typical node cannot transmit to anyone
with positive secrecy rate (out-isolation) is\begin{equation}
p_{\mathrm{out-isol}}=\frac{\LE}{\LM+\LE}.\label{eq:p-out-isol}\end{equation}
\end{corollary}

\begin{proof}This follows directly from Theorem~\ref{thm:p-Nout}.\end{proof}

\subsection{General Relationships Between In- and Out-Degree\label{sub:General-Relationships}}

We have so far considered the probabilistic distribution of the in-
and out-degrees in a separate fashion. This section establishes a
direct comparison between some characteristics of the in- and out-degrees.

\begin{property}\label{prop:ENin-ENout}For the Poisson $\is$graph
with $\LM>0$ and $\LE>0$, the average degrees of a typical node
satisfy\begin{equation}
\mathbb{E}\{N_{\mathrm{in}}\}=\mathbb{E}\{N_{\mathrm{out}}\}=\frac{\LM}{\LE}.\label{eq:ENin-ENout}\end{equation}
 \end{property}

\begin{proof}This follows directly by comparing (\ref{eq:E-Nout})
and (\ref{eq:E-Nin}).\end{proof}

The property~$\mathbb{E}\{N_{\mathrm{in}}\}=\mathbb{E}\{N_{\mathrm{out}}\}$
is valid in general for any directed random graph.

\begin{property}\label{prop:pin-pout-isol}For the Poisson $\is$graph
with $\LM>0$ and $\LE>0$, the probabilities of in- and out-isolation
of a typical node satisfy\begin{equation}
p_{\mathrm{in-isol}}<p_{\mathrm{out-isol}}.\label{eq:pout-pin-isol}\end{equation}
\end{property}

\begin{proof}Let $\PiE\{\mathcal{R}\}\triangleq\#\{\PiE\cap\mathcal{R}\}$
denote the number of eavesdroppers inside region~$\mathcal{R}$.
With this definition, we can rewrite the edge set~$\mathcal{E}$
in (\ref{eq:edges-pathloss}) as\begin{equation}
\mathcal{E}=\{\overrightarrow{x_{i}x_{j}}:\PiE\{\mathcal{B}_{x_{i}}(|x_{i}-x_{j}|\}=0\},\label{eq:edges-pathloss-alt}\end{equation}
i.e.,~$x_{i}$ is connected to $x_{j}$ if and only if the ball centered
at $x_{i}$ with radius~$|x_{i}-x_{j}|$ is free of eavesdroppers.
We consider the process~$\PiM\cup\{0\}$ obtained by adding a legitimate
node to the origin of the coordinate system. Let $\breve{x}_{i}$
denote the ordered points in process~$\PiM$ of legitimate nodes,
such that $|\breve{x}_{1}|<|\breve{x}_{2}|<\ldots$. From (\ref{eq:edges-pathloss-alt}),
the node at the origin is out-isolated if and only if $\PiE\{\mathcal{B}_{0}(|\breve{x}_{j}|)\}\geq1$
for all $j\geq1$. This is depicted in Fig.~\ref{fig:aux-out-isol}.
Since the balls~$\mathcal{B}_{0}(|\breve{x}_{j}|)$, $j\geq1$, are
concentric at the origin, we have that\begin{align*}
p_{\mathrm{out-isol}} & =\mathbb{P}\left\{ \PiE\{\mathcal{B}_{0}(|\breve{x}_{1}|)\}\geq1\right\} .\end{align*}

Similarly, we see from (\ref{eq:edges-pathloss-alt}) that the node
at the origin is in-isolated if and only if $\PiE\{\mathcal{B}_{\breve{x}_{i}}(|\breve{x}_{i}|)\}\geq1$
for all $i\geq1$. This is depicted in Fig.~\ref{fig:aux-in-isol}.
Then,\begin{align}
p_{\mathrm{in-isol}} & =\mathbb{P}\left\{ \bigwedge_{i=1}^{\infty}\PiE\{\mathcal{B}_{\breve{x}_{i}}(|\breve{x}_{i}|)\}\geq1\right\} \label{eq:pin-isol-proof}\\
 & <\mathbb{P}\left\{ \PiE\{\mathcal{B}_{\breve{x}_{1}}(|\breve{x}_{1}|)\}\geq1\right\} \label{eq:pin-isol-strict}\\
 & =\mathbb{P}\left\{ \PiE\{\mathcal{B}_{0}(|\breve{x}_{1}|)\}\geq1\right\} \label{eq:pout-pin-isol-repeat}\\
 & =p_{\mathrm{out-isol}}.\nonumber \end{align}
The fact that the inequality in (\ref{eq:pin-isol-strict}) is strict
proved in Appendix~\ref{sec:Deriv-Strict-Ineq}. Equation (\ref{eq:pout-pin-isol-repeat})
follows from the spatial invariance of the homogeneous Poisson process~$\PiE$.
This concludes the proof.\end{proof}

Intuitively, out-isolation is \emph{more likely} than in-isolation
because out-isolation only requires that one or more eavesdroppers
are closer than the nearest legitimate node~$\breve{x}_{1}$. On
the other hand, in-isolation requires that \emph{every} ball~$\mathcal{B}_{\breve{x}_{i}}(|\breve{x}_{i}|)$,
$i\geq1$, has one or more eavesdroppers, which is less likely. Property~\ref{prop:pin-pout-isol}
can then be restated in the following way: \emph{it is easier} \emph{for
an individual node to be in}-\emph{connected than out-connected.}

\subsection{Effect of the Wireless Propagation Characteristics\label{sub:Effect-Propagation}}

We have so far analyzed the local connectivity of the $\is$graph
in the presence of path loss only. However, wireless propagation typically
introduces random propagation effects such as multipath fading and
shadowing, which are modeled by the RV~$Z_{x_{i},x_{j}}$ in (\ref{eq:Pr}).
In this section, we aim to quantify the impact of such propagation
effects on the local connectivity of a node. 

Considering $\th=0$, $\WM=\WE=\W$, and arbitrary propagation effects~$Z_{x_{i},x_{j}}$
with PDF~$f_{Z}(z)$, we can combine (\ref{eq:Cs-ij-general}) with
the general propagation model of (\ref{eq:Pr}) and write\begin{equation}
\Cs(x_{i},x_{j})=\left[\log_{2}\left(1+\frac{\PM\cdot g(|x_{i}-x_{j}|,Z_{x_{i},x_{j}})}{\W}\right)-\log_{2}\left(1+\frac{\PM\cdot g(|x_{i}-e^{*}|,Z_{x_{i},e^{*}}),}{\W}\right)\right]^{+},\label{eq:Cs-ij-prop}\end{equation}
where\begin{equation}
e^{*}=\underset{e_{k}\in\PiE}{\mathrm{argmax}}\: g(|x_{i}-e_{k}|,Z_{x_{i},e_{k}}).\label{eq:g*-prop}\end{equation}
After some algebra, the edge set for the resulting $\is$graph can
be written as\begin{equation}
\mathcal{E}=\Bigl\{\overrightarrow{x_{i}x_{j}}:g(|x_{i}-x_{j}|,Z_{x_{i},x_{j}})>g(|x_{i}-e^{*}|,Z_{x_{i},e^{*}}),\quad e^{*}=\underset{e_{k}\in\PiE}{\mathrm{argmax}}\: g(|x_{i}-e_{k}|,Z_{x_{i},e_{k}})\Bigr\}.\label{eq:edges-prop}\end{equation}
Unlike the case of path-loss only, where the out-connections of a
node are determined only by the \emph{closest} eavesdropper, here
they are determined by the eavesdropper with the \emph{least attenuated}
channel. We start by characterizing the distribution of the out-degree
by the following theorem.

\begin{theorem}\label{thm:p-Nout-prop}For the Poisson $\is$graph
with propagation effects~$Z_{x_{i},x_{j}}$ whose PDF is given by
a continuous function~$f_{Z}(z)$, the PMF of the out-degree~$N_{\mathrm{out}}$
of a typical node is given in (\ref{eq:p-Nout}), and is \emph{invariant}
with respect to $f_{Z}(z)$.\end{theorem}

\begin{proof}We consider the process~$\PiM\cup\{0\}$ obtained by
adding a legitimate node to the origin of the coordinate system, and
denote the out-degree of the node at the origin by $N_{\mathrm{out}}$.
 For the legitimate nodes, let the distances to the origin (not necessarily
ordered) be $\RM i\triangleq|x_{i}|,$ $x_{i}\in\PiM$, and the corresponding
channel propagation effects be $\ZM i$. Similarly, we can define
$R_{\mathrm{e},i}\triangleq|e_{i}|$, $e_{i}\in\PiE$, and $Z_{\mathrm{e},i}$
for the eavesdroppers. Define also the loss function as $l(r,z)\triangleq1/g(r,z)$.
We can now consider the one-dimensional  loss processes for the legitimate
nodes, $\LambdaM\triangleq\{\lM i\}_{i=1}^{\infty}$ with $\lM i\triangleq l(\RM i,\ZM i)$,
and for the eavesdroppers, $\Lambda_{\mathrm{e}}\triangleq\{L_{\mathrm{e},i}\}_{i=1}^{\infty}$
with $L_{\mathrm{e},i}\triangleq l(R_{\mathrm{e},i},Z_{\mathrm{e},i})$.
Note that loss process~$\{\lM i\}$ can be interpreted as a stochastic
mapping of the distance process~$\{\RM i\}$, where the mapping depends
on the random sequence~$\{\ZM i\}$ (a similar statement can be made
for $\{L_{\mathrm{e},i}\}$, $\{R_{\mathrm{e},i}\}$, and $\{Z_{\mathrm{e},i}\}$).
With these definitions, the out-degree of node~$0$ can be expressed
as $N_{\mathrm{out}}=\#\{\lM i:\lM i<\min_{k}L_{\mathrm{e},k}\}$,
i.e., it is the number of occurrences in the process~$\LambdaM$
before the \emph{first} occurrence in the process~$\Lambda_{\mathrm{e}}$.
In the remainder of the proof, we first characterize the processes~$\LambdaM$
and $\Lambda_{\mathrm{e}}$; then, using appropriate transformations,
we map them into homogeneous processes, where the distribution of
$N_{\mathrm{out}}$ can be readily determined.

Since the RVs~$\{\ZM i\}$ are IID in $i$ and independent of $\{\RM i\}$,
we know from the marking theorem~\cite[Section 5.2]{Kin:93} that
the points~$\{(\RM i,\ZM i)\}$ form a non-homogeneous Poisson process
on $\mathbb{R}^{+}\times\mathbb{R}^{+}$ with density~$2\pi\LM rf_{Z}(z)$,
where $f_{Z}(z)$ is the PDF of $\ZM i$. Then, from the mapping theorem~\cite[Section 2.3]{Kin:93},
$\LambdaM=\{l(\RM i,\ZM i)\}$ is also a non-homogeneous Poisson process
on $\mathbb{R}^{+}$ with density denoted by $\lambda_{\LambdaM}(l)$.%
\footnote{In our theorem, the continuity of the function~$f_{Z}(z)$ is sufficient
to ensure that $\LambdaM$ is a Poisson process. In general, we may
allow Dirac impulses in $f_{Z}(z)$, as long as the distinct points~$\{(\RM i,\ZM i)\}$
do not pile on top of one another when forming the process~$\LambdaM=\{l(\RM i,\ZM i)\}$.%
} Furthermore, the process~$\LambdaM$ can be made homogeneous through
the transformation $M_{\LambdaM}(t)\triangleq\int_{0}^{t}\lambda_{\LambdaM}(l)dl$,
such that $M_{\LambdaM}(\LambdaM)$ is a Poisson process with density~$1$.
The homogenizing function~$M_{\Lambda_{\ell}}(t)$ can be calculated
as follows\begin{align*}
M_{\LambdaM}(t) & =\int_{0}^{t}\lambda_{\LambdaM}(l)dl\\
 & ={\int\intop}_{0<l(r,z)<t}2\pi\LM rf_{\zM}(z)drdz\end{align*}
Using a completely analogous reasoning for the process~$\Lambda_{\mathrm{e}}$,
its homogenizing function~$M_{\Lambda_{\mathrm{e}}}(t)$ can be written
as\begin{align*}
M_{\Lambda_{\mathrm{e}}}(t) & =\int_{0}^{t}\lambda_{\Lambda_{\mathrm{e}}}(l)dl\\
 & ={\int\intop}_{0<l(r,z)<t}2\pi\LE rf_{Z_{\mathrm{e}}}(z)drdz.\end{align*}
But since $f_{\zM}(z)=f_{Z_{\mathrm{e}}}(z)$, it follows that $M_{\LambdaM}(t)=\frac{\LM}{\LE}M_{\Lambda_{\mathrm{e}}}(t).$
The out-degree~$N_{\mathrm{out}}$ can now be easily obtained in
the homogenized domain. Consider that both processes~$\LambdaM$
and $\Lambda_{\mathrm{e}}$ are homogenized by the \emph{same} transformation~$M_{\LambdaM}(\cdot)$,
such that $M_{\LambdaM}(\Lambda_{\mathrm{l}})$ and $M_{\LambdaM}(\Lambda_{\mathrm{e}})$
are independent Poisson processes with density~$1$ and $\frac{\LE}{\LM}$.
Furthermore, since $M_{\LambdaM}(\cdot)$ is monotonically increasing,
$N_{\mathrm{out}}$ can be re-expressed as\begin{align*}
N_{\mathrm{out}} & =\#\{\lM i:\lM i<\min_{k}L_{\mathrm{e},k}\},\\
 & =\#\{\lM i:M_{\LambdaM}(\lM i)<M_{\LambdaM}(\min_{k}L_{\mathrm{e},k})\}.\end{align*}
In this homogenized domain, the propagation effects have disappeared,
and the problem is now equivalent to that in Theorem~\ref{thm:p-Nout}.
Specifically, when there is an arrival in the merged process~$M_{\LambdaM}(\LambdaM)\cup M_{\LambdaM}(\Lambda_{\mathrm{e}})$,
it comes from process~$M_{\LambdaM}(\LambdaM)$ with probability~$p=\frac{1}{1+\LE/\LM}=\frac{\LM}{\LM+\LE}$,
and from $M_{\LambdaM}(\Lambda_{\mathrm{e}})$ with probability~$1-p=\frac{\LE}{\LM+\LE}$.
As a result, $N_{\mathrm{out}}$ has the geometric PMF~$p_{N_{\mathrm{out}}}(n)=p^{n}(1-p),$
$n\geq0$, with parameter~$p=\frac{\LM}{\LM+\LE}$. This is exactly
the same PMF as the one given in (\ref{eq:p-Nout}), and is therefore
invariant with respect to the distribution~$f_{Z}(z)$. This concludes
the proof.\end{proof}

Intuitively, the propagation environment affect both the legitimate
nodes and eavesdroppers in the same way (in the sense that $\ZM i$
and $Z_{\mathrm{e},i}$ have the same distribution), such that the
PMF\ of $N_{\mathrm{out}}$ is invariant with respect to the PDF~$f_{Z}(z)$.
However, the PMF\ of $N_{\mathrm{in}}$ \emph{does} depend on $f_{Z}(z)$
in a non-trivial way, although its mean remains the same, as specified
in the following corollary.

\begin{corollary}For the Poisson $\is$graph with propagation effects~$Z_{x_{i},x_{j}}$
distributed according to $f_{Z}(z)$, the average node degrees are\begin{equation}
\mathbb{E}\{N_{\mathrm{in}}\}=\mathbb{E}\{N_{\mathrm{out}}\}=\frac{\LM}{\LE},\label{eq:ENin-ENout-prop}\end{equation}
for any distribution~$f_{Z}(z)$.\end{corollary}

\begin{proof}This follows directly from Theorem~\ref{thm:p-Nout-prop}
and the fact that $\mathbb{E}\{N_{\mathrm{in}}\}=\mathbb{E}\{N_{\mathrm{out}}\}$
in any directed random graph.\end{proof}

We thus conclude that the expected node degrees are invariant with
respect to the distribution characterizing the propagation effects,
and always equal the ratio~$\frac{\LM}{\LE}$ of spatial densities.

\subsection{Effect of the Secrecy Rate Threshold and Noise Powers\label{sub:Effect-Theta}}

We have so far analyzed the local connectivity of the $\is$graph
based on the \emph{existence} of positive MSR, by considering that
the infimum desired secrecy rate is zero, i.e.,~$\th=0$ in (\ref{eq:edges-general}).
This implies that the edge~$\overrightarrow{x_{i}x_{j}}$ is present
if and only if there exists a positive rate at which $x_{i}$ can
transmit to $x_{j}$ with information-theoretic security. We have
furthermore considered that the noise powers of the legitimate users
and eavesdroppers are equal, i.e.,~$\WM=\WE$ in (\ref{eq:Cs-ij-general}).
Under these two conditions, the $\is$graph can be reduced to the
simple geometric description in (\ref{eq:edges-pathloss}), where
the edge~$\overrightarrow{x_{i}x_{j}}$ is present if and only if
$x_{j}$ is \emph{closer} to $x_{i}$ than any other eavesdropper.
In this section, we study the effect of non-zero secrecy rate threshold,
i.e.,~$\th>0$, and unequal noise powers, i.e., $\WM\neq\WE$, on
the $\is$graph.

Considering $Z_{x_{i},x_{j}}=1$ and arbitrary noise powers~$\WM,\WE$,
we can combine (\ref{eq:Cs-ij-general}) with the general propagation
model of (\ref{eq:Pr}) and write \begin{equation}
\Cs(x_{i},x_{j})=\left[\log_{2}\left(1+\frac{\PM\cdot g(|x_{i}-x_{j}|)}{\WM}\right)-\log_{2}\left(1+\frac{\PM\cdot g(|x_{i}-e^{*}|)}{\WE}\right)\right]^{+},\label{eq:Cs-ij-thetaW}\end{equation}
where\begin{equation}
e^{*}=\underset{e_{k}\in\PiE}{\mathrm{argmin}}\:|x_{i}-e_{k}|.\label{eq:e*-thetaW}\end{equation}
We can now replace this expression for $\Cs(x_{i},x_{j})$ into (\ref{eq:edges-general})
while allowing an arbitrary threshold~$\th$. After some algebra,
the edge set for the resulting $\is$graph can be written as\begin{equation}
\mathcal{E}=\Bigl\{\overrightarrow{x_{i}x_{j}}:g(|x_{i}-x_{j}|)>\frac{\WM}{\WE}2^{\th}g(|x_{i}-e^{*}|)+\frac{\WM}{\PM}(2^{\th}-1),\quad e^{*}=\underset{e_{k}\in\PiE}{\mathrm{argmin}}\:|x_{i}-e_{k}|\Bigr\}.\label{eq:edges-thetaW}\end{equation}
By setting $\th=0$ and $\WM=\WE$ in (\ref{eq:edges-thetaW}) we
obtain the edge set in (\ref{eq:edges-pathloss}) as a special case.
However, for arbitrary parameters~$\th,\WM,\WE$, the $\is$graph
can no longer be characterized by the simple geometric description
of (\ref{eq:edges-pathloss}). We now analyze the impact of the secrecy
rate threshold~$\th$ and the noise powers~$\WM,\WE$ on the average
node degrees, for a general channel gain function~$g(r)$.

\begin{property}\label{prop:ENin-ENout-thetaW}For the Poisson $\is$graph
with edge set in (\ref{eq:edges-thetaW}) and any channel gain function~$g(r)$
satisfying the conditions in Section~\ref{sub:Propagation}, the
average node degrees~$\mathbb{E}\{N_{\mathrm{out}}\}=\mathbb{E}\{N_{\mathrm{in}}\}$
are decreasing functions of $\th$ and $\WM$, and increasing functions
of $\WE$.\end{property}

\begin{proof}We prove the theorem with a coupling argument. We consider
the process~$\PiM\cup\{0\}$ obtained by adding a legitimate node
to the origin of the coordinate system, and denote the out-degree
of the node at the origin by $N_{\mathrm{out}}$. Let $\RE 1\triangleq\min_{e_{i}\in\PiE}|e_{i}|$
be the random distance between the origin and its closest eavesdropper.
We first consider the variation of $\mathbb{E}\{N_{\mathrm{out}}\}$
with $\th$, for fixed~$\WM,\WE$. Let $X(\th)\triangleq\left\{ x_{i}\in\PiM:g(|x_{i}|)>\frac{\WM}{\WE}2^{\th}g(\RE 1)+\frac{\WM}{\PM}(2^{\th}-1)\right\} $
be the set of legitimate nodes to which the origin is out-connected.
With this definition,\[
\mathbb{E}\{N_{\mathrm{out}}(\th)\}=\mathbb{E}_{\PiM,\RE 1}\{\#X(\th)\},\]
where we have explicitly indicated the dependence of $\mathbb{E}\{N_{\mathrm{out}}\}$
on $\th$. Since $\frac{\WM}{\WE}2^{\th}g(\RE 1)+\frac{\WM}{P_{\mathrm{l}}}(2^{\th}-1)$
is increasing in $\th$, for each realization of $\Pi$ and $\RE 1$
we have that $X(\th_{1})\supseteq X(\th_{2})$, whenever $0<\th_{1}<\th_{2}$.
This implies that $\mathbb{E}_{\PiM,\RE 1}\{\#X(\th_{1})\}\geq\mathbb{E}_{\PiM,\RE 1}\{\#X(\th_{2})\}$,
or equivalently, $\mathbb{E}\{N_{\mathrm{out}}(\th_{1})\}\geq\mathbb{E}\{N_{\mathrm{out}}(\th_{2})\}$
for $0<\th_{1}<\th_{2}$, and thus $\mathbb{E}\{N_{\mathrm{out}}(\th)\}$
is decreasing with $\th$. A similar argument holds for the parameters~$\WM,\WE$,
showing that $\mathbb{E}\{N_{\mathrm{out}}\}$ is decreasing with
$\WM$ and increasing with $\WE$. This concludes the proof.\end{proof}

In essence, by increasing the secrecy rate threshold~$\th$, the
requirement~$C_{s}(x_{i},x_{j})>\th$ for any two nodes~$x_{i},x_{j}$
to be securely connected becomes stricter, and thus the local connectivity
(as measured by the average node degrees) becomes worse. On the other
hand, increasing $\WM$ or decreasing $\WE$ makes the requirement~$C_{s}(x_{i},x_{j})>\th$
harder to satisfy for any two legitimate nodes~$x_{i},x_{j}$. As
a result, the local connectivity (as measured by the average node
degrees) becomes worse.

The exact dependence of the average node degree on the parameters~$\th,\WM,\WE$
depends on the function~$g(r)$. To gain further insights, we now
consider the specific channel gain function\begin{equation}
g(r)=\frac{1}{r^{2b}},\quad r>0.\label{eq:g-ub}\end{equation}
This function has been widely used in the literature to model path
loss behavior as a function of distance, and satisfies the conditions
in Section~\ref{sub:Propagation}. Replacing (\ref{eq:g-ub}) into
(\ref{eq:edges-thetaW}) and rearranging terms, the edge set reduces
to\begin{equation}
\mathcal{E}=\left\{ \overrightarrow{x_{i}x_{j}}:|x_{i}-x_{j}|<\frac{|x_{i}-e^{*}|}{\left(\frac{\WM}{\WE}2^{\th}+\frac{\WM}{\PM}(2^{\th}-1)|x_{i}-e^{*}|^{2b}\right)^{1/2b}},\quad e^{*}=\underset{e_{k}\in\PiE}{\mathrm{argmin}}\:|x_{i}-e_{k}|\Bigr\}\right\} .\label{eq:edges-thetaW-ub}\end{equation}
For this case, a characterization of the first order moments of $N_{\mathrm{in}}$
and $N_{\mathrm{out}}$ is possible, and is provided in the following
theorem.

\begin{theorem}For the Poisson $\is$graph with secrecy rate threshold~$\th$,
noise powers~$\WM,\WE$, and channel gain function~$g(r)=\frac{1}{r^{2b}}$,
the average node degrees are\begin{align}
\mathbb{E}\{N_{\mathrm{in}}\}=\mathbb{E}\{N_{\mathrm{out}}\} & =\pi^{2}\LM\LE\int_{0}^{\infty}\frac{xe^{-\pi\LE x}}{\left(\frac{\WM}{\WE}2^{\th}+\frac{\WM}{\PM}(2^{\th}-1)x^{b}\right)^{1/b}}dx\label{eq:EN-thetaW-ub-exact}\\
 & \leq\frac{\LM}{\LE}\frac{1}{\left(\frac{\WM}{\WE}2^{\th}+\frac{\WM}{\PM(\pi\LE)^{b}}(2^{\th}-1)\right)^{1/b}}.\label{eq:EN-thetaW-ub-bound}\end{align}
 \end{theorem}

\begin{proof}We consider the process~$\PiM\cup\{0\}$ obtained by
adding a legitimate node to the origin of the coordinate system, and
denote the out-degree of the node at the origin by $N_{\mathrm{out}}$.
Let $\RE 1\triangleq\min_{e_{i}\in\PiE}|e_{i}|$ be the random distance
between the origin and its closest eavesdropper. Define the function\begin{equation}
\psi(r)\triangleq\frac{r}{\left(\frac{\WM}{\WE}2^{\th}+\frac{\WM}{\PM}(2^{\th}-1)r^{2b}\right)^{1/2b}},\quad r\geq0,\label{eq:f-ub}\end{equation}
so that (\ref{eq:edges-thetaW-ub}) can simply be written as $\mathcal{E}=\{\overrightarrow{x_{i}x_{j}}:|x_{i}-x_{j}|<\ff(|x_{i}-e^{*}|)\}$.
This function is depicted in Figure~\ref{fig:out-degree-theta}.
The average out-degree is then given by\begin{align*}
\mathbb{E}\{N_{\mathrm{out}}\} & =\mathbb{E}_{\PiM,\RE 1}\{\PiM\{\mathcal{B}_{0}(\ff(\RE 1))\}\}\\
 & =\pi\LM\mathbb{E}_{\RE 1}\{\ff^{2}(\RE 1)\}\end{align*}
Defining $X\triangleq\RE 1^{2}$, we can write\begin{align}
\mathbb{E}\{N_{\mathrm{out}}\} & =\pi\LM\mathbb{E}_{X}\left\{ \frac{X}{\left(\frac{\WM}{\WE}2^{\th}+\frac{\WM}{\PM}(2^{\th}-1)X^{b}\right)^{1/b}}\right\} \label{eq:concave-temp}\\
 & =\pi\LM\int_{0}^{\infty}\frac{x}{\left(\frac{\WM}{\WE}2^{\th}+\frac{\WM}{\PM}(2^{\th}-1)x^{b}\right)^{1/b}}\pi\LE e^{-\pi\LE x}dx,\nonumber \end{align}
where we used the fact that $X$ is an exponential RV with mean~$\frac{1}{\pi\LE}$.
This proves the result in (\ref{eq:EN-thetaW-ub-exact}). To obtain
the upper bound, we note that the function inside the expectation
in (\ref{eq:concave-temp}) is concave in $x$, and apply Jensen's
inequality as follows\begin{align*}
\mathbb{E}\{N_{\mathrm{out}}\} & \leq\frac{\LM}{\LE}\frac{1}{\left(\frac{\WM}{\WE}2^{\th}+\frac{\WM}{\PM(\pi\LE)^{b}}(2^{\th}-1)\right)^{1/b}}.\end{align*}
This is the result in (\ref{eq:EN-thetaW-ub-bound}). Noting that
$\mathbb{E}\{N_{\mathrm{in}}\}=\mathbb{E}\{N_{\mathrm{out}}\}$ for
any directed random graph, the proof is concluded.\end{proof}

\subsection{Numerical Results}

Figure~\ref{fig:pmf-degrees} compares the PMFs of the in- and out-degree
of a node. We clearly observe that the RV~$N_{\mathrm{in}}$ does
not have a geometric distribution, unlike the RV~$N_{\mathrm{out}}$.
However, the two RVs have the same mean~$\frac{\LM}{\LE}$, according
to Property~\ref{prop:ENin-ENout}.

Figure~\ref{fig:p-isol} compares the probabilities of out-isolation
and in-isolation of a node for various ratios~$\frac{\LE}{\LM}$.
The curve for $p_{\mathrm{out-isol}}$ was plotted using the closed
form expression in (\ref{eq:p-out-isol}). The curve for $p_{\mathrm{in-isol}}$
was obtained according to (\ref{eq:p-in-isol}) through Monte Carlo
simulation of the random area~$\widetilde{A}$ of a typical Voronoi
cell, induced by a unit-density Poisson process. We observe that $p_{\mathrm{in-isol}}<p_{\mathrm{out-isol}}$
for any fixed~$\frac{\LE}{\LM}$, as proved in Property~\ref{prop:pin-pout-isol}.

Figure~\ref{fig:theta-PW-plot} illustrates the effect of the secrecy
rate threshold~$\th$ on the average node degrees. For the case of
$g(r)=\frac{1}{r^{2b}}$ in particular, it compares the exact value
of $\mathbb{E}\{N_{\mathrm{out}}\}$ given in (\ref{eq:EN-thetaW-ub-exact})
with its upper bound in (\ref{eq:EN-thetaW-ub-bound}). We observe
that the average node degree attains its maximum value of $\frac{\LM}{\LE}=10$
at $\th=0$, and is monotonically decreasing with $\th$. As proved
in Property~\ref{prop:ENin-ENout-thetaW}, such behavior occurs for
any function~$g(r)$ satisfying the conditions in Section~\ref{sub:Propagation}.
Furthermore, we can show that the upper bound is asymptotically tight
-- in the sense that the difference between the exact average node
degree and its upper bound approaches $0$ -- in the following two
extreme cases: 
\begin{itemize}
\item $\th\rightarrow0$: In this regime, both (\ref{eq:EN-thetaW-ub-exact})
and (\ref{eq:EN-thetaW-ub-bound}) approach $\frac{\LM}{\LE}\left(\frac{\WE}{\WM}\right)^{1/b}$,
and thus the bound is asymptotically tight.
\item $\PM\rightarrow\infty$: In this high-SNR regime, both (\ref{eq:EN-thetaW-ub-exact})
and (\ref{eq:EN-thetaW-ub-bound}) converge to $\frac{\LM}{\LE}\left(\frac{\WE}{\WM}2^{-\th}\right)^{1/b}$,
and thus the bound is asymptotically tight.
\end{itemize}

\section{Techniques for Communication with Enhanced Secrecy\label{sec:Techniques}}

Based on the results derived in Section~\ref{sec:Node-Degrees},
we observe that even a small density of eavesdroppers is enough to
significantly disrupt connectivity of the $\is$graph. For example,
if the density of legitimate nodes is half the density of eavesdroppers,
then from (\ref{eq:ENin-ENout}) the average node degree is reduced
to $2$. In this section, we explore two techniques for communication
with enhanced secrecy: i)~\emph{sectorized transmission}, whereby
 each legitimate node is able to transmit independently in $L$~sectors
of the plane (e.g.,~through the use of directional antennas); and
ii)~\emph{eavesdropper neutralization}, whereby legitimate nodes
are able to physically monitor its surrounding area and guarantee
that there are no eavesdroppers inside  a neutralization\emph{ }region~$\Theta$
(e.g.,~ by neutralizing such eavesdroppers). For these two techniques,
we quantify the improvements in terms of the resulting average node
degree of the $\is$graph.

\subsection{Sectorized Transmission\label{sub:Sectorized-Transmission}}

We have so far assumed that the legitimate nodes employ omnidirectional
antennas, distributing power equally among all directions. We now
consider that each legitimate node is able to transmit independently
in $L$~sectors of the plane, with $L\geq1$. This can be accomplished,
for example, through the use of $L$~directional antennas. In this
section, we characterize the impact of the number of sectors~$L$
on the local connectivity of the $\is$graph.

With each node~$x_{i}\in\Pi$, we associate $L$~transmission sectors~$\{\mathcal{S}_{i}^{(l)}\}_{l=1}^{L}$
, defined as\[
\mathcal{S}_{i}^{(l)}\triangleq\left\{ z\in\mathbb{R}^{2}:\phi_{i}+(l-1)\frac{2\pi}{L}<\angle\overrightarrow{x_{i}z}<\phi_{i}+l\frac{2\pi}{L}\right\} ,\quad l=1\ldots L,\]
where $\{\phi_{i}\}_{i=1}^{\infty}$ are random offset angles with
an arbitrary joint distribution. The resulting $\is$graph~$G_{L}=\{\PiM,\mathcal{E}_{L}\}$
has an edge set given by\begin{equation}
\mathcal{E}_{L}=\Bigl\{\overrightarrow{x_{i}x_{j}}:|x_{i}-x_{j}|<|x_{i}-e^{*}|,\quad e^{*}=\underset{e_{k}\in\PiE\cap\mathcal{S}^{*}}{\mathrm{argmin}}\:|x_{i}-e_{k}|,\quad\mathcal{S}^{*}=\{\mathcal{S}_{i}^{(l)}:x_{j}\in\mathcal{S}_{i}^{(l)}\}\Bigr\}.\label{eq:edges-sectors}\end{equation}
Here, $\mathcal{S}^{*}$ is the transmission sector of $x_{i}$ that
contains the destination node~$x_{j}$, and $e^{*}$ is the eavesdropper
inside $\mathcal{S}^{*}$ that is closest to the transmitter~$x_{i}$.
Then, the secure link~$\overrightarrow{x_{i}x_{j}}$ exists if and
only if $x_{j}$ is closer to $x_{i}$ than any other eavesdropper
inside the same transmission sector where the destination~$x_{j}$
is located. We start by characterizing the distribution of the out-degree
by the following theorem.

\begin{theorem}\label{thm:p-Nout-sectors}For the Poisson $\is$graph~$G_{L}$
with $L$~sectors, the out-degree~$N_{\mathrm{out}}$ of a node
has the following negative binomial PMF\begin{equation}
p_{N_{\mathrm{out}}}(n)={L+n-1 \choose L-1}\left(\frac{\LM}{\LM+\LE}\right)^{n}\left(\frac{\LE}{\LM+\LE}\right)^{L},\quad n\geq0.\label{eq:p-Nout-sectors}\end{equation}
\end{theorem}

\begin{proof}We consider the process~$\PiM\cup\{0\}$ obtained by
adding a legitimate node to the origin of the coordinate system, and
denote the out-degree of the node at the origin by $N_{\mathrm{out}}$.
 This is depicted in Fig.~\ref{fig:out-degree-sectors}. Consider
the set of legitimate nodes in the sector~$\mathcal{S}^{(l)}$. Let
$\{\RM i^{(l)}\}_{i=1}^{\infty}$ be the distances (not necessarily
ordered) from these legitimate nodes and the origin, such that $\RM i^{(l)}=|x_{i}^{(l)}|$,
with $\{x_{i}^{(l)}\}=\PiM\cap\mathcal{S}^{(l)}$. For the eavesdroppers,
we similarly define $\{\RE i^{(l)}\}_{i=1}^{\infty}$, such that $\RE i^{(l)}=|e_{i}^{(l)}|$
with $\{e_{i}^{(l)}\}=\PiE\cap\mathcal{S}^{(l)}$. Because the sectors~$\mathcal{S}^{(l)}$
are non-overlapping and $\PiM$ is Poisson, the processes~$\{\RM i^{(l)}\}_{i=1}^{\infty}$
are independent for different~$l$ (a similar argument can be made
for the independence of $\{\RE i^{(l)}\}_{i=1}^{\infty}$ for different~$l$).
As a result, we can analyze the out-degrees of node~$0$ in each
sector, and add these independent RVs to obtain the total out-degree.
Specifically,\begin{equation}
N_{\mathrm{out}}=\sum_{l=1}^{L}N_{\mathrm{out}}^{(l)},\label{eq:Nout-sum}\end{equation}
where the RVs\[
N_{\mathrm{out}}^{(l)}\triangleq\#\{\RM i^{(l)}:\RM i^{(l)}<\min_{k}\RE k^{(l)}\},\]
are IID in $l$. 

From the mapping theorem, we know that $\{(\RM i^{(l)})^{2}\}_{i=1}^{\infty}$
and $\{(\RE i^{(l)})^{2}\}_{i=1}^{\infty}$ are homogeneous Poisson
processes with rates~$\frac{\pi\LM}{L}$ and $\frac{\pi\LE}{L}$,
respectively. Following the steps analogous to the proof of Theorem~\ref{thm:p-Nout},
we can show that each RV~$N_{\mathrm{out}}^{(l)}$ has the geometric
PMF~ $p_{N_{\mathrm{out}}^{(l)}}(n)=p^{n}(1-p),$ $n\geq0$, with
parameter~$p=\frac{\LM}{\LM+\LE}$. In other words, each RV~$N_{\mathrm{out}}^{(l)}$
has the same distribution of the total out-degree with\emph{ $L=1$}.
The PMF of $N_{\mathrm{out}}$ with $L$~sectors can be obtained
through convolution of the individual PMFs~$p_{N_{\mathrm{out}}^{(l)}}$,
and results in a negative binomial PMF\ with $L$~degrees of freedom
having the same parameter~$p$, i.e., $p_{N_{\mathrm{out}}}(n)={L+n-1 \choose L-1}p^{n}(1-p)^{L},$
$n\geq0$, with $p=\frac{\LM}{\LM+\LE}$. This is the result in (\ref{eq:p-Nout-sectors})
and the proof is completed.\end{proof}

When $L=1$, (\ref{eq:p-Nout-sectors}) reduces to the PMF\ without
sectorization given in (\ref{eq:p-Nout}), as expected. The above
theorem directly gives the average node degrees as a function of~$L$,
as given in the following corollary.

\begin{corollary}For the Poisson $\is$graph~$G_{L}$ with $L$~sectors,
the average node degrees are\begin{equation}
\mathbb{E}\{N_{\mathrm{in}}\}=\mathbb{E}\{N_{\mathrm{out}}\}=L\frac{\LM}{\LE}.\label{eq:ENin-ENout-sectors}\end{equation}
\end{corollary}

\begin{proof}Using (\ref{eq:Nout-sum}), we have that $\mathbb{E}\{N_{\mathrm{out}}\}=L\mathbb{E}\{N_{\mathrm{out},l}\}=L\frac{p}{1-p}$,
with $p=\frac{\LM}{\LM+\LE}$. In addition, we have that $\mathbb{E}\{N_{\mathrm{in}}\}=\mathbb{E}\{N_{\mathrm{out}}\}$
for any directed random graph, and (\ref{eq:ENin-ENout-sectors})
follows.\end{proof}

We conclude that the expected node degrees increases \emph{linearly}
with the number of sectors~$L$, and hence sectorized transmission
is an effective technique for enhancing the secrecy of communications.
Figure~\ref{fig:out-degree-sectors} provides an intuitive understanding
of why sectorization works. Specifically, if there was no sectorization,
node~$0$ would be out-isolated, due to the close proximity of the
eavesdropper in sector~$\mathcal{S}^{(4)}$. However, if we allow
independent transmissions in $4$ non-overlapping sectors, that same
eavesdropper can only hear the transmissions inside sector~$\mathcal{S}^{(4)}$.
Thus, even though node~$0$ is out-isolated with respect to sector~$\mathcal{S}^{(4)}$,
it may still communicate securely with legitimate nodes in sectors~$\mathcal{S}^{(1)},$
$\mathcal{S}^{(2)},$ and $\mathcal{S}^{(3)}$.

\subsection{Eavesdropper Neutralization\label{sub:Eavesdropper-Neutralization}}

In some scenarios, the legitimate nodes may be able to physically
inspect its surrounding area and guarantee that there are no eavesdroppers
inside a \emph{neutralization region}~$\Theta$ (for example, by
deactivating such eavesdroppers). In this section, we characterize
the impact of such region on the local connectivity of node.

With each node~$x_{i}\in\PiM$, we associate a neutralization set~$\Theta_{i}$
around $x_{i}$ that is guaranteed to be free of eavesdroppers. The
\emph{total neutralization region~}$\Theta$ can then be seen as
a Boolean model with points~$\{x_{i}\}$ and associated sets~$\{\Theta_{i}\}$,
i.e.,%
\footnote{In other fields such as materials science, the points~$\{x_{i}\}$
are also called \emph{germs}, and the sets~$\{\Theta_{i}\}$ are
also called \emph{grains}.%
}\[
\Theta=\bigcup_{i=1}^{\infty}(x_{i}+\Theta_{i}).\]
Since the homogeneous Poisson process~$\PiM$ is stationary, it follows
that $\Theta$ is also stationary, in the sense that its distribution
is translation-invariant. Since eavesdroppers cannot occur inside
$\Theta$, the \emph{effective eavesdropper process} after neutralization
is $\PiE\cap\overline{\Theta}$, where $\overline{\Theta}\triangleq\mathbb{R}^{2}\backslash\Theta$
denotes the complement of $\Theta$.%
\footnote{In the materials science literature, $\Theta$ is typically referred
to as the \emph{occupied region}, since it is occupied by grains.
In our problem, however, $\Theta$ corresponds to a \emph{vacant region},
in the sense that it is free of eavesdroppers. To prevent confusion
with the literature, we avoid the use of the terms {}``occupied''
and {}``vacant'' altogether.%
} The resulting $\is$graph~$G_{\Theta}=\{\PiM,\mathcal{E}_{\Theta}\}$
has an edge set given by\begin{equation}
\mathcal{E}_{\Theta}=\Bigl\{\overrightarrow{x_{i}x_{j}}:|x_{i}-x_{j}|<|x_{i}-e^{*}|,\quad e^{*}=\underset{e_{k}\in\PiE\cap\overline{\Theta}}{\mathrm{argmin}}\:|x_{i}-e_{k}|\Bigr\}\label{eq:edges-exclusion}\end{equation}
i.e.,~the secure link~$\overrightarrow{x_{i}x_{j}}$ exists if and
only if $x_{j}$ is closer to $x_{i}$ than any other eavesdropper
that has not been neutralized. Since $\PiE\cap\overline{\Theta}\subseteq\PiE$,
it is intuitively obvious that eavesdropper neutralization improves
the local connectivity, and that such improvement is monotonic with
the area of the neutralization set~$\Theta_{i}$. In the following,
we consider the case of a circular neutralization set, i.e,~$\Theta_{i}=\mathcal{B}_{0}(\rho)$,
where $\rho$ is a deterministic \emph{neutralization radius}. We
denote the corresponding $\is$graph by $G_{\rho}$. Even in this
simple scenario, the full distributions of the corresponding node
degrees~$N_{\mathrm{in}}$ and $N_{\mathrm{out}}$ are difficult
to obtain, since the underlying process~$\PiE\cap\overline{\Theta}$
is quite complex to characterize. However, it is easier to carry out
an analysis of the first order moments, namely of $\mathbb{E}\{N_{\mathrm{out}}\}$.
We can use this metric to compare eavesdropper neutralization with
the other techniques discussed in this paper, in terms of their effectiveness
in enhancing security. The following theorem provides the desired
result.

\begin{theorem}For the Poisson $\is$graph~$G_{\rho}$ with neutralization
radius~$\rho$, the average node degrees are lower-bounded by \begin{equation}
\mathbb{E}\{N_{\mathrm{in}}\}=\mathbb{E}\{N_{\mathrm{out}}\}\geq\frac{\LM}{\LE}\left(\pi\LE\rho^{2}+e^{\pi\LM\rho^{2}}\right).\label{eq:ENin-ENout-exclusion}\end{equation}
\end{theorem}

\begin{proof}We consider the process~$\PiM\cup\{0\}$ obtained by
adding a legitimate node to the origin of the coordinate system, and
denote the out-degree of the node at the origin by $N_{\mathrm{out}}$.
 This is depicted in Fig.~\ref{fig:out-degree-exclusion}. Let $\RE 1\triangleq\min_{e_{k}\in\PiE\cap\overline{\Theta}}|e_{k}|$
be the random distance between the first non-neutralized eavesdropper
and the origin. Let $\An(a,b)\triangleq\{x\in\mathbb{R}^{2}:a\leq|x|\leq b\}$
denote the annular region between radiuses~$a$ and $b$, and $\mathbb{A}\{\mathcal{R}\}$
denote the area of the arbitrary region~$\mathcal{R}$. Noting that\begin{align*}
N_{\mathrm{out}} & =\sum_{x_{i}\in\PiM}\mathbbm{1}\{|x_{i}|<\RE 1\}\\
 & ={\int\intop}_{\mathbb{R}^{2}}\mathbbm{1}\{|x|<\RE 1\}\PiM(dx),\end{align*}
we can use Fubini's theorem to write\begin{align}
\mathbb{E}\{N_{\mathrm{out}}\} & =\LM{\int\intop}_{\mathbb{R}^{2}}\mathbb{P}_{x}\{|x|<\RE 1\}dx\nonumber \\
 & =\LM\pi\rho^{2}+\LM{\int\intop}_{\An(\rho,\infty)}\mathbb{P}_{x}\{|x|<\RE 1\}dx,\label{eq:ENout-temp}\end{align}
where $\mathbb{P}_{x}\{\cdot\}$ denotes the Palm probability associated
with point~$x$ of process~$\PiM$.%
\footnote{Informally, the Palm probability~$\mathbb{P}_{x}\{\cdot\}$ can be
interpreted as the conditional probability~$\mathbb{P}\{\cdot|x\in\PiM\}$.
Since the conditioning event has probability zero, such conditional
probability is ambiguous without further explanation. Palm theory
makes this notion mathematically precise (see \cite[Sec. 4.4]{StoKenMec:95}
for a detailed treatment).%
} Appendix~\ref{sec:Deriv-Px-bound} shows that the integrand above
satisfies\begin{equation}
\mathbb{P}_{x}\{|x|<\RE 1\}\geq\exp\left(-\pi\LE e^{-\LM\pi\rho^{2}}(|x|^{2}-\rho^{2})\right).\label{eq:Px-bound}\end{equation}
Replacing (\ref{eq:Px-bound}) into (\ref{eq:ENout-temp}), we obtain\begin{align*}
\mathbb{E}\{N_{\mathrm{out}}\} & \geq\LM\pi\rho^{2}+\LM{\int\intop}_{\An(\rho,\infty)}\exp\left(-\pi\LE e^{-\LM\pi\rho^{2}}(|x|^{2}-\rho^{2})\right)dx\\
 & =\LM\pi\rho^{2}+\frac{\LM}{\LE}e^{\LM\pi\rho^{2}}.\end{align*}
Rearranging terms and noting that $\mathbb{E}\{N_{\mathrm{in}}\}=\mathbb{E}\{N_{\mathrm{out}}\}$
for any directed random graph, we obtain the desired result in (\ref{eq:ENin-ENout-exclusion}).
This concludes the proof.\end{proof}

We conclude that the expected node degrees increases at a rate that
is at least \emph{exponential} with the neutralization radius~$\rho$,
making eavesdropper neutralization an effective technique for enhancing
secure connectivity. Such exponential dependence is intimately tied
to the fact that the \emph{fractional area}~$p_{\Theta}=1-e^{-\LM\pi\rho^{2}}$
of the neutralization region~$\Theta$ also approaches $1$ exponentially
as $\rho$ increases.

\subsection{Numerical Results}

Figure~\ref{fig:exclusion-lambdaE-plot} illustrates effectiveness
of eavesdropper neutralization in enhancing secure connectivity. In
particular, it plots the average node degree versus the neutralization
radius~$\rho$, for various values of $\LE$. We observe that $\mathbb{E}\{N_{\mathrm{out}}\}$
increases at a rate that is at least exponential with the neutralization
radius~$\rho$, as expected from (\ref{eq:ENin-ENout-exclusion}).
Furthermore, the analytical lower-bound is in general very close to
the simulated value of $\mathbb{E}\{N_{\mathrm{out}}\}$, and becomes
tight in the following two asymptotic cases:
\begin{itemize}
\item $\rho\rightarrow0$: In this regime, the neutralization region vanishes,
and therefore $\mathbb{E}\{N_{\mathrm{out}}\}\rightarrow\frac{\LM}{\LE}$,
as given in (\ref{eq:ENin-ENout}). Since right side of (\ref{eq:ENin-ENout-exclusion})
also approaches $\frac{\LM}{\LE}$ as $\rho\rightarrow0$, the bound
is asymptotically tight.
\item $\LE\rightarrow\infty$: In this regime, an eavesdropper will occurs
a.s.\ at a distance close to $\rho$ from the origin. As a result,
$\mathbb{E}\{N_{\mathrm{out}}\}$ is approaches the expected number
of legitimate nodes inside the ball~$\mathcal{B}_{0}(\rho)$, i.e.,~$\LM\pi\rho^{2}$.
Since right side of (\ref{eq:ENin-ENout-exclusion}) also approaches
$\LM\pi\rho^{2}$ as $\LE\rightarrow\infty$, the bound is asymptotically
tight.
\end{itemize}

\section{Maximum Secrecy Rate in the Poisson $\is$Graph\label{sec:Secrecy-Capacity}}

In this section, we analyze the MSR between a node and each of its
neighbours, as well as the probability of existence of a non-zero
MSR, and the probability of secrecy outage\textit{\emph{.}}

\subsection{Distribution of the Maximum Secrecy Rate\label{sub:Distribution-Cs}}

Considering the coordinate system depicted in Fig.~\ref{fig:out-degree}
and using (\ref{eq:Cs-ij-pathloss}), the MSR~$\cs i$ between the
node at the origin and its $i$-th closest neighbour, $i\geq1$, can
be written for a given realization of the node positions~$\PiM$
and $\PiE$ as\begin{equation}
\cs i=\left[\log_{2}\left(1+\frac{\PM}{\RM i^{2b}\W}\right)-\log_{2}\left(1+\frac{\PM}{\RE 1^{2b}\W}\right)\right]^{+},\label{eq:Csi}\end{equation}
in bits per complex dimension. For each instantiation of the random
Poisson processes~$\PiM$ and $\PiE$, a realization of the RV~$\cs i$
is obtained. The following theorem provides the distribution of this
RV.

\begin{theorem}The MSR~$\cs i$ between a typical node and its $i$-th
closest neighbour, $i\geq1$, is a RV~whose cumulative distribution
function (CDF)~$F_{\cs i}(\th)$ is given by\begin{multline}
F_{\cs i}(\th)=1-\frac{\ln2(\pi\LM)^{i}}{(i-1)!b}\left(\frac{\PM}{\W}\right)^{\frac{i}{b}}\\
\times\int_{\th}^{+\infty}\frac{2^{z}}{(2^{z}-1)^{1+\frac{i}{b}}}\exp\left(-\pi\LM\left(\frac{\frac{\PM}{\W}}{2^{z}-1}\right)^{\frac{1}{b}}-\pi\LE\left(\frac{\frac{\PM}{\W}}{2^{z-\th}-1}\right)^{\frac{1}{b}}\right)dz,\label{eq:F-Csi}\end{multline}
for $\th\geq0$.\end{theorem}

\begin{proof}The MSR~$\cs i$ in (\ref{eq:Csi}) can be expressed
as $\cs i=\left[\cM i-\CE\right]^{+}$, where $\cM i=\log_{2}\left(1+\frac{\PM}{\RM i^{2b}\W}\right)$
and $\CE=\log_{2}\left(1+\frac{\PM}{\RE 1^{2b}\W}\right).$ The RV~$\cM i$
is a transformation of the RV~$X_{i}\triangleq\RM i^{2}$ through
the monotonic function~$g(x)=\log_{2}\left(1+\frac{\PM}{x^{b}\W}\right)$,
and thus its PDF is given by the rule~$f_{\CE}(\th)=\left.\frac{1}{|g^{\prime}(x)|}f_{X_{i}}(x)\right|_{x=g^{-1}(\th)}$.
Note that the sequence~$\{X_{i}\}_{i=1}^{\infty}$ represents Poisson
arrivals \emph{on the line} with the constant arrival rate~$\pi\LM$,
as can be easily shown using the mapping theorem~\cite[Section 2.3]{Kin:93}.
Therefore, the RV~$X_{i}$ has an Erlang distribution of order~$i$
with rate~$\pi\LM$, and its PDF is given by\[
f_{X_{i}}(x)=\frac{(\pi\LM)^{i}x^{i-1}e^{-\pi\LM x}}{(i-1)!},\quad x\geq0.\]
Then, applying the above rule, $f_{\cM i}(\th)$ can be shown to be\begin{equation}
f_{\cM i}(\th)=\ln2\frac{(\pi\LM)^{i}}{(i-1)!b}\left(\frac{\PM}{\W}\right)^{\frac{i}{b}}\frac{2^{\th}}{(2^{\th}-1)^{1+\frac{i}{b}}}\exp\left(-\pi\LM\left(\frac{\frac{\PM}{\W}}{2^{\th}-1}\right)^{\frac{1}{b}}\right),\quad\th\geq0.\label{eq:f-CMi}\end{equation}
Replace $\LM$ with $\LE$ and setting $i=1$, we obtain the PDF of
$\CE$ as\begin{equation}
f_{\CE}(\th)=\ln2\frac{\pi\LE}{b}\left(\frac{\PM}{\W}\right)^{\frac{1}{b}}\frac{2^{\th}}{(2^{\th}-1)^{1+\frac{1}{b}}}\exp\left(-\pi\LE\left(\frac{\frac{\PM}{\W}}{2^{\th}-1}\right)^{\frac{1}{b}}\right),\quad\th\geq0.\label{eq:f-CE}\end{equation}
Since the sequences~$\{\RM i\}_{i=1}^{\infty}$ and $\{\RE i\}_{i=1}^{\infty}$
are mutually independent, so are the RVs~$\cM i$ and $\CE$. This
implies that CDF of $\cs i=\left[\cM i-\CE\right]^{+}$ can be obtained
through convolution of $f_{\cM i}(\th)$ and $f_{\CE}(\th)$ as\begin{align}
F_{\cs i}(\th) & =\mathbb{P}\left\{ \left[\cM i-\CE\right]^{+}\leq\th\right\} \nonumber \\
 & =1-\mathbb{P}\{\cM i-\CE>\th\}\nonumber \\
 & =1-\int_{\th}^{\infty}f_{\cM i}(z)*f_{\CE}(-z)\, dz,\label{eq:F-Csi-proof}\end{align}
for $\th\geq0$. Replacing (\ref{eq:f-CMi}) and (\ref{eq:f-CE})
into (\ref{eq:F-Csi-proof}), we obtain after some algebra\begin{multline*}
F_{\cs i}(\th)=1-\frac{\ln2(\pi\LM)^{i}}{(i-1)!b}\left(\frac{\PM}{\W}\right)^{\frac{i}{b}}\\
\times\int_{\th}^{+\infty}\frac{2^{z}}{(2^{z}-1)^{1+\frac{i}{b}}}\exp\left(-\pi\LM\left(\frac{\frac{\PM}{\W}}{2^{z}-1}\right)^{\frac{1}{b}}-\pi\LE\left(\frac{\frac{\PM}{\W}}{2^{z-\th}-1}\right)^{\frac{1}{b}}\right)dz,\end{multline*}
for $\th\geq0$. This is the result in (\ref{eq:F-Csi}) and the proof
is concluded.\end{proof}

\subsection{Existence and Outage of the Maximum Secrecy Rate}

Based on the results of Section~\ref{sub:Distribution-Cs}, we can
now obtain the probability of existence of a non-zero MSR, and the
probability of secrecy outage.\textit{\emph{ The following corollary
provides such probabilities.}}

\begin{corollary}Considering the link between a typical node and
its $i$-th closest neighbour, $i\geq1$, the probability of \emph{existence}
of a non-zero MSR, $p_{\mathrm{exist},i}=\mathbb{P}\{\cs i>0\}$,
is given by\begin{equation}
p_{\mathrm{exist},i}=\left(\frac{\LM}{\LM+\LE}\right)^{i}.\label{eq:Pexist}\end{equation}
and the probability of an \emph{outage} in MSR, $p_{\mathrm{outage},i}(\th)=\mathbb{P}\{\cs i<\th\}$
for $\th>0$, is given by\begin{multline}
p_{\mathrm{outage},i}(\th)=1-\frac{\ln2(\pi\LM)^{i}}{(i-1)!b}\left(\frac{\PM}{\W}\right)^{\frac{i}{b}}\\
\times\int_{\th}^{+\infty}\frac{2^{z}}{(2^{z}-1)^{1+\frac{i}{b}}}\exp\left(-\pi\LM\left(\frac{\frac{\PM}{\W}}{2^{z}-1}\right)^{\frac{1}{b}}-\pi\LE\left(\frac{\frac{\PM}{\W}}{2^{z-\th}-1}\right)^{\frac{1}{b}}\right)dz\label{eq:Poutage}\end{multline}
\end{corollary}

\begin{proof}To obtain (\ref{eq:Pexist}), we note that the event~$\{\cM i>\CE\}$
is equivalent to $\{N_{\mathrm{out}}\geq i\}$. Thus, we use (\ref{eq:p-Nout})
to write\begin{align*}
p_{\mathrm{exist},i} & =\mathbb{P}\{\cM i>\CE\}\\
 & =\sum_{n=i}^{\infty}\left(\frac{\LM}{\LM+\LE}\right)^{n}\left(\frac{\LE}{\LM+\LE}\right)\\
 & =\left(\frac{\LM}{\LM+\LE}\right)^{i}.\end{align*}
The expression for $p_{\mathrm{outage}}(\th)$ follows directly from
(\ref{eq:F-Csi}).\end{proof}

\subsection{Numerical Results}

Figure~\ref{fig:pexist-lambda} shows the probability~$p_{\mathrm{exist},i}$
of existence of a non-zero MSR from a typical node to its $i$-th
neighbour, as a function of the eavesdropper density~$\LE$. It can
be seen that the existence of a non-zero MSR~$\cs i$ to any neighbour~$i$
becomes more likely as the value of $\LE$ increases. Furthermore,
since $\RM 1\leq\RM 2\leq\ldots$, as the value of $i$ increases,
the $i$-th neighbour becomes further away, and the corresponding
$p_{\mathrm{exist},i}$ decreases.

Figure~\ref{fig:cdf-csi} shows the probability~$p_{\mathrm{outage},i}$
of secrecy outage of a typical node transmitting to its $i$-th neighbour,
as a function of the desired secrecy rate~$\th$. As expected, a
secrecy outage become more likely as we increase the target secrecy
rate~$\th$ set by the transmitter.

\section{The Case of Colluding Eavesdroppers\label{sec:Secrecy-Colluding}}

We now aim to study the effect of colluding eavesdroppers on the secrecy
of communications. In order to focus on the effect of eavesdropper
collusion on the MSR of the legitimate link, we first consider in
Sections~\ref{sec:Cs-Colluding} to \ref{sub:Comparison-Colluding}
a \emph{single} legitimate link with deterministic length~$\rM$
in the presence of a random process~$\PiE$. Such simplification
eliminates the randomness associated with the position of the legitimate
nodes. We then consider both random processes~$\PiM$ and $\PiE$
in Section~\ref{sub:Collusion-SGraph}, and characterize the average
node degree in the presence of eavesdropper collusion.

\subsection{Maximum Secrecy Rate of a Single Link\label{sec:Cs-Colluding}}

We consider the scenario depicted in Fig.~\ref{fig:collusion}, where
a legitimate link is composed of two nodes: one transmitter located
at the origin (Alice), and one receiver located at a deterministic
distance~$\rM$ from the origin (Bob). The eavesdroppers have ability
to \emph{collude}, i.e.,~they can exchange and combine the information
received by all the eavesdroppers to decode the secret message. The
eavesdroppers are scattered in the two-dimensional plane according
to an \emph{arbitrary} spatial process~$\PiE$, and their distances
to the origin are denoted by $\{\RE i\}_{i=1}^{\infty}$, where $\RE 1\leq\RE 2\leq\ldots$.

Since the colluding eavesdroppers may gather the received information
and send it to a central processor, the scenario depicted in Fig.~\ref{fig:collusion}
can be viewed as a SIMO Gaussian wiretap channel depicted in Fig.~\ref{fig:SIMO}.
Here, the input is the signal transmitted by Alice, and the output
of the wiretap channel is the collection of signals received by all
the eavesdroppers. We consider that Alice sends a symbol~$x\in\mathbb{C}$
with power constraint~$\mathbb{E}\{|x|^{2}\}\leq\PM$. The vectors~$\mathbf{h}_{\M}\in\mathbb{C}^{m}$
and $\mathbf{h}_{\mathrm{e}}\in\mathbb{C}^{n}$ represent, respectively,
the gains of the legitimate and eavesdropper channels.%
\footnote{We use boldface letters to denote vectors and matrices.%
} The noise is represented by the vectors~$\mathbf{w}_{\M}\in\mathbb{C}^{m}$
and $\mathbf{w}_{\mathrm{e}}\in\mathbb{C}^{n}$, which are considered
to be mutually independent Gaussian RVs with zero mean and non-singular
covariance matrices~$\boldsymbol{\Sigma}_{\M}$ and $\boldsymbol{\Sigma}_{\mathrm{e}}$,
respectively. The system of Fig.~\ref{fig:SIMO} can then be summarized
as\begin{align}
\mathbf{y}_{\M} & =\mathbf{h}_{\M}x+\mathbf{w}_{\M}\label{eq:yM-simo}\\
\mathbf{y}_{\mathrm{e}} & =\mathbf{h}_{\mathrm{e}}x+\mathbf{w}_{\mathrm{e}}.\label{eq:yE-simo}\end{align}
The scenario of interest can be obtained from the SIMO Gaussian wiretap
channel in Fig.~\ref{fig:SIMO} by appropriate choice of the parameters~$\mathbf{h}_{\M}$,
$\mathbf{h}_{\mathrm{e}}$, $\boldsymbol{\Sigma}_{\M}$, and $\boldsymbol{\Sigma}_{\mathrm{e}}$.

In this section, we determine the MSR of the legitimate link, in the
presence of colluding eavesdroppers scattered in the plane according
to an arbitrary spatial process. The result is given in the following
theorem.

\begin{theorem}\label{thm:Cs-Colluding}For a given realization of
the arbitrary eavesdropper process~$\PiE$, the MSR of the legitimate
link is given by\begin{equation}
\Cs=\left[\log_{2}\left(1+\frac{\PM\cdot g(\rM)}{\WM}\right)-\log_{2}\left(1+\frac{\PM\sum_{i=1}^{\infty}g(\RE i)}{\WE}\right)\right]^{+},\label{eq:Cs-collusion}\end{equation}
where $\PM\sum_{i=1}^{\infty}g(\RE i)\triangleq\PE$ is the aggregate
power received by all the eavesdroppers.

\end{theorem}

\begin{proof}For a given realization of the channels~$\mathbf{h}_{\M}$
and $\mathbf{h}_{\E}$, it can be shown~\cite{TseVis:05} that $\widetilde{y}_{\M}=\mathbf{h}_{\M}^{\dagger}\boldsymbol{\Sigma}_{\M}^{-1}\mathbf{y}_{\M}$
and $\widetilde{y}_{\E}=\mathbf{h}_{\E}^{\dagger}\boldsymbol{\Sigma}_{\E}^{-1}\mathbf{y}_{\E}$
are sufficient statistics to estimate $x$ from the corresponding
observations~$\mathbf{y}_{\M}$ and $\mathbf{y}_{\E}$.%
\footnote{We use $\dagger$ to denote the conjugate transpose operator.%
} Since sufficient statistics preserve mutual information~\cite{CovTho:06},
for the purpose of determining the MSR the vector channels in (\ref{eq:yM-simo})
and (\ref{eq:yE-simo}) can equivalently be written in a (complex)
scalar form corresponding to the Gaussian wiretap channel introduced
in \cite{LeuHel:78}. Then, the MSR~$\Cs$ of the legitimate channel
for a given realization of the channels~$\mathbf{h}_{\M}$ and $\mathbf{h}_{\mathrm{e}}$
is given by\begin{align}
\Cs & =\left[\log_{2}\left(\frac{1+\mathbf{h}_{\M}^{\dagger}\Sigma_{\M}^{-1}\mathbf{h}_{\M}\PM}{1+\mathbf{h}_{\mathrm{e}}^{\dagger}\Sigma_{\mathrm{e}}^{-1}\mathbf{h}_{\mathrm{e}}\PM}\right)\right]^{+}.\label{eq:Cs-SIMO}\end{align}
Setting $\mathbf{h}_{\M}=\sqrt{g(\rM)}$, $\mathbf{h}_{\mathrm{e}}=\left[\sqrt{g(\RE 1)},\sqrt{g(\RE 2)},\cdots\right]^{T},$
$\boldsymbol{\Sigma}_{\M}=\WM\mathbf{I}_{1}$, and $\boldsymbol{\Sigma}_{\mathrm{e}}=\WE\mathbf{I}_{\infty}$,
where $\WM$ and $\WE$ are the noise powers of the legitimate and
eavesdropper receivers, respectively, and $\mathbf{I}_{n}$ is the
$n\times n$ identity matrix, (\ref{eq:Cs-SIMO}) reduces to (\ref{eq:Cs-collusion}).
This concludes the proof.\end{proof}

\subsection{Distribution of the Maximum Secrecy Rate of a Single Link\label{sec:Distribution-Cs-Colluding}}

Theorem~\ref{thm:Cs-Colluding} is valid for a given realization
of the spatial process~$\PiE$. In general, the MSR~$\Cs$ of the
legitimate link in (\ref{eq:Cs-collusion}) is a RV, since it is a
function the random eavesdropper distances $\{\RE i\}_{i=1}^{\infty}$.
In what follows, we analyze the case where $\PiE$ is a homogeneous
Poisson process on the two-dimensional plane with density~$\LE$,
and the channel gain is of the form~$g(r)=\frac{1}{r^{2b}}$ with
$b>1$. The following theorem characterizes the distribution of the
MSR in this scenario.

\begin{theorem}\label{thm:Distribution-Cs-Colluding}If $\PiE$ is
a  Poisson process with density~$\LE$ and $g(r)=\frac{1}{r^{2b}}$,
$b>1$, the MSR~$\Cs$ of the legitimate link is a RV whose CDF~$F_{\Cs}(\th)$
is given by\begin{equation}
F_{\Cs}(\th)=\begin{cases}
0, & \th<0,\\
1-F_{\PET}\left(\frac{\left(1+\frac{\PM}{\rM^{2b}\WM}\right)2^{-\th}-1}{(\pi\LE\mathcal{C}_{1/b}^{-1})^{b}\frac{\PM}{\WE}}\right), & 0\leq\th<\CM,\\
1, & \th\geq\CM,\end{cases}\label{eq:F-Cs-colluding}\end{equation}
where $\CM=\log_{2}\left(1+\frac{\PM}{\rM^{2b}\WM}\right)$ is the
capacity of the legitimate channel; $\mathcal{C}_{\alpha}$ is defined
as\begin{equation}
\mathcal{C}_{\alpha}\triangleq\frac{1-\alpha}{\Gamma(2-\alpha)\cos\left(\frac{\pi\alpha}{2}\right)}\label{eq:C-x}\end{equation}
with $\Gamma(\cdot)$ denoting the gamma function; and $F_{\PET}(\cdot)$
is the CDF of a skewed stable RV~$\PET$, with parameters%
\footnote{We use $\mathcal{S}(\alpha,\beta,\gamma)$ to denote the distribution
of a real stable RV with characteristic exponent~$\mbox{$\alpha\in(0,2]$}$,
skewness~$\mbox{$\beta\in[-1,1]$}$, and dispersion~$\mbox{$\gamma\in[0,\infty)$}$.
The corresponding characteristic function is~\cite{SamTaq:94}\begin{equation}
\phi(w)=\begin{cases}
\exp\left(-\gamma|w|^{\alpha}\left[1-j\beta\mathrm{\, sign}(w)\tan\left(\frac{\pi\alpha}{2}\right)\right]\right), & \alpha\neq1,\\
\exp\left(-\gamma|w|\left[1+j\frac{2}{\pi}\beta\mathrm{\, sign}(w)\ln|w|\right]\right), & \alpha=1.\end{cases}\label{eq:stable-CF}\end{equation}
}\begin{equation}
\PET\sim\mathcal{S}\left(\alpha=\frac{1}{b},\:\beta=1,\:\gamma=1\right).\label{eq:Pe-tilde-stable}\end{equation}
\end{theorem}

\begin{proof}For $g(r)=\frac{1}{r^{2b}}$, the MSR~$\Cs$ of the
legitimate channel in (\ref{eq:Cs-collusion}) is a function of the
total power received by the eavesdroppers, $\PE=\sum_{i=1}^{\infty}\frac{\PM}{\RE i^{2b}}$.
If $\PiE$ is a  Poisson process, the characteristic function of $\PE$
can be written as~\cite{WinPinShe:J09}\begin{equation}
\PE\sim\mathcal{S}\left(\alpha=\frac{1}{b},\:\beta=1,\:\gamma=\pi\LE\mathcal{C}_{1/b}^{-1}\PM^{1/b}\right),\label{eq:Pe-stable}\end{equation}
for $b>1$. Defining the normalized stable RV~$\PET\triangleq\PE\gamma^{-b}$
with $\gamma=\pi\LE\mathcal{C}_{1/b}^{-1}\PM^{1/b}$, we have that
$\PET\sim\mathcal{S}\left(\frac{1}{b},1,1\right)$ from the scaling
property~\cite{SamTaq:94}. In general, the CDF\ $F_{\PET}(\cdot)$
cannot be expressed in closed form except in the case where $\mbox{$b=2$}$,
which is analyzed in Section~\ref{sec:Case-Study}. However, the
characteristic function of $\PET$ has the simple form of $\phi_{\PET}(w)=\exp\left(-|w|^{1/b}\left[1-j\mathrm{\, sign}(w)\tan\left(\frac{\pi}{2b}\right)\right]\right)$,
and thus $F_{\PET}(\cdot)$ can always be expressed in the integral
form for numerical evaluation.

Using (\ref{eq:Cs-collusion}), we can now express $F_{\Cs}(\th)$
in terms of the CDF of $\PET$, for $0\leq\th<\CM$, as\begin{align*}
F_{\Cs}(\th) & =\mathbb{P}\{\Cs\leq\th\}\\
 & =\mathbb{P}\left\{ \log_{2}\left(1+\frac{\PM}{\rM^{2b}\WM}\right)-\log_{2}\left(1+\frac{\PE}{\WE}\right)\leq\th\right\} \\
 & =1-\mathbb{P}\left\{ \PE\leq\WE\left[\left(1+\frac{\PM}{\rM^{2b}\WM}\right)2^{-\th}-1\right]\right\} \\
 & =1-F_{\PET}\left(\frac{\left(1+\frac{\PM}{\rM^{2b}\WM}\right)2^{-\th}-1}{(\pi\LE\mathcal{C}_{1/b}^{-1})^{b}\frac{\PM}{\WE}}\right).\end{align*}
In addition, $F_{\Cs}(\th)=0$ for $\th<0$ and $F_{\Cs}(\th)=1$
for $\th\geq\CM$, since the RV~$\Cs$ in (\ref{eq:Cs-collusion})
satisfies $0\leq\Cs\leq\CM$, i.e., the MSR of the legitimate link
in the presence of colluding eavesdroppers is a positive quantity
which cannot be greater than the MSR of the legitimate link \emph{in
the absence of eavesdroppers}. This is the result in (\ref{eq:F-Cs-non-colluding})
and the proof is complete.\end{proof}

\subsection{Existence and Outage of the Maximum Secrecy Rate of a Single Link\label{sec:Exist-Out-Colluding}}

Based on the results of Section~\ref{sec:Distribution-Cs-Colluding},
we can now obtain the probability of existence of a non-zero MSR,
and the probability of secrecy outage for a single legitimate link
in the presence of colluding eavesdroppers. \textit{\emph{The following
corollary provides such probabilities.}}

\begin{corollary}If $\PiE$ is a Poisson process with density~$\LE$
and $g(r)=\frac{1}{r^{2b}}$, $b>1$, the probability of \emph{existence}
of a non-zero MSR in the legitimate link, $p_{\mathrm{exist}}=\mathbb{P}\{\Cs>0\}$,
is given by\begin{equation}
p_{\mathrm{exist}}=F_{\PET}\left(\frac{\WE}{(\pi\LE\rM^{2}\mathcal{C}_{1/b}^{-1})^{b}\WM}\right),\label{eq:Pexist-collusion}\end{equation}
and the probability of an \emph{outage} in the MSR of the legitimate
link, $p_{\mathrm{outage}}(\th)=\mathbb{P}\{\Cs<\th\}$ for $\th>0$,
is given by\begin{equation}
p_{\mathrm{outage}}(\th)=\begin{cases}
1-F_{\PET}\left(\frac{\left(1+\frac{\PM}{\rM^{2b}\WM}\right)2^{-\th}-1}{(\pi\LE\mathcal{C}_{1/b}^{-1})^{b}\frac{\PM}{\WE}}\right), & 0<\th<\CM,\\
1, & \th\geq\CM,\end{cases}\label{eq:Poutage-collusion}\end{equation}
where $\CM=\log_{2}\left(1+\frac{\PM}{\rM^{2b}\WM}\right)$ is the
capacity of the legitimate channel; and $F_{\PET}(\cdot)$ is the
CDF of the normalized stable RV~$\PET$, with parameters given in
(\ref{eq:Pe-tilde-stable}).\end{corollary}

\begin{proof}The expressions for $p_{\mathrm{exist}}$ and $p_{\mathrm{outage}}(\th)$
follow directly from (\ref{eq:F-Cs-colluding}).\end{proof}

\subsection{Colluding vs.\ Non-Colluding Eavesdroppers for a Single Link\label{sub:Comparison-Colluding}}

We have so far considered the fundamental secrecy limits of a single
legitimate  link in the presence of colluding eavesdroppers. According
to Theorem~\ref{thm:Cs-Colluding}, such scenario is equivalent to
having a single eavesdropper with an array that collects a total power~$\PET=\sum_{i=1}^{\infty}\PM/\RE i^{2b}$.
In particular, when the eavesdroppers are positioned according to
an homogeneous Poisson process, Theorem~\ref{thm:Distribution-Cs-Colluding}
shows that the RV~$\PE$ has a skewed stable distribution.

We can obtain further insights by establishing a comparison with the
case of a single legitimate link in the presence of \emph{non-colluding
eavesdroppers}. In such scenario, the MSR does not depend on all eavesdroppers,
but only on that with maximum received power (i.e.,~the closest one,
when only path loss is present). Thus, the total eavesdropper power
is given by $\PE=\frac{\PM}{\RE 1^{2b}}$. Using the fact that $\RE 1^{2}$
is exponentially distributed with rate~$\pi\LE$, the PDF of $\PE$
can be written as\[
f_{\PE}(x)=\frac{\pi\LE}{bx}\left(\frac{\PM}{x}\right)^{1/b}\exp\left(-\pi\LE\left(\frac{\PM}{x}\right)^{1/b}\right),\quad x\geq0,\]
and the CDF of the corresponding MSR~$\Cs$ can be easily determined
from (\ref{eq:Cs-collusion}) as\begin{equation}
F_{\Cs}(\th)=\begin{cases}
0, & \th<0,\\
1-\exp\left(-\pi\LE\left(\frac{\frac{\PM}{\WE}}{\left(1+\frac{\PM}{\rM^{2b}\WM}\right)2^{-\th}-1}\right)^{1/b}\right), & 0\leq\th<\CM,\\
1, & \th\geq\CM.\end{cases}\label{eq:F-Cs-non-colluding}\end{equation}
From this CDF, we can readily determine the probability of existence
of a non-zero MSR, and the probability of secrecy outage, similarly
to the colluding case. Table~\ref{tab:c-vs-nc} summarizes the differences
between the colluding and non-colluding scenarios for a single legitimate
link.

\subsection{$\is$Graph with Colluding Eavesdroppers\label{sub:Collusion-SGraph}}

To study the effect of colluding eavesdroppers, we have so far made
a simplification concerning the legitimate nodes. Specifically, we
considered only a single legitimate link with deterministic length~$\rM$
as depicted in Fig.~\ref{fig:collusion}, thus eliminating the randomness
associated with the position of the legitimate nodes. We now revisit
the $\is$graph model depicted in Fig.~\ref{fig:secrecy-graph},
where both legitimate nodes and eavesdroppers are distributed according
to Poisson processes~$\PiM$ and $\PiE$. In particular, the following
theorem characterizes the effect of collusion in terms of the resulting
average node degree in such graph.

\begin{theorem}For the Poisson $\is$graph with colluding eavesdroppers,
secrecy rate threshold~$\th=0$, equal noise powers~$\WM=\WE$,
and channel gain function~$g(r)=\frac{1}{r^{2b}}$, $b>1$, the average
degrees of a typical node are\begin{equation}
\mathbb{E}\{N_{\mathrm{in}}\}=\mathbb{E}\{N_{\mathrm{out}}\}=\frac{\LM}{\LE}\,\textrm{sinc}\!\left(\frac{1}{b}\right),\label{eq:EN-colluding}\end{equation}
where $\textrm{sinc}(x)\triangleq\frac{\sin(\pi x)}{\pi x}$.\end{theorem}

\begin{proof}We consider the process~$\PiM\cup\{0\}$ obtained by
adding a legitimate node to the origin of the coordinate system, and
denote the out-degree of the node at the origin by $N_{\mathrm{out}}$.
Using (\ref{eq:Cs-collusion}), we can write\begin{align*}
N_{\mathrm{out}} & =\#\left\{ x_{i}\in\PiM:\cs i>0\right\} \\
 & =\#\Biggl\{ x_{i}\in\PiM:\RM i^{2}<\underbrace{\left(\frac{\PM}{\PE}\right)^{1/b}}_{\triangleq\nu^{2}}\Biggr\}.\end{align*}
The average out-degree can be determined as\begin{align}
\mathbb{E}\{N_{\mathrm{out}}\} & =\mathbb{E}_{\PiM,\PiE}\{\PiM\{\mathcal{B}_{0}(\nu)\}\}\nonumber \\
 & =\mathbb{E}_{\PiE}\{\LM\pi\nu^{2}\}\nonumber \\
 & =\LM\pi\mathbb{E}_{\PiE}\left\{ \left(\frac{\PM}{\PE}\right)^{1/b}\right\} .\label{eq:EN-colluding-temp}\end{align}
where the RV~$\PE$ has a stable distribution with parameters given
in (\ref{eq:Pe-stable}). As before, we define the normalized stable
RV~~$\PET\triangleq\PE\gamma^{-b}$ with $\gamma=\pi\LE\mathcal{C}_{1/b}^{-1}\PM^{1/b}$,
such that $\PET\sim\mathcal{S}\left(\frac{1}{b},1,1\right)$. Then,
we can rewrite (\ref{eq:EN-colluding-temp}) as\begin{equation}
\mathbb{E}\{N_{\mathrm{out}}\}=\frac{\LM}{\LE}\mathcal{C}_{1/b}\mathbb{E}\{\PET^{-1/b}\}.\label{eq:EN-colluding-temp2}\end{equation}
Using the Mellin transform of a stable RV, we show in Appendix~\ref{sec:Deriv-Mellin}
that (\ref{eq:EN-colluding-temp2}) simplifies to\begin{equation}
\mathbb{E}\{N_{\mathrm{out}}\}=\frac{\LM}{\LE}\,\textrm{sinc}\!\left(\frac{1}{b}\right).\label{eq:EN-colluding-temp3}\end{equation}
Noting that $\mathbb{E}\{N_{\mathrm{in}}\}=\mathbb{E}\{N_{\mathrm{out}}\}$
for any directed random graph, we obtain the desired result in (\ref{eq:EN-colluding}).\end{proof}

It is insightful to rewrite (\ref{eq:EN-colluding}) as $\mathbb{E}\{N_{\mathrm{out}}|\textrm{colluding}\}=\mathbb{E}\{N_{\mathrm{out}}|\textrm{non-colluding}\}\cdot\eta(b)$,
where $\eta(b)=\textrm{sinc}\!\left(\frac{1}{b}\right)$, and $\eta(b)<1$
for $b>1$. The function~$\eta(b)$ can be interpreted as the \emph{degradation
factor in average connectivity due to eavesdropper collusion. }In
the extreme where $b=1$, we have complete loss of secure connectivity
with $\eta(1)=0$. This is because the series~$\PE=\sum_{i=1}^{\infty}\frac{\PM}{\RE i^{2b}}$
diverges (i.e., the total received eavesdropper power is infinite),
so the resulting average node degree is zero. In the other extreme
where $b\rightarrow\infty$, we achieve the highest secure connectivity
with $\eta(\infty)=1$. This is because the first term~$\frac{\PM}{\RE 1^{2b}}$
in the $\PE$ series (corresponding to the non-colluding term) is
dominant, so the average node degree in the colluding case approaches
the non-colluding one. In conclusion, cluttered environments with
larger amplitude loss exponents~$b$ are more favorable for secure
communication, in the sense that in such environments collusion only
provides a marginal performance improvement for the eavesdroppers.

\subsection{Numerical Results\label{sec:Case-Study}}

We now illustrate the results obtained in the previous sections with
a simple case study. We consider the case where $\WM=\WE=\W$, i.e.,
the legitimate link and the eavesdroppers are subject to the same
noise power, which is introduced by the electronics of the respective
receivers. Furthermore, we consider that the amplitude loss exponent
is $b=2$, in which case the CDF of $\PET$ for colluding eavesdroppers
can be expressed using the Gaussian $Q$-function as $F_{\PET}(x)=2Q(1/\sqrt{x}),x\geq0$.
The CDF of $\Cs$ in (\ref{eq:F-Cs-colluding}) reduces to\begin{equation}
F_{\Cs}(\th)=\begin{cases}
0, & \th<0,\\
1-2Q\left(\pi\LE\mathcal{C}_{1/2}^{-1}\sqrt{\frac{\frac{\PM}{\W}}{\left(1+\frac{\PM}{\rM^{4}\W}\right)2^{-\th}-1}}\right), & 0\leq\th<\CM,\\
1, & \th\geq\CM.\end{cases}\label{eq:F-Cs-example}\end{equation}
In addition, (\ref{eq:Pexist-collusion}) and (\ref{eq:Poutage-collusion})
reduce, respectively, to\begin{equation}
p_{\mathrm{exist}}=2Q\left(\pi\LE\rM^{2}\mathcal{C}_{1/2}^{-1}\right)\label{eq:Pexist-example}\end{equation}
and\begin{equation}
p_{\mathrm{outage}}(\th)=\begin{cases}
1-2Q\left(\pi\LE\mathcal{C}_{1/2}^{-1}\sqrt{\frac{\frac{\PM}{\W}}{\left(1+\frac{\PM}{\rM^{4}\W}\right)2^{-\th}-1}}\right), & 0<\th<\CM,\\
1, & \th\geq\CM.\end{cases}\label{eq:Poutage-example}\end{equation}
From these analytical results, we observe that of the following factors
lead to\emph{ }a \emph{degradation} of the security of communications:
increasing $\LE$ or $\rM$, decreasing $\PM/\W$, or allowing the
eavesdroppers to collude. In particular, as we let $\PM/\W\rightarrow\infty$,
$p_{\mathrm{outage}}$ decreases monotonically, converging to the
curve~$p_{\mathrm{outage}}=1-\exp\left(-\pi\LE\rM^{2}2^{\th/2}\right)$
in the non-colluding case, and to $p_{\mathrm{outage}}=1-2Q\left(\pi\LE\rM^{2}\mathcal{C}_{1/2}^{-1}2^{\th/2}\right)$
in the colluding case.

Figure~\ref{fig:pdf-pe-c-nc} compares the PDFs of the (normalized)
received eavesdropper power~$\frac{\PE}{\PM}$, for the cases of
colluding and non-colluding eavesdroppers. For $b>1$, it is clear
that $\sum_{i=1}^{\infty}\frac{1}{\RE i^{2b}}>\frac{1}{\RE 1^{2b}}$
a.s., i.e., the received eavesdropper power~$\PE$ is larger in the
colluding case, resulting in a PDF whose mass is more biased towards
higher realizations of $\PE$. 

Figure~\ref{fig:pexist-lambda-collusion} plots the probability~$p_{\mathrm{exist}}$
of existence of a non-zero MSR, given in (\ref{eq:Pexist-example}),
as a function of the eavesdropper density~$\LE$, for various values
of the legitimate link length~$\rM$. As predicted by analytically,
the existence of a non-zero MSR becomes \emph{less likely }by increasing
$\LE$ or $\rM$.%
\footnote{Note that $p_{\mathrm{exist}}$ in (\ref{eq:Pexist-example}) depends
on $\LE$ and $\rM$ only through the product~$\LE\rM^{2}$. %
} A similar degradation in secrecy occurs by allowing the eavesdroppers
to collude, since more signal power from the legitimate user is available
to the eavesdroppers, improving their ability to decode the secret
message.

Figure~\ref{fig:cdf-cs-collusion} quantifies the probability~$p_{\mathrm{outage}}$
of secrecy outage, given in (\ref{eq:Poutage-example}), as a function
of the desired secrecy rate~$\th$, for various values of eavesdropper
density. The vertical line marks the capacity~$\CM$ of the legitimate
link, which for the parameters indicated in Fig.~\ref{fig:cdf-cs-collusion}
is $\CM=\log_{2}\left(1+\frac{\PM}{\rM^{2b}\WM}\right)=3.46$ bits
per complex dimension. As expected, if the target secrecy rate~$\th$
set by the transmitter exceeds $\CM$, a secrecy outage occurs with
probability~1, since the MSR~$\Cs$ cannot be greater that the capacity~$\CM$
of the legitimate link. In comparison with the non-colluding case,
the ability of the eavesdroppers to collude leads to higher probabilities
of secrecy outage. This is because more signal power from the legitimate
user is available to the eavesdroppers, improving their ability to
decode the secret message. A similar degradation in secrecy occurs
by increasing the eavesdropper density~$\LE$.

Figure~\ref{fig:out-degree-collusion} quantifies the (normalized)
average node degree of the $\is$graph, $\frac{\mathbb{E}\{N_{\mathrm{out}}\}}{\LM/\LE}$,
versus the amplitude loss exponent~$b$. The normalizing factor~$\LM/\LE$
corresponds to the average out-degree in the non-colluding case. As
predicted analytically, we observe that in the colluding case, the
normalized average out-degree~$\eta(b)=\frac{\mathbb{E}\{N_{\mathrm{out}}\}}{\LM/\LE}$
is strictly increasing with $b$. Furthermore, $\eta(1)=0$ because
the received eavesdropper power~$\PE$ is infinite, and $\eta(\infty)=1$
because the first (non-colluding) term in the $\PE$ series dominates
the other terms. It is apparent from the figure that cluttered environments
with larger amplitude loss exponents~$b$ are more favorable for
secure communication, in the sense that in such environments collusion
only provides a marginal performance improvement for the eavesdroppers.

\section{Conclusions\label{sec:Conclusion}}

Using the notion of strong secrecy, we provided an information-theoretic
definition of the $\is$graph as a model for intrinsically secure
communication in large-scale networks. Fundamental tools from stochastic
geometry allowed us to describe in detail how the spatial densities
of legitimate and eavesdropper nodes influence various properties
of the Poisson $\is$graph, such as node degrees and isolation probabilities.
In particular, we proved that the average in- and out-degrees equal
$\frac{\LM}{\LE}$, and that out-isolation is more probable than in-isolation.
In addition, we \textit{\emph{considered the effect of the wireless
propagation on the degree of the legitimate nodes. Surprisingly, }}the
average node degree is invariant with respect to the distribution
of the propagation effects (e.g.,~type of fading or shadowing), and
is always equal to the ratio~$\frac{\LM}{\LE}$ of spatial densities.\textit{\emph{
}}We then studied the effect of non-zero secrecy rate threshold~$\th$
and unequal noise powers $\WM,\WE$ on the $\is$graph. Specifically,
we showed that $\mathbb{E}\{N_{\mathrm{out}}\}$ is decreasing in
$\th$ and $\WM$, and is increasing in $\WE$. Furthermore, when
the channel gain is of the form~$g(r)=\frac{1}{r^{2b}}$, we obtained
expressions for $\mathbb{E}\{N_{\mathrm{out}}\}$ as a function of
$\th,\WM,\WE$, and showed that it decays exponentially with $\th$.

We explored the potential of sectorized transmission and eavesdropper
neutralization as two techniques for enhancing the secrecy of communications.
If each legitimate node is able to transmit independently in $L$~sectors
of the plane, our results prove that $\mathbb{E}\{N_{\mathrm{out}}\}$
increases \emph{linearly} with $L$. On the other hand, if legitimate
nodes are able to inspect their surrounding area to guarantee that
there are no eavesdroppers within a neutralization radius~$\rho$,
then $\mathbb{E}\{N_{\mathrm{out}}\}$ increases \emph{at least} \emph{exponentially}
with $\rho$.

The PDF of the MSR~$\cs i$ between a legitimate node and its $i$-th
neighbor was characterized, as well as the probability of existence
of a non-zero MSR, and the probability of secrecy outage. In particular,
we quantified how these metrics depend on the densities~$\LM,\LE$,
the SNR~$\frac{\PM}{\W}$, and the amplitude loss exponent~$b$.

Finally, we established the fundamental secrecy limits when the eavesdroppers
are allowed to collude, by showing that this scenario is equivalent
to a SIMO Gaussian wiretap channel. For an arbitrary spatial process~$\PiE$
of the eavesdroppers, we derived the MSR of a legitimate link. Then,
for the case where $\PiE$ is a spatial Poisson process and the channel
gain is of the form~$g(r)=\frac{1}{r^{2b}}$,  we obtained the CDF
of MSR of a legitimate link, and the average degree in the $\is$graph
with colluding eavesdroppers. We concluded that as we increase the
density~$\LE$ of eavesdroppers, or allow the eavesdroppers to collude,
more power is available to the adversary, improving their ability
to decode the secret message, and hence decreasing the MSR of legitimate
links. Furthermore, we showed that cluttered environments with large
amplitude loss exponent~$b$ are move favorable for secure communications,
in the sense that in such regime collusion only provides a marginal
performance improvement for the eavesdroppers.

Perhaps the most interesting insight to be gained from our results,
is the exact quantification of the impact of the eavesdropper density~$\LE$
on the achievable secrecy rates --- a modest density of scattered
eavesdroppers can potentially cause a drastic reduction in the MSR
provided at the physical layer of wireless communication networks.
Our work has not yet addressed all of the far reaching implications
of the broadcast property of the wireless medium. In the most general
scenario, legitimate nodes could for example transmit their signals
in a cooperative fashion, whereas malicious nodes could use jamming
to disrupt all communications. We hope that further efforts in combining
stochastic geometry with information-theoretic principles will lead
to a more comprehensive treatment of wireless security.

\appendices

\section{Proof that Inequality (\ref{eq:pin-isol-strict}) is Strict\label{sec:Deriv-Strict-Ineq}}

Define the event~$F_{i}\triangleq\{\PiE\{\mathcal{B}_{\breve{x}_{i}}(|\breve{x}_{i}|)\}\geq1\}$
and its complementary event~$E_{i}$, which denote \emph{full} and
\emph{empty}, respectively. Using this notation, we can rewrite (\ref{eq:pin-isol-proof})
as\begin{align*}
p_{\mathrm{in-isol}} & =\mathbb{P}\left\{ \bigwedge_{i=1}^{\infty}F_{i}\right\} \\
 & \leq\mathbb{P}\{F_{1}\wedge F_{2}\}.\end{align*}
To prove that $p_{\mathrm{in-isol}}<\mathbb{P}\{F_{1}\}$ as in (\ref{eq:pin-isol-strict}),
it is sufficient to show that $\mathbb{P}\{F_{1}\wedge F_{2}\}<\mathbb{P}\{F_{1}\}$,
or equivalently, $\mathbb{P}\{F_{1}\}-\mathbb{P}\{F_{1}\wedge F_{2}\}=\mathbb{P}\{F_{1}\wedge E_{2}\}>0$.
Define the ball~$\mathcal{B}_{i}\triangleq\mathcal{B}_{\breve{x}_{i}}(|\breve{x}_{i}|)$.
Then, with reference to the auxiliary diagram in Fig.~\ref{fig:aux-in-isol},
we can write\begin{align}
\mathbb{P}\{F_{1}\wedge E_{2}\} & =\mathbb{E}_{\PiM}\{\mathbb{P}\{F_{1}\wedge E_{2}|\PiM\}\}\nonumber \\
 & =\mathbb{E}_{\PiM}\left\{ \left(1-e^{-\LE\mathbb{A}\{\mathcal{B}_{1}\backslash\mathcal{B}_{2}\}}\right)\cdot e^{-\LE\mathbb{A}\{\mathcal{B}_{2}\}}\right\} .\label{eq:P-FE}\end{align}
Since $\mathcal{B}_{1}\nsubseteq\mathcal{B}_{2}$ a.s., then $\mathbb{A}\{\mathcal{B}_{1}\backslash\mathcal{B}_{2}\}>0$
a.s., and the argument inside the expectation in (\ref{eq:P-FE})
is strictly positive, and thus $\mathbb{P}\{F_{1}\wedge E_{2}\}>0$.
This concludes the proof.

\section{Derivation of (\ref{eq:Px-bound})\label{sec:Deriv-Px-bound}}

Because $\PiM$ is a Poisson process, the Palm probability~$\mathbb{P}_{x}\{|x|<\RE 1\}$
in (\ref{eq:ENout-temp}) can be computed using Slivnyak's theorem
by adding a legitimate node at location~$x$ to $\PiM$. For a fixed\textbf{~$x\in\An(\rho,\infty)$},
we can thus write\begin{align}
\mathbb{P}_{x}\{|x|<\RE 1\} & =\mathbb{P}_{\Theta,\PiE}\{\PiE\{\overline{\Theta}\cap\An(\rho,|x|)\backslash\mathcal{B}_{x}(\rho)\}=0\}\label{eq:Px-temp-1}\\
 & \geq\mathbb{P}_{\Theta,\PiE}\{\PiE\{\overline{\Theta}\cap\An(\rho,|x|)\}=0\}\label{eq:Px-temp-2}\\
 & =\mathbb{E}_{\Theta}\{\exp(-\LE\mathbb{A}\{\overline{\Theta}\cap\An(\rho,|x|)\}\}\label{eq:Px-temp-3}\\
 & \geq\exp(-\LE\mathbb{E}_{\Theta}\{\mathbb{A}\{\overline{\Theta}\cap\An(\rho,|x|)\}\}),\label{eq:Px-temp-4}\end{align}
Equation (\ref{eq:Px-temp-3}) follows from conditioning on $\Theta$,
and using the fact that $\PiE$ and $\Theta$ are independent. Equation
(\ref{eq:Px-temp-4}) follows from Jensen's inequality. The term inside
the exponential in (\ref{eq:Px-temp-4}) corresponds to the average
area of a random shape, and can  be computed using Fubini's theorem
as \begin{align}
\mathbb{E}_{\Theta}\{\mathbb{A}\{\overline{\Theta}\cap\An(\rho,|x|)\}\} & =\mathbb{E}_{\Theta}\left\{ {\int\intop}_{\mathbb{R}^{2}}\mathbbm{1}\{y\in\overline{\Theta}\cap\An(\rho,|x|)\}dy\right\} \nonumber \\
 & ={\int\intop}_{\An(\rho,|x|)}\mathbb{P}\{y\in\overline{\Theta}\}dy\nonumber \\
 & ={\int\intop}_{\An(\rho,|x|)}\mathbb{P}\{\PiM\{\mathcal{B}_{y}(\rho)\}=0\}dy\nonumber \\
 & ={\int\intop}_{\An(\rho,|x|)}\underbrace{e^{-\LM\pi\rho^{2}}}_{\triangleq p_{\overline{\Theta}}}dy\nonumber \\
 & =p_{\overline{\Theta}}\pi(|x|^{2}-\rho^{2})\label{eq:EA-temp}\end{align}
Note that $p_{\overline{\Theta}}$ corresponds to the probability
that a fixed point~$y$ is \emph{outside} the total neutralization
region~$\Theta$, and does not depend on the coordinates of $y$
due to the stationarity of the process~$\Theta$. Replacing (\ref{eq:EA-temp})
into (\ref{eq:Px-temp-4}), we obtain the desired inequality in (\ref{eq:Px-bound}).

\section{Derivation of (\ref{eq:EN-colluding-temp3})\label{sec:Deriv-Mellin}}

Let the Mellin transform of a RV~$X$ with PDF~$f_{X}(x)$ be defined
as%
\footnote{In the literature, the Mellin transform is sometimes defined differently
as $\mathcal{M}_{X}(s)\triangleq\int_{0}^{\infty}x^{s-1}f_{X}(x)dx$.
For simplicity, we prefer the definition in (\ref{eq:Mellin-transform}).%
}\begin{equation}
\mathcal{M}_{X}(s)\triangleq\int_{0}^{\infty}x^{s}f_{X}(x)dx.\label{eq:Mellin-transform}\end{equation}
If $X\sim\mathcal{S}\left(\alpha,1,1\right)$ with $0<\alpha<1$,
then~\cite[Eq. (17)]{Zol:57}\begin{equation}
\mathcal{M}_{X}(s)=\left(\cos\left(\frac{\pi\alpha}{2}\right)\right)^{-s/\alpha}\frac{\Gamma\left(1-\frac{s}{\alpha}\right)}{\Gamma(1-s)},\label{eq:Mellin-stable}\end{equation}
for $-1<\textrm{Re}\{s\}<\alpha$. Then, since $\PET\sim\mathcal{S}\left(\alpha,1,1\right)$
with $\alpha=\frac{1}{b}\in(0,1)$, we use (\ref{eq:Mellin-stable})
to write\begin{align}
\mathbb{E}\{\PET^{-\alpha}\} & =\int_{0}^{\infty}x^{-\alpha}f_{\PET}(x)dx\nonumber \\
 & =\mathcal{M}_{\PET}(-\alpha)\nonumber \\
 & =\frac{\cos\left(\frac{\pi\alpha}{2}\right)}{\Gamma(1+\alpha)}.\label{eq:E-Pe-temp}\end{align}
Using (\ref{eq:C-x}) and (\ref{eq:E-Pe-temp}), we expand (\ref{eq:EN-colluding-temp2})
as\begin{align*}
\mathbb{E}\{N_{\mathrm{out}}\} & =\frac{\LM}{\LE}\mathcal{C}_{\alpha}\mathbb{E}\{\PET^{-\alpha}\}\\
 & =\frac{\LM}{\LE}\cdot\frac{1-\alpha}{\Gamma(2-\alpha)\cos\left(\frac{\pi\alpha}{2}\right)}\cdot\frac{\cos\left(\frac{\pi\alpha}{2}\right)}{\Gamma(1+\alpha)}\\
 & =\frac{\LM}{\LE}\cdot\frac{1-\alpha}{\Gamma(2-\alpha)\Gamma(1+\alpha)}\\
 & =\frac{\LM}{\LE}\cdot\frac{\sin(\pi\alpha)}{\pi\alpha},\end{align*}
where we used the following properties of the gamma function: $\Gamma(z+1)=z\Gamma(z)$
and $\Gamma(z)\Gamma(1-z)=\frac{\pi}{\sin(\pi z)}$. Defining $\textrm{sinc}(x)\triangleq\frac{\sin(\pi x)}{\pi x}$
and noting that $\alpha=\frac{1}{b}$, we obtain (\ref{eq:EN-colluding-temp3}).

\section*{Acknowledgements}

\noindent The authors would like to thank L.~A.~Shepp, Y.~Shen,
and W.~Swantantisuk for their helpful suggestions.

\bibliographystyle{bibtex/IEEEtran}
\bibliography{bibtex/IEEEabrv,bibtex/StringDefinitions,bibtex/WGroup,bibtex/BiblioCV}

\begin{table}[p]
\begin{centering}
\begin{tabular}{|c|c|}
\hline 
Symbol & Usage\tabularnewline
\hline
$\mathbb{E}\{\cdot\}$ & Expectation operator\tabularnewline
$\mathbb{P}\{\cdot\}$ & Probability operator\tabularnewline
$*$ & Convolution operator\tabularnewline
$\dagger$ & Conjugate transpose operator\tabularnewline
$f_{X}(x)$ & Probability density function of $X$\tabularnewline
$F_{X}(x)$ & Cumulative distribution function of $X$\tabularnewline
$H(X)$ & Entropy of $X$\tabularnewline
$\PiM=\{x_{i}\},\PiE=\{e_{i}\}$ & Poisson processes of legitimate nodes and eavesdroppers\tabularnewline
$\LM,\LE$ & Spatial densities of legitimate nodes and eavesdroppers\tabularnewline
$\Pi\{\mathcal{R}\}$ & Number of nodes of process~$\Pi$ in region~$\mathcal{R}$ \tabularnewline
$N_{\mathrm{in}},N_{\mathrm{out}}$ & In-degree and out-degree of a node\tabularnewline
$\mathcal{B}_{x}(\rho)$ & Ball centered at $x$ with radius~$\rho$\tabularnewline
$\An(a,b)$ & Annular region between radiuses~$a$ and $b$, centered at the origin\tabularnewline
$\mathbb{A}\{\mathcal{R}\}$ & Area of region~$\mathcal{R}$\tabularnewline
$Z_{x_{i},x_{j}}$ & Random propagation effect between $x_{i}$ and $x_{j}$\tabularnewline
$\RM i$ & Distance between $x_{i}\in\PiM$ and origin\tabularnewline
$\RE i$ & Distance between $e_{i}\in\PiE$ and origin\tabularnewline
$\#S$ & Number of elements in the set~$S$\tabularnewline
$\mathcal{G}(x,\theta)$ & Gamma distribution with mean~$x\theta$ and variance~$x\theta^{2}$\tabularnewline
$\mathcal{N}(\mu,\sigma^{2})$ & Gaussian distribution with mean~$\mu$ and variance~$\sigma^{2}$\tabularnewline
$\mathcal{S}(\alpha,\beta,\gamma)$  & Stable distribution with characteristic exponent~$\mbox{\ensuremath{\alpha}}$,
skewness~$\mbox{\ensuremath{\beta}}$, and dispersion~$\mbox{\ensuremath{\gamma}}$\tabularnewline
\hline
\end{tabular}
\par\end{centering}

\caption{\label{tab:notation}Notation and symbols.}

\end{table}

\begin{figure}[p]
\begin{centering}
\scalebox{1}{\psfrag{Alice}{\footnotesize{\sf{Alice}}}
\psfrag{Bob}{\footnotesize{\sf{Bob}}}
\psfrag{Eve}{\footnotesize{\sf{Eve}}}
\psfrag{main channel}{\hspace{-3mm}\footnotesize{\sf{Legitimate channel}}}
\psfrag{wiretap channel}{\hspace{-2mm}\footnotesize{\sf{Eavesdropper channel}}}
\psfrag{encoder}{\hspace{1mm}\footnotesize{\sf{encoder}}}
\psfrag{decoder}{\hspace{1mm}\footnotesize{\sf{decoder}}}
\psfrag{s}{\normalsize{$s$}}
\psfrag{ssM}{\normalsize{$\hat{s}_{\M}$}}
\psfrag{ssE}{\normalsize{$\hat{s}_{\E}$}}
\psfrag{xM}{\normalsize{$x^{n}$}}
\psfrag{hM}{\normalsize{$h_{\M}^{n}$}}
\psfrag{hE}{\normalsize{$h_{\E}^{n}$}}
\psfrag{wM}{\normalsize{$w_{\M}^{n}$}}
\psfrag{wE}{\normalsize{$w_{\E}^{n}$}}
\psfrag{yM}{\normalsize{$y_{\M}^{n}$}}
\psfrag{yE}{\normalsize{$y_{\E}^{n}$}}\includegraphics{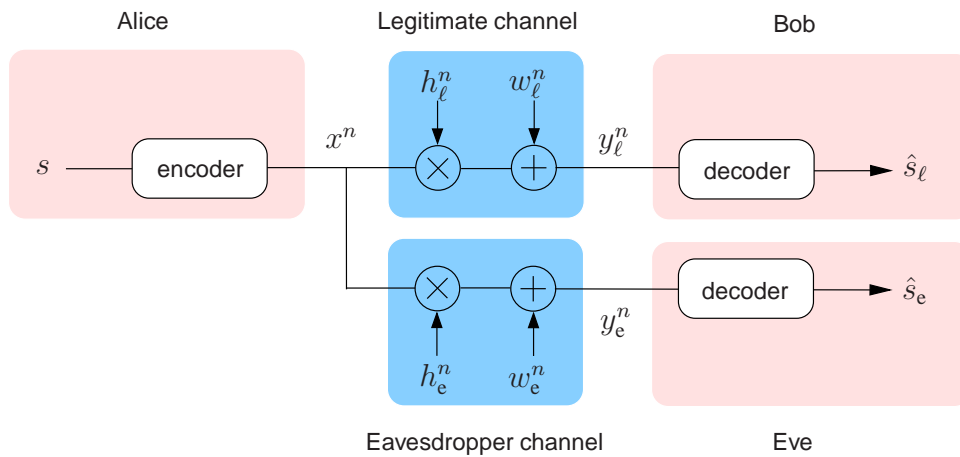}}
\par\end{centering}

\caption{\label{fig:wiretap-channel}Wireless wiretap channel.}

\end{figure}

\begin{figure}[p]
\begin{centering}
\scalebox{1}{\psfrag{Friendly node}{\footnotesize{\sf{Legitimate node}}}
\psfrag{Eavesdropper node}{\footnotesize{\sf{Eavesdropper node}}}\includegraphics{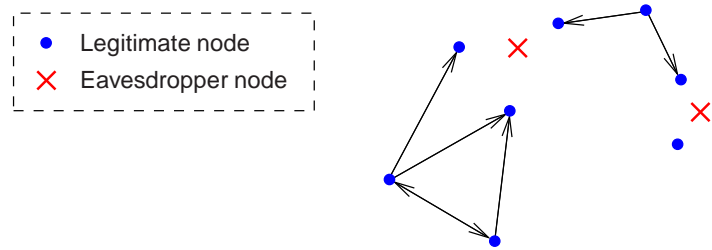}}
\par\end{centering}

\caption{\label{fig:secrecy-graph}Example of an $\is$graph on $\mathbb{R}^{2},$
considering that the secrecy rate threshold is zero, the wireless
environment introduces only path loss, and the noise powers of the
legitimate and eavesdropper nodes are equal. In such scenario, a transmitter~$x_{i}$
is connected to a receiver~$x_{j}$ if and only if $x_{j}$ is closer
to $x_{i}$ than any other eavesdropper, as described in (\ref{eq:edges-pathloss}).}

\end{figure}

\begin{figure}[p]
\begin{centering}
\scalebox{1}{\psfrag{Nout=3}{\normalsize{$N_\textrm{out}=3$}}
\psfrag{Re}{\hspace{-2mm}\normalsize{$\RE{1}$}}
\psfrag{R1}{\hspace{-2mm}\normalsize{$\RM{1}$}}
\psfrag{R2}{\hspace{-2mm}\normalsize{$\RM{2}$}}
\psfrag{R3}{\hspace{-2mm}\normalsize{$\RM{3}$}}\includegraphics{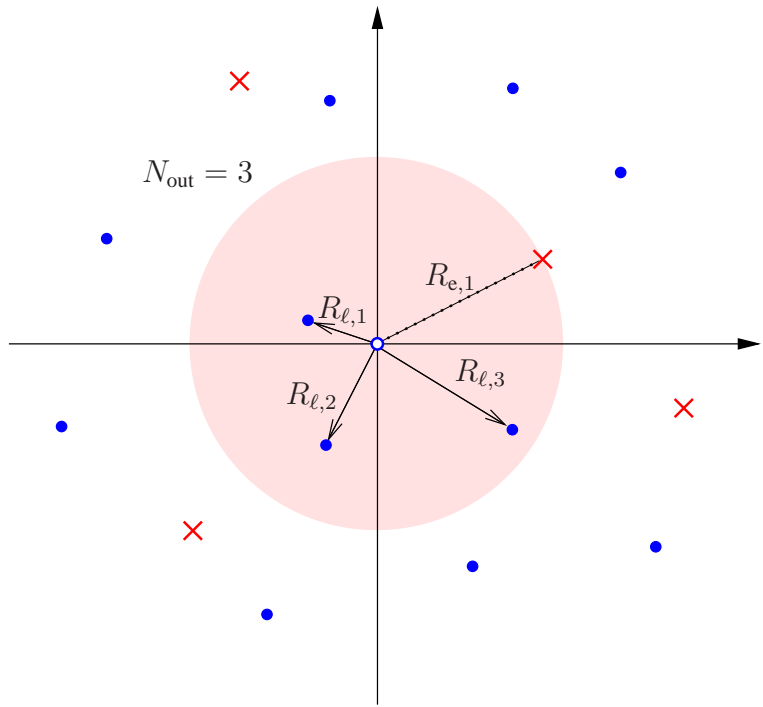}}
\par\end{centering}

\caption{\label{fig:out-degree}Out-degree of a node. In this example, the
node at the origin can transmit messages with information-theoretic
security to $N_{\mathrm{out}}=3$ nodes.}

\end{figure}

\begin{figure}[p]
\begin{centering}
\scalebox{1}{\psfrag{Nin=2}{\normalsize{$N_\textrm{in}=2$}}
\psfrag{R1}{\hspace{-2mm}\normalsize{$\RM{1}$}}
\psfrag{R2}{\hspace{-2mm}\normalsize{$\RM{2}$}}
\psfrag{area A}{\footnotesize{\sf{area $A$}}}\includegraphics{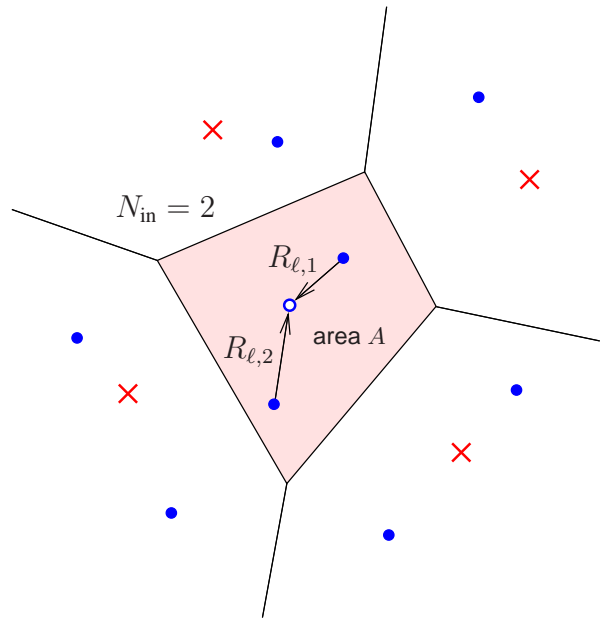}}
\par\end{centering}

\caption{\label{fig:in-degree}In-degree of a node. In this example, the node
at the origin can receive messages with information-theoretic security
from $N_{\mathrm{in}}=2$ nodes. The RV~$A$ is the area of a typical
Voronoi cell, induced by the eavesdropper Poisson process~$\PiE$
with density~$\LE$.}

\end{figure}

\begin{figure}
\centering{}\subfigure[]{
  \label{fig:aux-voronoi}
  \scalebox{0.8}{
   \psfrag{y}{\normalsize{$z$}}
   \psfrag{C0}{\normalsize{$\mathcal{C}_0$}}
   \psfrag{0}{\normalsize{$0$}}
   \psfrag{Pi1}{\normalsize{$\widetilde{\Pi}$}}  

   \includegraphics{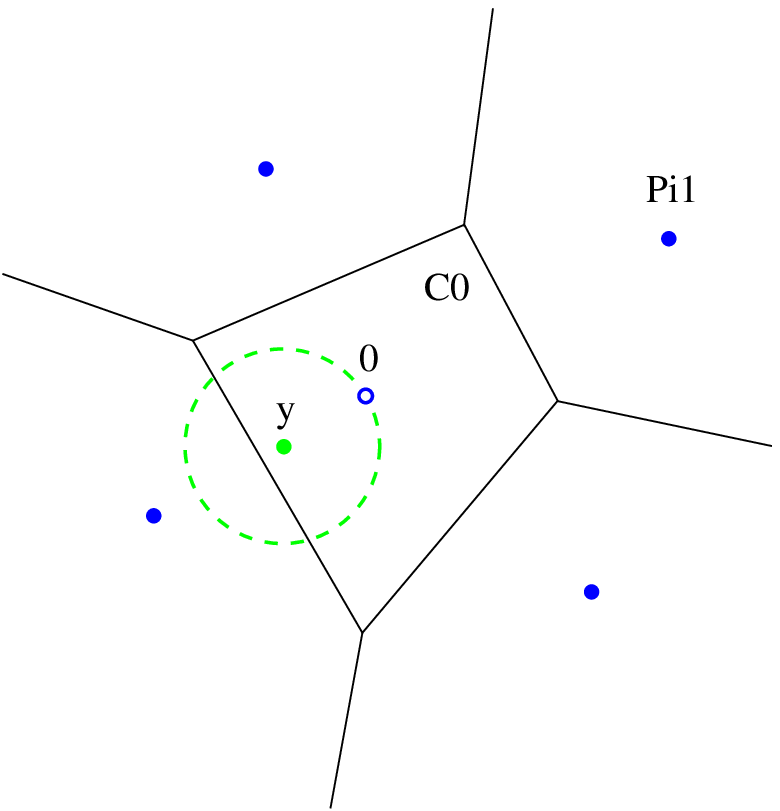}
  }
}\subfigure[]{
  \label{fig:aux-out-isol}
  \scalebox{0.8}{
   \psfrag{x1}{\normalsize{$\breve{x}_{1}$}}
   \psfrag{x2}{\normalsize{$\breve{x}_{2}$}}
   \psfrag{0}{\normalsize{$0$}}

   \includegraphics{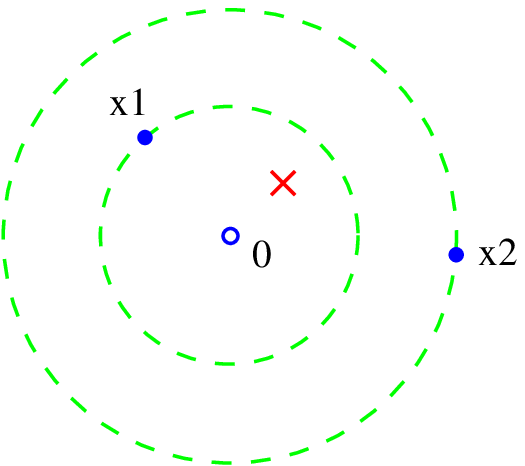}
  }
}\subfigure[]{
  \label{fig:aux-in-isol}
  \scalebox{0.8}{
   \psfrag{x1}{\normalsize{$\breve{x}_{1}$}}
   \psfrag{x2}{\normalsize{$\breve{x}_{2}$}}
   \psfrag{0}{\normalsize{$0$}}

   \includegraphics{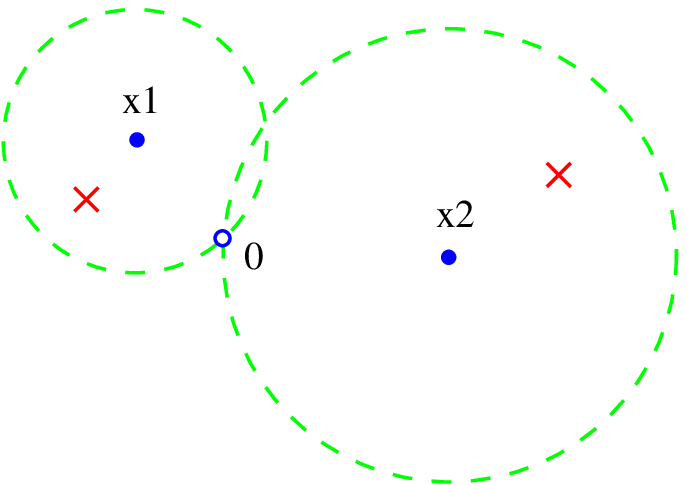}
  }
}\caption{\label{fig:aux-drawings}Auxiliary diagrams.}

\end{figure}

\begin{table}[p]
\begin{centering}
\begin{tabular}{|c|ccccccc|}
\hline 
$n\backslash k$ & 1 & 2 & 3 & 4 & 5 & 6 & 7\tabularnewline
\hline
1 & 1 &  &  &  &  &  & \tabularnewline
2 & 1 & 1 &  &  &  &  & \tabularnewline
3 & 1 & 3 & 1 &  &  &  & \tabularnewline
4 & 1 & 7 & 6 & 1 &  &  & \tabularnewline
5 & 1 & 15 & 25 & 10 & 1 &  & \tabularnewline
6 & 1  & 31  & 90  & 65  & 15  & 1 & \tabularnewline
7 & 1  & 63  & 301  & 350  & 140  & 21  & 1\tabularnewline
\hline
\end{tabular}
\par\end{centering}

\caption{\label{tab:Stirling}Stirling numbers of the second kind.}

\end{table}

\begin{table}[p]
\begin{centering}
\begin{tabular}{|c|cccc|}
\hline 
$k$ & 1 & 2 & 3 & 4\tabularnewline
\hline 
$\mathbb{E}\{\widetilde{A}^{k}\}$ & 1 & 1.280 & 1.993 & 3.650\tabularnewline
\hline
\end{tabular}
\par\end{centering}

\caption{\label{tab:Moments-A}First four moments of the random area $\widetilde{A}$
of a typical Voronoi cell, induced by a unit-density Poisson process~\cite{Bra:86}.}

\end{table}

\begin{figure}[p]
\begin{centering}
\scalebox{0.6}{\psfrag{probability}{\large{\sf{probability}}}
\psfrag{n}{\hspace{-8mm}\large{\sf{degree $n$}}}
\psfrag{PMF Nout}{\Large{$p_{N_{\mathrm{out}}}(n)$}}
\psfrag{PMF Nin}{\Large{$p_{N_{\mathrm{in}}}(n)$}}\includegraphics{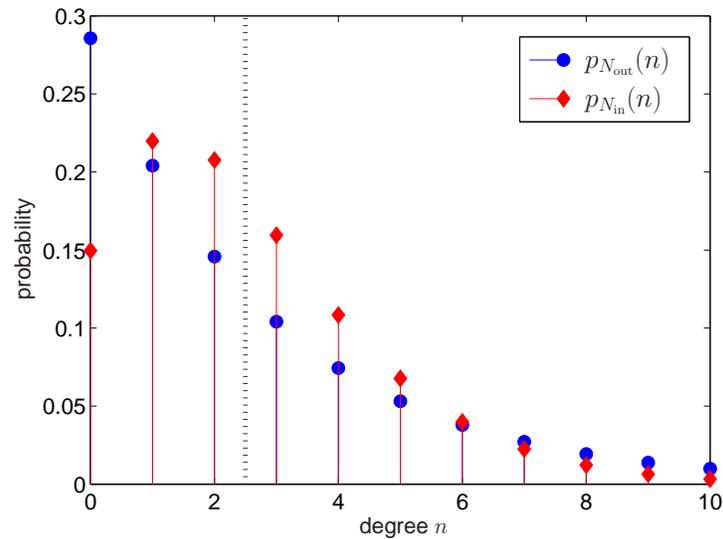}}
\par\end{centering}

\caption{\label{fig:pmf-degrees}PMF of the in- and out-degree of a node ($\frac{\LE}{\LM}=0.4$).
The vertical line marks the average node degrees, $\mathbb{E}\{N_{\mathrm{out}}\}=\mathbb{E}\{N_{\mathrm{in}}\}=\frac{\LM}{\LE}=2.5$,
in accordance with Property~\ref{prop:ENin-ENout}.}

\end{figure}

\begin{figure}
\begin{centering}
\scalebox{0.6}{\psfrag{probability}{\large{\sf{probability}}}
\psfrag{lambda}{\large{$\;\LE/\LM$}}
\psfrag{P-out}{\Large{$p_{\textrm{out-isol}}$}}
\psfrag{P-in}{\Large{$p_{\textrm{in-isol}}$}}\includegraphics{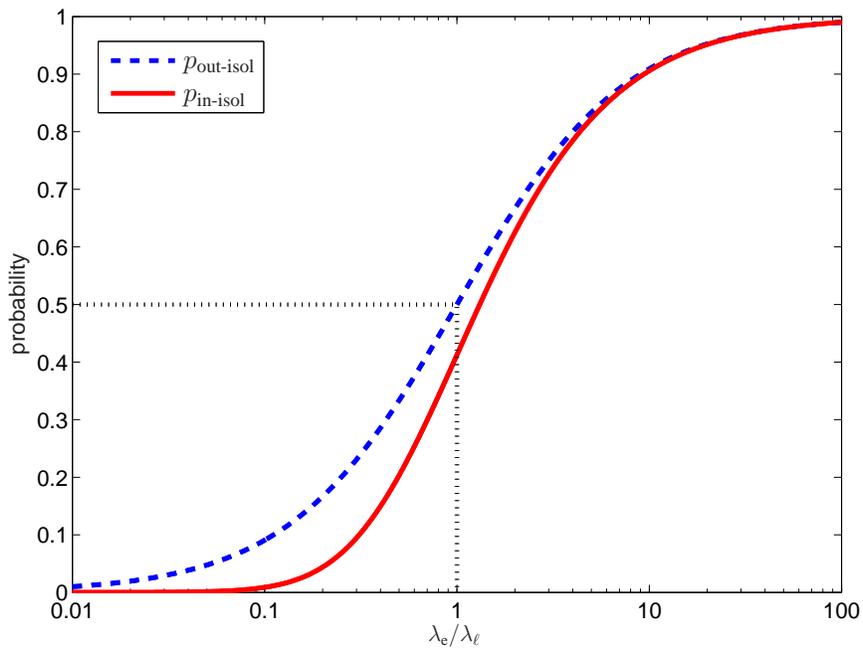}}
\par\end{centering}

\caption{\label{fig:p-isol}Probabilities of in- and out-isolation of a node,
versus the ratio~$\frac{\LE}{\LM}$. Note that $p_{\mathrm{in-isol}}<p_{\mathrm{out-isol}}$
for any fixed~$\frac{\LE}{\LM}$, as proved in Property~\ref{prop:pin-pout-isol}.}

\end{figure}

\begin{figure}
\begin{centering}
\scalebox{1}{\psfrag{R}{\normalsize{$y$}}
\psfrag{RE}{\normalsize{$r$}}
\psfrag{Pi}{\normalsize{$\{\RM{i}\}$}}
\psfrag{PiE}{\normalsize{$\{\RE{i}\}$}}
\psfrag{Nout=3}{\normalsize{$N_\textrm{out}=3$}}
\psfrag{g}{\normalsize{$\left(\frac{\PM}{\WM(2^{\th}-1)}\right)^{1/2b}$}}
\psfrag{f}{\normalsize{$y=\ff(r)$}}\includegraphics{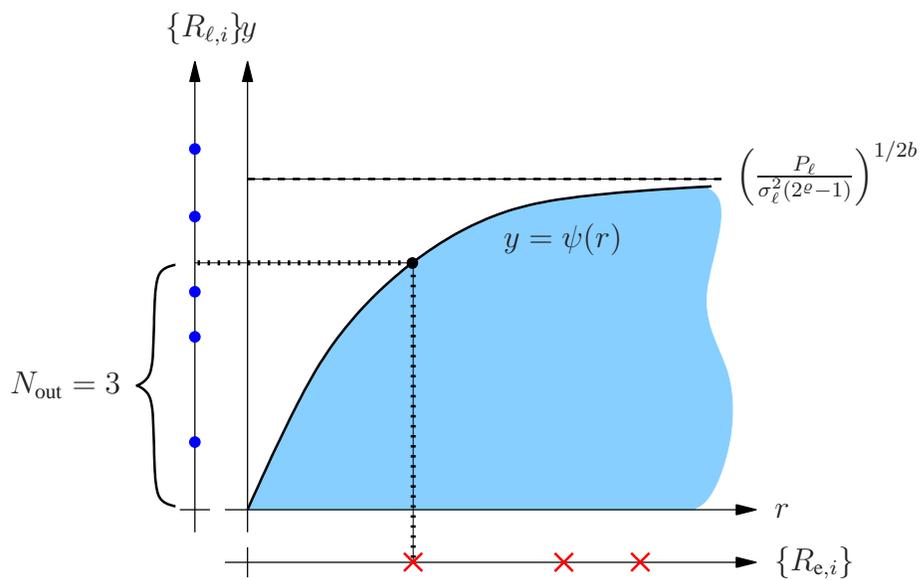}}
\par\end{centering}

\caption{\label{fig:out-degree-theta}The effect of non-zero secrecy rate threshold~$\th$
and unequal noise powers~$\WM,\WE$ on the average node degree, for
the case of $g(r)=\frac{1}{r^{2b}}$. The function~$\ff(r)$ was
defined in (\ref{eq:f-ub}).}

\end{figure}

\begin{figure}
\begin{centering}
\scalebox{0.6}{\psfrag{node degree}{\hspace{-10mm}\large{\sf{average node degree}}}
\psfrag{theta }{\large{\hspace{-3mm}$\th\;\mathsf{(bits)}$}}
\psfrag{PW=0.5}{\Large{$\PM/\W=0.5$}}
\psfrag{PW=5}{\Large{$\PM/\W=5$}}
\psfrag{PW=inf}{\Large{$\PM/\W=\infty$}}
\psfrag{analytical E-Nout}{\large{\sf{analytical $\mathbb{E}\{N_{\mathrm{out}}\}$}}}
\psfrag{analytical upper bound for E-Nout}{\large{\sf{analytical upper bound for $\mathbb{E}\{N_{\mathrm{out}}\}$}}}\includegraphics{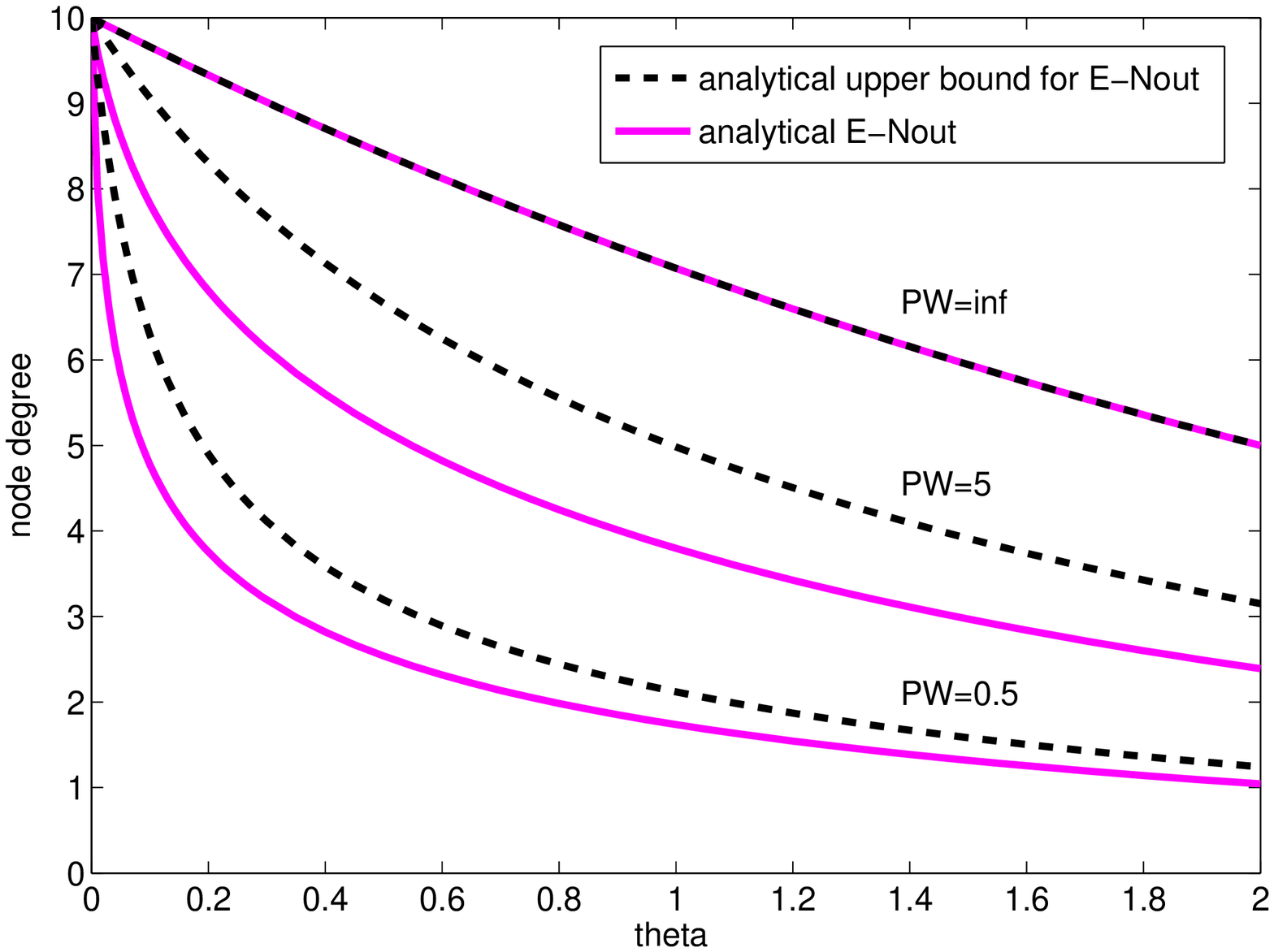}}
\par\end{centering}

\caption{\label{fig:theta-PW-plot}Average node degree versus the secrecy rate
threshold~$\th$, for various values of $\PM/\W$ ($\WM=\WE=\W$,
$g(r)=\frac{1}{r^{2b}}$, $b=2$, $\LM=1\,\textrm{m}^{-2}$, $\LE=0.1\,\textrm{m}^{-2}$). }

\end{figure}

\begin{figure}[p]
\begin{centering}
\scalebox{1}{\psfrag{Nout=5}{\normalsize{$N_\textrm{out}=5$}}
\psfrag{Re}{\normalsize{$\Gamma_\textrm{1}$}}
\psfrag{S1}{\normalsize{$\mathcal{S}^{(1)}$}}
\psfrag{S2}{\normalsize{$\mathcal{S}^{(2)}$}}
\psfrag{S3}{\normalsize{$\mathcal{S}^{(3)}$}}
\psfrag{S4}{\normalsize{$\mathcal{S}^{(4)}$}}
\includegraphics{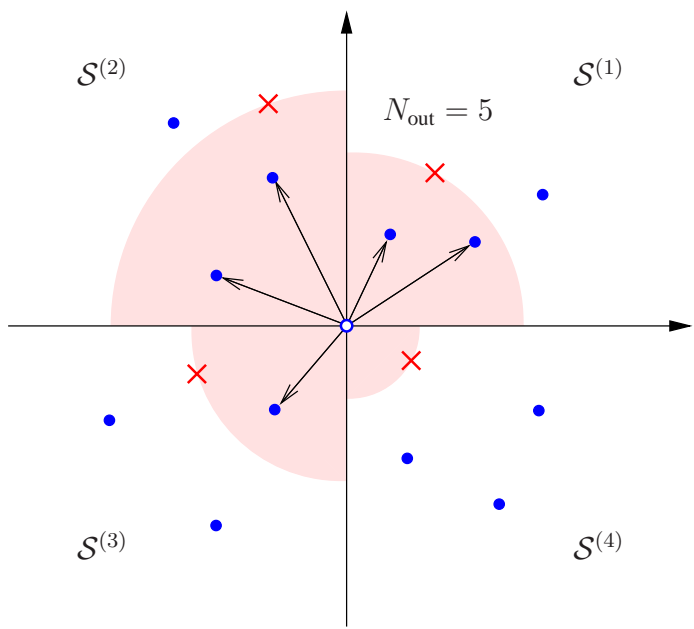}}
\par\end{centering}

\caption{\label{fig:out-degree-sectors}Out-degree of a node with sectorized
transmission. In this example with $L=4$ sectors, the node at the
origin can transmit messages with information-theoretic security to
$N_{\mathrm{out}}=5$ nodes.}

\end{figure}

\begin{figure}[p]
\begin{centering}
\scalebox{1}{\psfrag{Nout=5}{\normalsize{$N_\textrm{out}=5$}}
\psfrag{G1}{\hspace{-2mm}\normalsize{$\RE{1}$}}
\psfrag{p}{\normalsize{$\rho$}}
\psfrag{R}{\normalsize{$\An(0,\RE{1})\cap\overline{\Theta}$}}\includegraphics{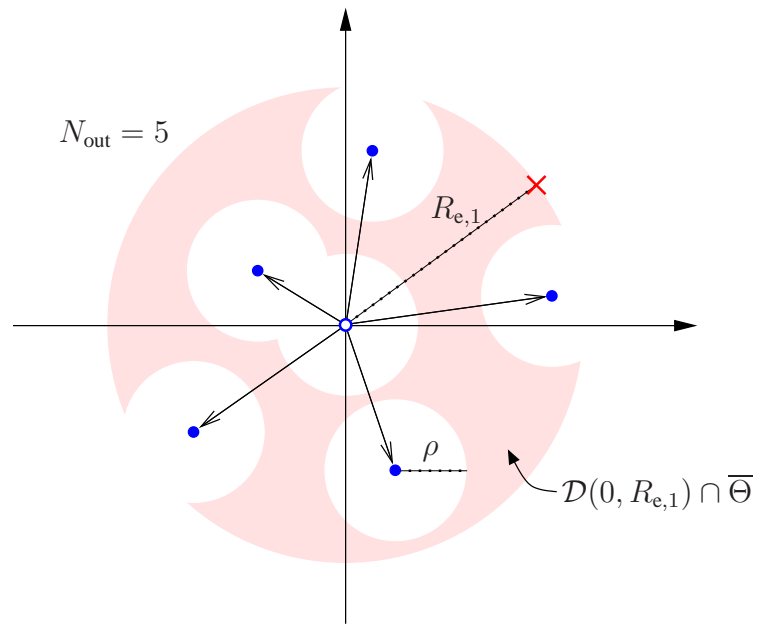}}
\par\end{centering}

\caption{\label{fig:out-degree-exclusion}Out-degree of a node with eavesdropper
neutralization. In this example, the node at the origin can transmit
messages with information-theoretic security to $N_{\mathrm{out}}=5$
nodes.}

\end{figure}

\begin{figure}
\begin{centering}
\scalebox{0.6}{\psfrag{node degree}{\hspace{-10mm}\large{\sf{average node degree}}}
\psfrag{rc (m)}{\large{$\rho$ ($\mathsf{m}$)}}
\psfrag{laE=0.1}{\Large{$\mathsf{\LE=0.1\, m^{-2}}$}}
\psfrag{laE=0.2}{\Large{$\mathsf{\LE=0.2\, m^{-2}}$}}
\psfrag{laE=0.5}{\Large{$\mathsf{\LE=0.5\, m^{-2}}$}}
\psfrag{simulated E-Nout}{\large{\sf{simulated $\mathbb{E}\{N_{\mathrm{out}}\}$}}}
\psfrag{analytical lower bound for E-Nout}{\large{\sf{analytical lower bound for $\mathbb{E}\{N_{\mathrm{out}}\}$}}}\includegraphics{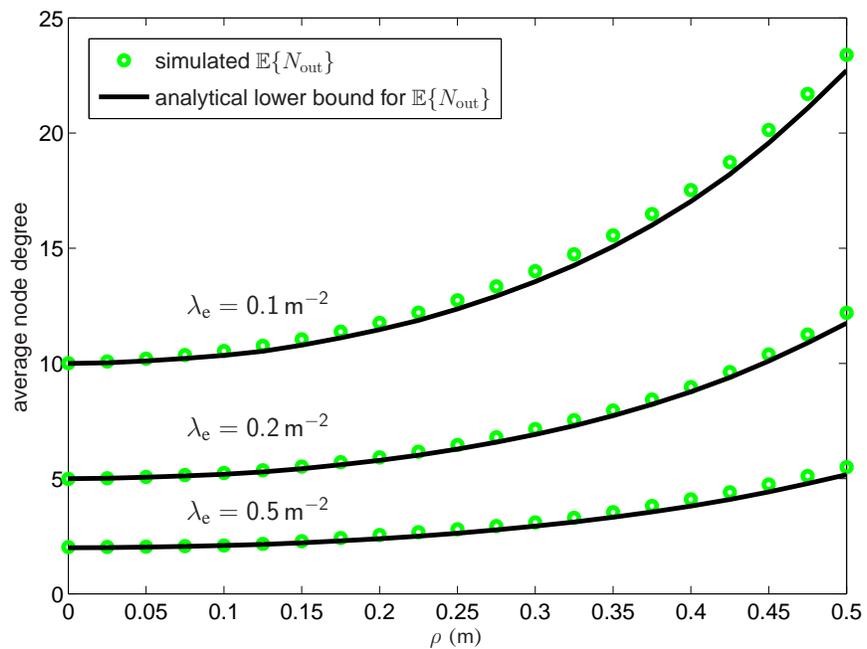}}
\par\end{centering}

\caption{\label{fig:exclusion-lambdaE-plot}Average node degree versus the
neutralization radius~$\rho$, for various values of $\LE$ ($\LM=1\,\textrm{m}^{-2}$).}

\end{figure}

\clearpage

\begin{figure}
\begin{centering}
\scalebox{0.75}{\psfrag{pexist}{\Large{$p_{\mathrm{exist},i}$}}
\psfrag{laE}{\hspace{-5mm}\large{$\LE$ ($\mathsf{m^{-2}}$)}}
\psfrag{i=1,2,4,6}{\large{$i=\mathsf{1,2,4,6}$}}\includegraphics{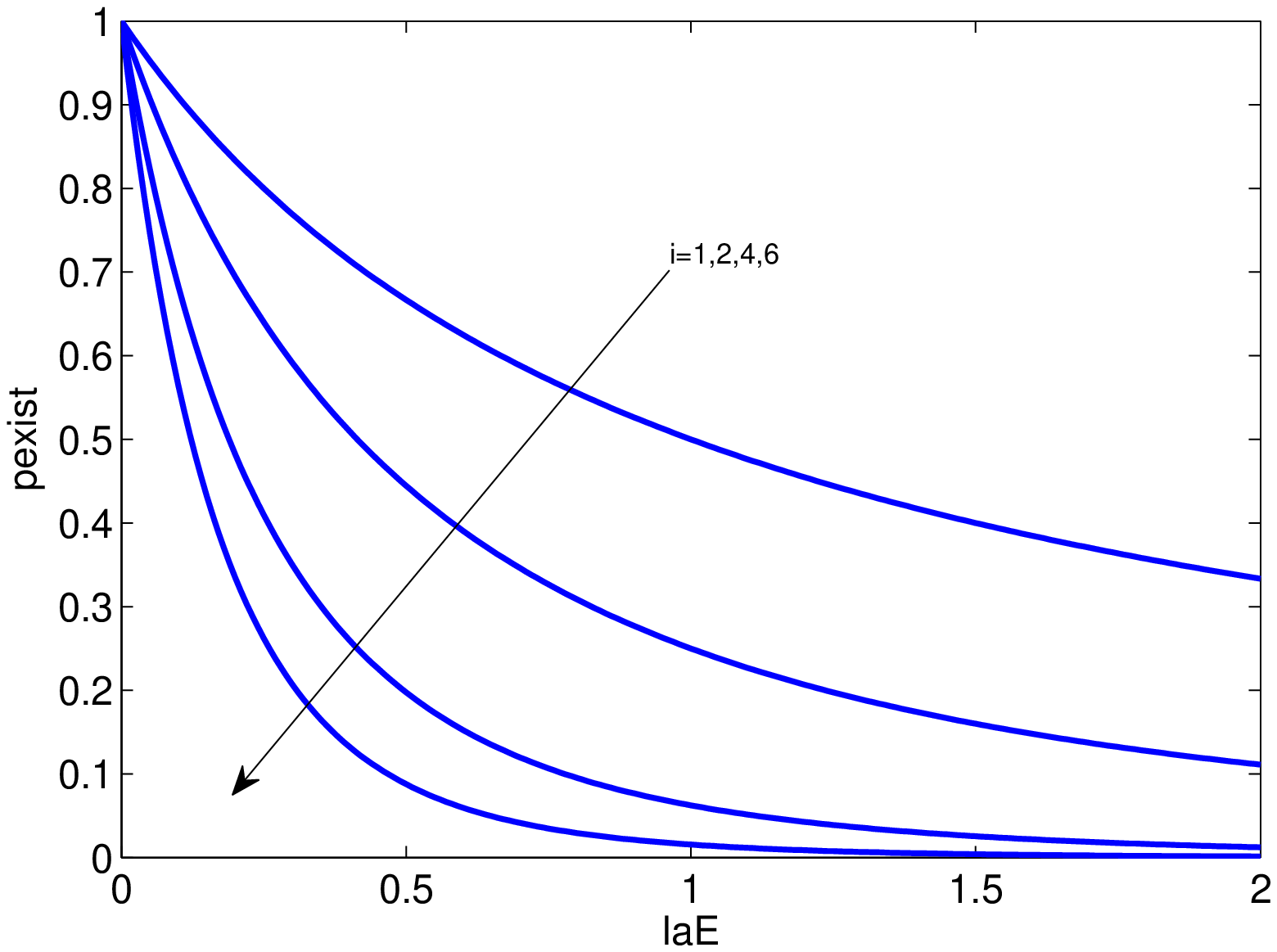}}
\par\end{centering}

\caption{\label{fig:pexist-lambda}Probability~$p_{\mathrm{exist},i}$ of
existence of a non-zero MSR versus the eavesdropper density~$\LE$,
for various values of the neighbour index~$i$ ($\LM=1\,\textrm{m}^{-2}$,
$b=2$, $\PM/\W=10$, $\th=1\,\textrm{bit}$).}

\end{figure}

\begin{figure}
\begin{centering}
\scalebox{0.75}{\psfrag{poutage(Rs)}{\Large{$p_{\mathrm{outage},i}(\th)$}}
\psfrag{Rs (bits/s)}{\hspace{-15mm}\large{\sf{secrecy rate $\th$ (bits)}}}
\psfrag{i=1,2,4,6}{\large{$i=\mathsf{1,2,4,6}$}}\includegraphics{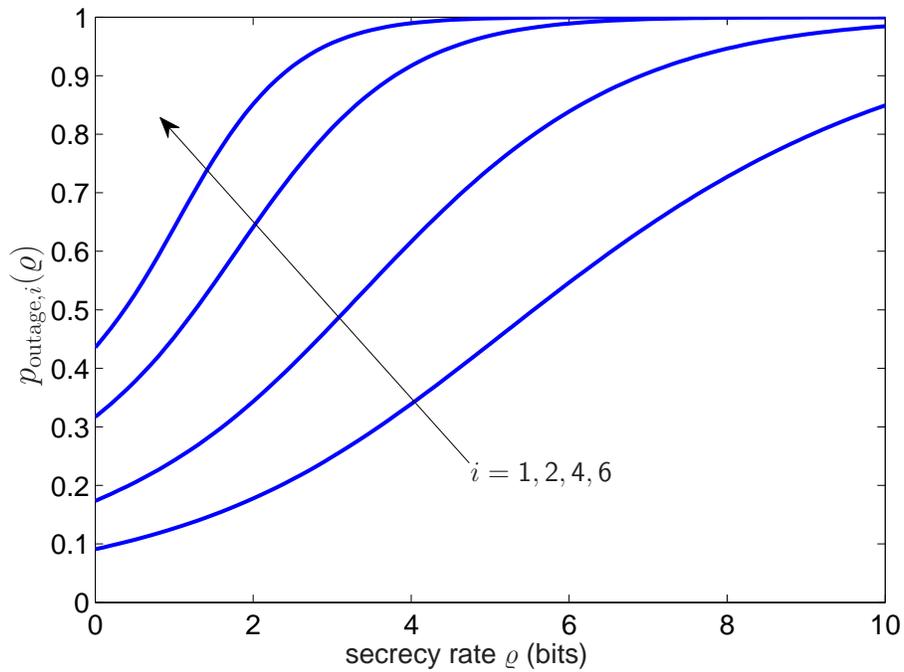}}
\par\end{centering}

\caption{\label{fig:cdf-csi}Probability~$p_{\mathrm{outage},i}$ of secrecy
outage between a node and its $i$-th closest neighbour, for various
values of the neighbour index~$i$ ($\LM=1\,\textrm{m}^{-2}$, $\LE=0.1\,\textrm{m}^{-2}$,
$b=2$, $\PM/\W=10$).}

\end{figure}

\begin{figure}
\begin{centering}
\scalebox{1}{\psfrag{Alice}{\footnotesize{\sf{Alice (probe transmitter)}}}
\psfrag{Bob}{\footnotesize{\sf{Bob (probe receiver)}}}
\psfrag{Colluding eavesdroppers}{\footnotesize{\sf{Colluding eavesdroppers}}}
\psfrag{r0}{\hspace{1mm}\normalsize{$\rM$}}
\psfrag{R1}{\normalsize{$\RE{1}$}}
\psfrag{R2}{\hspace{-2mm}\normalsize{$\RE{2}$}}
\psfrag{R3}{\hspace{-2mm}\normalsize{$\RE{3}$}}\includegraphics{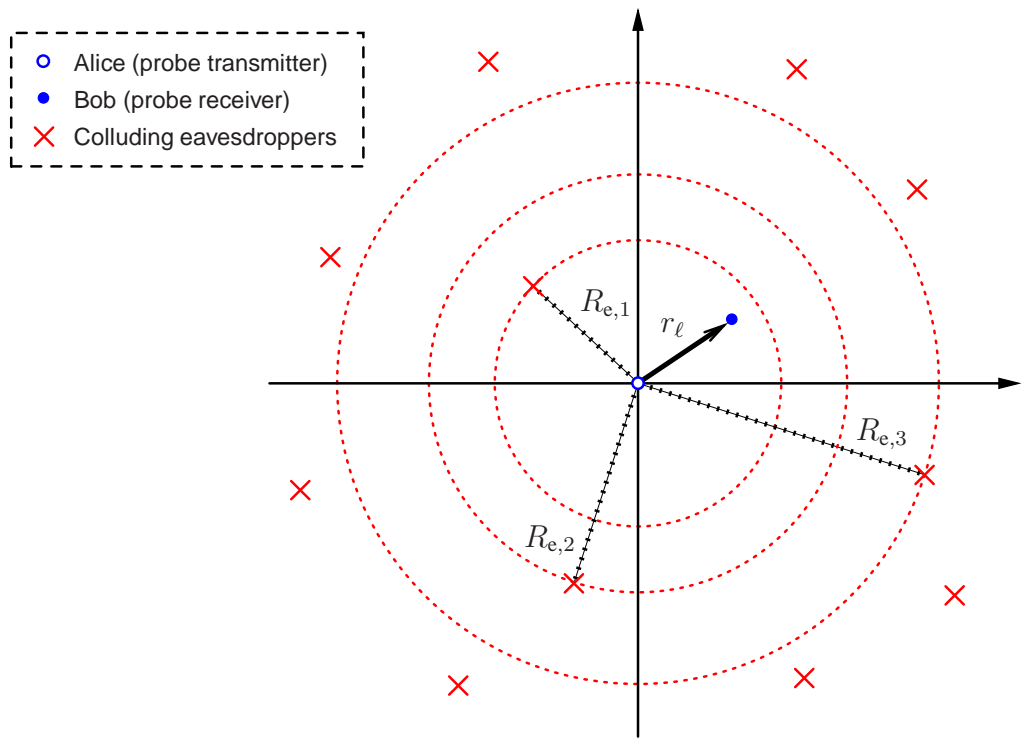}}
\par\end{centering}

\caption{\label{fig:collusion}Communication in the presence of colluding eavesdroppers.}

\end{figure}

\begin{figure}
\begin{centering}
\scalebox{1}{\psfrag{Alice}{\normalsize{\sf{Alice}}}
\psfrag{Bob}{\normalsize{\sf{Bob}}}
\psfrag{Eve}{\normalsize{\sf{Eve}}}
\psfrag{main channel}{\hspace{-5mm}\small{\sf{Legitimate channel}}}
\psfrag{wiretap channel}{\hspace{-5mm}\small{\sf{Eavesdropper channel}}}
\psfrag{xM}{\large{$x$}}
\psfrag{hM}{\large{$\mathbf{h}_{\M}$}}
\psfrag{hE}{\large{$\mathbf{h}_{\E}$}}
\psfrag{wM}{\large{$\mathbf{w}_{\M}$}}
\psfrag{wE}{\large{$\mathbf{w}_{\E}$}}
\psfrag{yM}{\large{$\mathbf{y}_{\M}$}}
\psfrag{yE}{\large{$\mathbf{y}_{\E}$}}\includegraphics{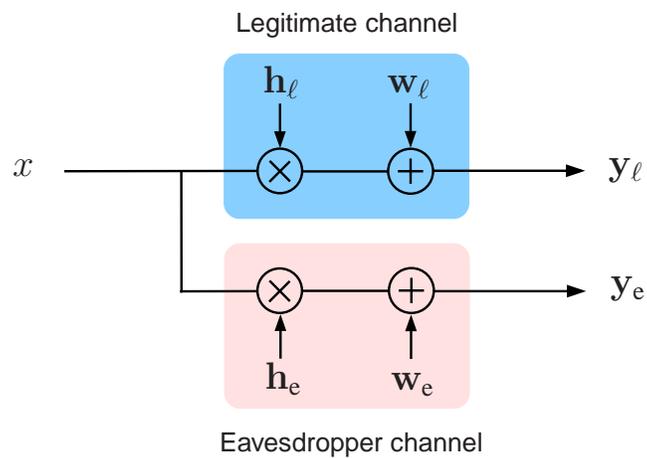}}
\par\end{centering}

\caption{\label{fig:SIMO}SIMO Gaussian wiretap channel, which can be used
to analyze the scenario of colluding eavesdroppers depicted in Fig.~\ref{fig:collusion}.}

\end{figure}

\begin{table}[p]
\noindent \begin{centering}
\begin{tabular}{|c|c|}
\hline 
Non-colluding & Colluding\tabularnewline
\hline 
$\PE=\frac{P_{\mathrm{l}}}{\RE 1^{2b}}$ & $\PE=\sum_{i=1}^{\infty}\frac{\PM}{\RE i^{2b}}$\tabularnewline
\hline 
$f_{\PE}(x)=\frac{\pi\LE}{bx}\left(\frac{\PM}{x}\right)^{1/b}\exp\left(-\pi\LE\left(\frac{\PM}{x}\right)^{1/b}\right),x\geq0$ & $\PE\sim\mathcal{S}\left(\alpha=\frac{1}{b},\:\beta=1,\:\gamma=\pi\LE\mathcal{C}_{1/b}^{-1}\PM^{1/b}\right)$\tabularnewline
\hline 
$F_{\Cs}(c)=1-\exp\left(-\pi\LE\left(\frac{\frac{\PM}{\WE}}{\left(1+\frac{\PM}{\rM^{2b}\WM}\right)2^{-\th}-1}\right)^{1/b}\right),0\leq\th<\CM$ & $F_{\Cs}(c)=1-F_{\PET}\left(\frac{\left(1+\frac{\PM}{\rM^{2b}\WM}\right)2^{-\th}-1}{(\pi\LE\mathcal{C}_{1/b}^{-1})^{b}\frac{\PM}{\WE}}\right),0\leq\th<\CM$\tabularnewline
 & with $\PET\sim\mathcal{S}\left(\alpha=\frac{1}{b},\:\beta=1,\:\gamma=1\right)$\tabularnewline
\hline
$p_{\mathrm{exist}}=\exp\left(-\pi\LE\rM^{2}\left(\frac{\WM}{\WE}\right)^{1/b}\right)$ & $p_{\mathrm{exist}}=F_{\PET}\left(\frac{\WE}{(\pi\LE\rM^{2}\mathcal{C}_{1/b}^{-1})^{b}\WM}\right)$\tabularnewline
\hline
$\mathbb{E}\{N_{\mathrm{out}}\}=\frac{\LM}{\LE}$ & $\mathbb{E}\{N_{\mathrm{out}}\}=\frac{\LM}{\LE}\,\textrm{sinc}\!\left(\frac{1}{b}\right)$\tabularnewline
\hline
\end{tabular}
\par\end{centering}

\caption{\label{tab:c-vs-nc}Comparison between the cases of non-colluding
and colluding eavesdroppers, considering a single legitimate link,
and a channel gain of the form~$g(r)=\frac{1}{r^{2b}}$.}

\end{table}

\begin{figure}
\begin{centering}
\scalebox{0.6}{\psfrag{pdf}{\Large{\hspace{-5mm}$f_{\PE/\PM}(x)$}}
\psfrag{x}{\hspace{-20mm}\large{\sf{normalized power $x$}}}

\psfrag{Colluding}{\large{\sf{Colluding}}}
\psfrag{Non-colluding}{\large{\sf{Non-colluding}}}\includegraphics{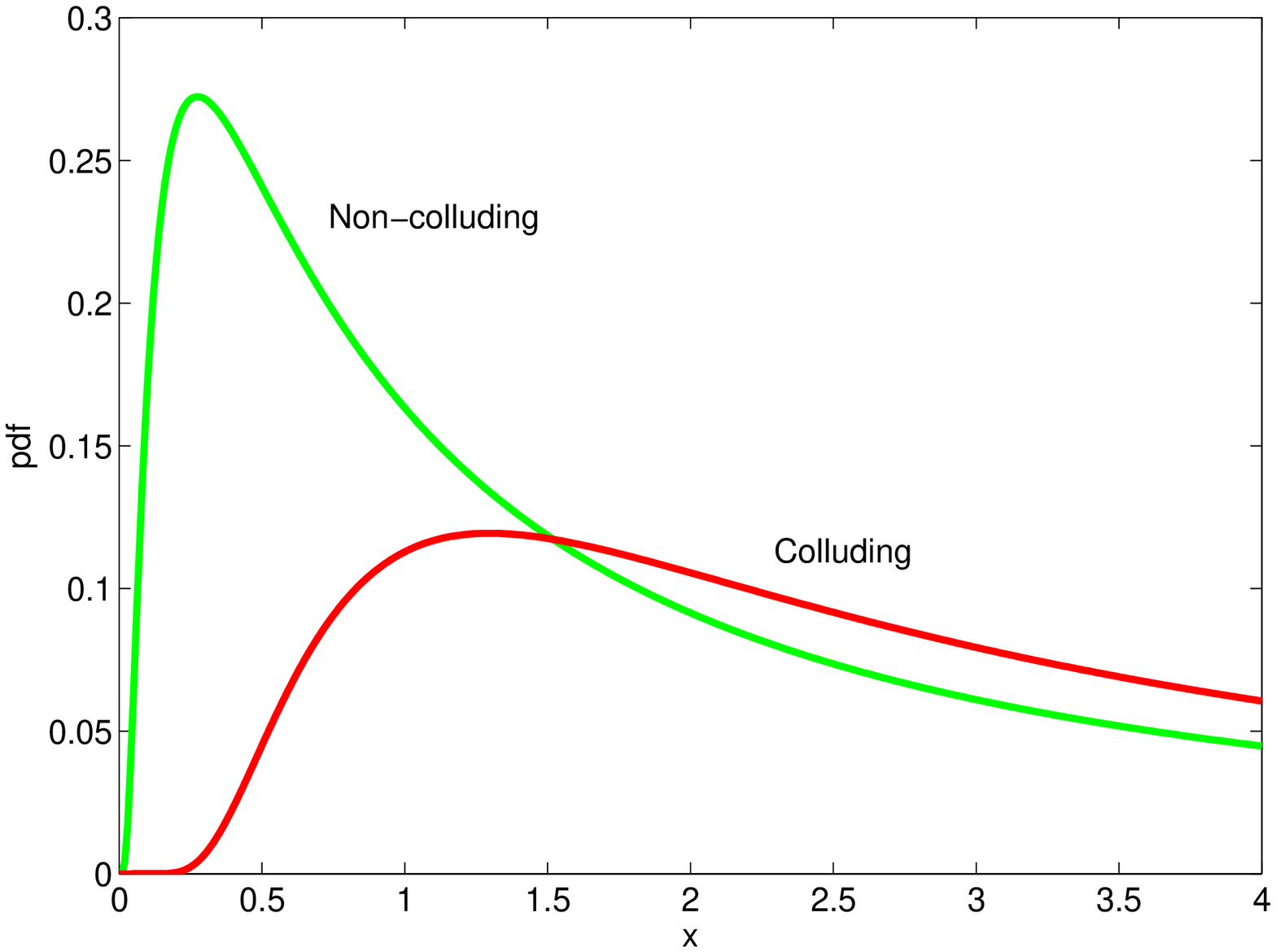}}
\par\end{centering}

\caption{\label{fig:pdf-pe-c-nc}PDF~$f_{\PE/\PM}(x)$ of the (normalized)
received eavesdropper power~$\PE/\PM$, for the cases of colluding
and non-colluding eavesdroppers ($b=2$, $\LE=0.5\,\textrm{m}^{-2}$).}

\end{figure}

\begin{figure}
\begin{centering}
\scalebox{0.6}{\psfrag{pexist}{\Large{$p_{\mathrm{exist}}$}}
\psfrag{lambda}{\large{$\LE$ ($\mathsf{m^{-2}}$)}}

\psfrag{Non-colluding,r0=1}{\large{\sf{Non-colluding, $\rM=\mathsf{1\, m}$}}}
\psfrag{Non-colluding,r0=2}{\large{\sf{Non-colluding, $\rM=\mathsf{2\, m}$}}}
\psfrag{Colluding,r0=1}{\large{\sf{Colluding, $\rM=\mathsf{1\, m}$}}}
\psfrag{Colluding,r0=2}{\large{\sf{Colluding, $\rM=\mathsf{2\, m}$}}}\includegraphics{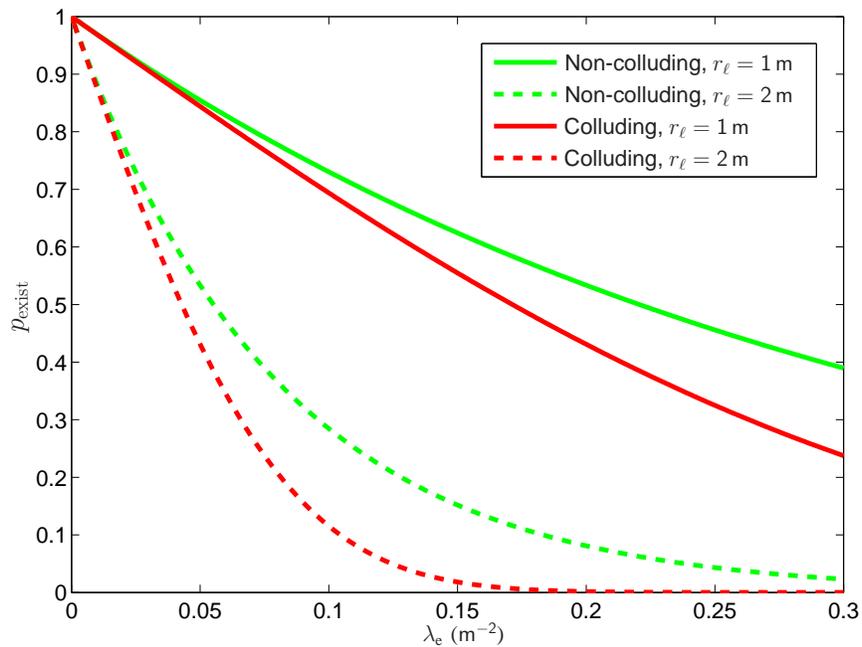}}
\par\end{centering}

\caption{\label{fig:pexist-lambda-collusion}Probability~$p_{\mathrm{exist}}$
of existence of a non-zero MSR versus the eavesdropper density~$\LE$,
for the cases of colluding and non-colluding eavesdroppers, and various
values of $\rM$ ($b=2$).}

\end{figure}

\begin{figure}
\begin{centering}
\scalebox{0.55}{\psfrag{Fc(c)}{\Large{$p_{\mathrm{outage}}(\th)$}}
\psfrag{c (bits/complex dimension)}{\hspace{3mm}\large{\sf{secrecy rate $\th$ (bits)}}}

\psfrag{Non-colluding,lambda=0.1}{\large{\sf{Non-colluding, $\mathsf{\LE=0.1\, m^{-2}}$}}}
\psfrag{Non-colluding,lambda=0.2}{\large{\sf{Non-colluding, $\mathsf{\LE=0.2\, m^{-2}}$}}}
\psfrag{Colluding,lambda=0.1}{\large{\sf{Colluding, $\mathsf{\LE=0.1\, m^{-2}}$}}}
\psfrag{Colluding,lambda=0.2}{\large{\sf{Colluding, $\mathsf{\LE=0.2\, m^{-2}}$}}}\includegraphics{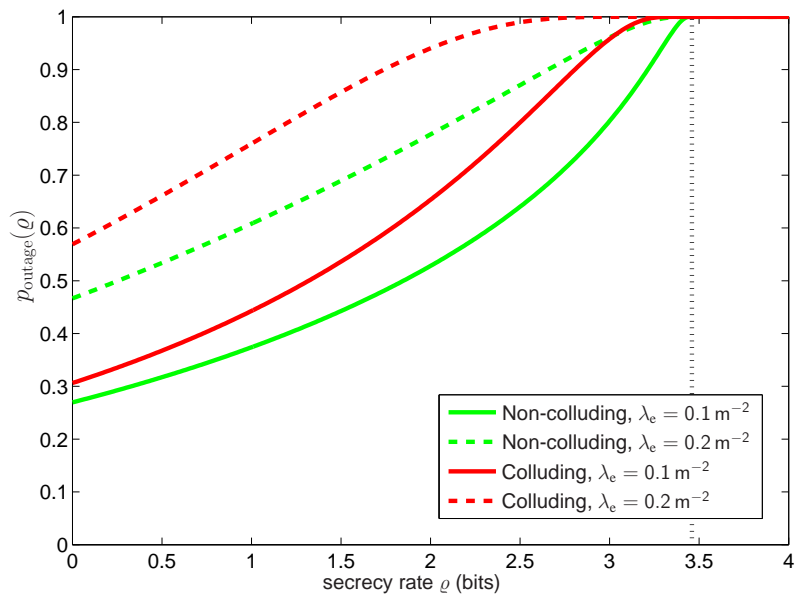}}
\par\end{centering}

\caption{\label{fig:cdf-cs-collusion}Probability~$p_{\mathrm{outage}}$ of
secrecy outage for the cases of colluding and non-colluding eavesdroppers,
and various densities~$\LE$ of eavesdroppers ($b=2$, $\PM/\W=10$,
$\rM=1\,\textrm{m}$). The vertical line marks the capacity of the
legitimate link, which for these system parameters is $\CM=3.46$
bits/complex dimension. }

\end{figure}

\begin{figure}
\begin{centering}
\scalebox{0.6}{\psfrag{b}{\large{\hspace{-0mm}$b$}}
\psfrag{average node degree}{\hspace{-13mm}\large{\sf{normalized average node degree}}}

\psfrag{Colluding}{\large{\sf{Colluding}}}
\psfrag{Non-colluding}{\large{\sf{Non-colluding}}}\includegraphics{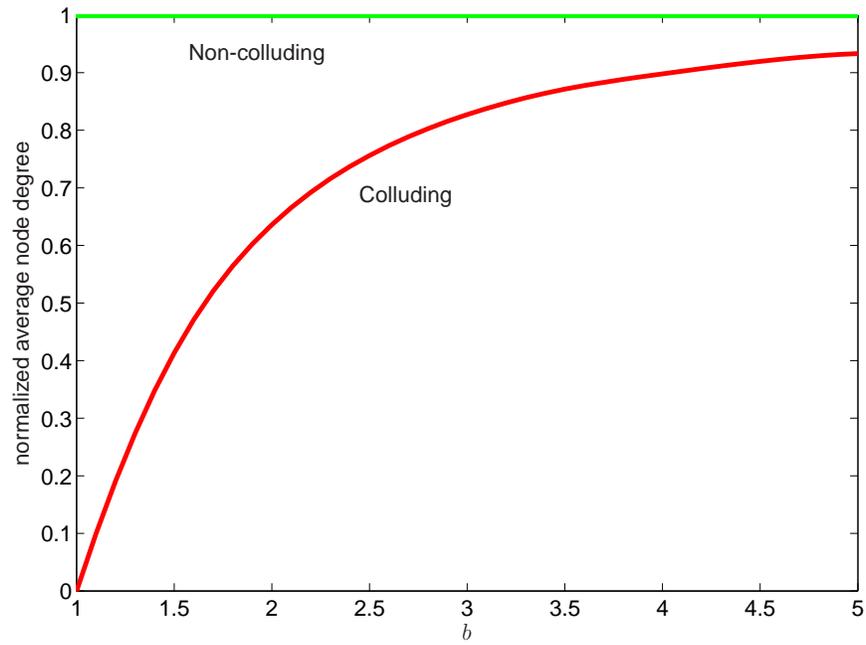}}
\par\end{centering}

\caption{\label{fig:out-degree-collusion}Normalized average node degree of
the $\is$graph, $\frac{\mathbb{E}\{N_{\mathrm{out}}\}}{\LM/\LE}$,
versus the amplitude loss exponent~$b$, for the cases of colluding
and non-colluding eavesdroppers.}

\end{figure}

\end{document}